\documentclass{theoretics}

\title{MSO-Enumeration Over SLP-Compressed Unranked Forests}
\ThCSauthor[Siegen]{Markus Lohrey}{lohrey@eti.uni-siegen.de}[https://orcid.org/0000-0002-4680-7198]
\ThCSaffil[Siegen]{Universit\"at Siegen, H\"olderlinstr.~3, 57076, Siegen, Germany}

\ThCSauthor[Humboldt]{Markus L.\ Schmid}{MLSchmid@MLSchmid.de}[https://orcid.org/0000-0001-5137-1504]
\ThCSaffil[Humboldt]{Humboldt-Universit\"at zu Berlin, Unter den Linden 6, D-10099, Berlin, Germany}

\ThCSthanks{%
    This article contains and extends the material of the conference contribution~\cite{LohreySchmid24}. %
    The second author is supported by the German Research Foundation (Deutsche Forschungsgemeinschaft, DFG) – project number
    522576760 (gef\"{o}rdert durch die Deutsche Forschungsgemeinschaft (DFG) – Projektnummer 522576760).%
    }
\ThCSshortnames{M. Lohrey, M. L. Schmid}  
\ThCSshorttitle{MSO-Enumeration Over SLP-Compressed Unranked Forests}
\ThCSyear{2026}
\ThCSarticlenum{6}
\ThCSreceived{Aug 28, 2025}
\ThCSrevised{Dec 16, 2026}
\ThCSaccepted{Jan 24, 2026}
\ThCSpublished{Feb 26, 2026}
\ThCSdoicreatedtrue
\ThCSkeywords{Query evaluation over SLP-compressed data, MSO Logic, Enumeration algorithms, Algorithmics over compressed inputs}

\usepackage{amsfonts}
\usepackage{url}
\usepackage{tikz,tikz-qtree}
\usetikzlibrary{backgrounds,calc,positioning}
\usepackage{forest}
\usepackage{stmaryrd}

\colorlet{lightgray}{white!80!black}

\tikzset{every tree node/.style={draw, circle, black, semithick, inner sep = 2pt, minimum size = 10pt}}
\tikzstyle{arrow} = [semithick,->,>=stealth]
\tikzset{level distance = 40pt, sibling distance = 20pt}
\tikzset{edge from parent/.append style={arrow, edge from parent path = {(\tikzparentnode) -- (\tikzchildnode)}}}
\tikzset{gate/.style={draw, circle, black, semithick, inner sep = 0, minimum size = 5mm}}
\tikzset{alpha/.style={inner sep = 1.2pt}}

\newcommand{\multiset}[1]{\{\!\{#1\}\!\}}

\newcommand{\order}[1]{\mathsf{h}(#1)}

\newcommand{\F}{\mathsf{F}}

\newcommand{\Conf}{\mathsf{Conf}}
\newcommand{\EXP}{\mathcal{E}}
\newcommand{\expr}{\phi}
\newcommand{\DAG}{\mathcal{D}}
\DeclareMathOperator{\SLP}{\mathcal{S}}
\DeclareMathOperator{\FSLP}{\mathcal{F}}
\newcommand{\leaves}{\mathsf{leaves}}
\newcommand{\unfold}{\mathsf{unfold}}
\newcommand{\po}{\mathsf{po}}
\newcommand{\suc}{\mathsf{succ}}
\newcommand{\sucDAG}{\mathsf{succ}_{\DAG}}

\newcommand{\sucFSLP}{\mathsf{succ}_{\FSLP}}

\renewcommand{\path}{\mathsf{path}}
\newcommand{\select}{\mathsf{select}}
\renewcommand{\alph}{\mathsf{alphabet}}
\newcommand{\conc}{\mathbin{\varodot}}

\newcommand{\concv}{\mathbin{\varobar}}
\newcommand{\conch}{\mathbin{\varominus}}
\newcommand{\valX}[1]{\llbracket #1 \rrbracket}

\newcommand{\pre}{\mathsf{pod}}
\DeclareMathOperator{\eword}{\varepsilon}
\newcommand{\card}[1]{\lvert #1 \rvert}

\DeclareMathOperator{\bigO}{\mathcal{O}}

\newcommand{\deriv}[1]{\valX{#1}}

\newcommand{\newNodeMarker}[1]{\tilde{#1}}
\newcommand{\derivsub}[2]{\valX{#2}_{#1}}
\newcommand{\ordersub}[2]{\mathsf{h}_{#1}(#2)}

\ThCSnewtheoita{assumption}

\begin{document}

\maketitle

\begin{abstract}
We study the problem of enumerating the answers to a query formulated in monadic second order logic (MSO) over an unranked forest $F$ that is compressed by a straight-line program (SLP) $\DAG$. Our main result states that this can be done after $\bigO(|\DAG|)$ preprocessing and with output-linear delay (in data complexity). This is a substantial improvement over the previously known algorithms for MSO-evaluation over trees, since the compressed size $|\DAG|$ might be much smaller than (or even logarithmic in) the actual data size $|F|$, and there are linear time SLP-compressors that yield very good compressions on practical inputs. In particular, this also constitutes a meta-theorem in the field of algorithmics on SLP-compressed inputs: all enumeration problems on trees or strings that can be formulated in MSO-logic can be solved with linear preprocessing and output-linear delay, even if the inputs are compressed by SLPs. We also show that our approach can support vertex relabelling updates in time that is logarithmic in the uncompressed data. Our result extends previous work on the enumeration of MSO-queries over uncompressed trees and on the enumeration of document spanners over compressed text documents. 
\end{abstract}

\section{Introduction}

The evaluation of queries formulated in \emph{monadic second order logic} (MSO) is a classical problem in database theory and finite model theory. If we consider unrestricted relational structures as the queried data, then even fixed MSO-queries can express NP-hard problems. On the other hand, a seminal result by Doner~\cite{Doner1970} and Thatcher and Wright~\cite{ThatcherWright1968} shows that if the input data is given as a vertex labelled binary tree, then MSO-model checking (for a fixed formula) can be done in linear time. Moreover, the same holds if the data is given by structures of bounded treewidth, which is Courcelle's famous meta-theorem~\cite{Courcelle1990}. 

However, these positive algorithmic results are formulated for \emph{Boolean} MSO-queries, which does not cover practical scenarios, where we wish to compute all answers to a query. Consequently, in the field of database theory, evaluation problems are nowadays mostly investigated in the context of enumeration algorithms that, after some preprocessing on the data and the query, enumerate all answers to the query without duplicates. In terms of running times, we measure the time for the preprocessing and the delay of the enumeration phase, i.e., the time needed to go from one answer to the next. The optimum is \emph{linear preprocessing} and a delay that is always linear in the size of the next answer that is produced, which is called \emph{output-linear delay}. As usual in database theory, we measure in data complexity, which means that the query is considered to be of constant size. 

Enumeration algorithms with linear preprocessing and output-linear delay are known for several query evaluation settings, including MSO-evaluation, e.g., for MSO-queries on trees and structures with bounded treewidth~\cite{AmarilliBMN19,Bagan06,Courcelle2009,KazanaS13} and for regular document spanners (which is a subclass of MSO-queries) on strings~\cite{AmarilliEtAl2021,FlorenzanoEtAl2020}. Moreover, for these cases the dynamic setting has been investigated as well, where we wish to update our data and then directly enumerate our query again without having to repeat the whole preprocessing (see~\cite{AmarilliBMN19,BerkholzEtAl2018,MMN22,NiewerthSegoufin2018}). 

In this work, we extend MSO-query evaluation on trees towards algorithmics on compressed data, where the input data is given in a \emph{compressed} form and we wish to evaluate the MSO-query without decompressing our data. 

\subsection{Algorithmics on compressed data}\label{sec:ACD}

The paradigm of \emph{algorithmics on compressed data} (ACD) aims to solve fundamental computational tasks directly on compressed data objects, without prior decompression. This allows us to work in a completely compressed setting, where our data is always stored and processed in a compressed form. ACD works very well with respect to \emph{grammar-based compression} using so-called \emph{straight-line programs} (SLPs). Such SLPs use grammar-like formalisms in order to specify how to construct the data object from small building blocks. If the data is given by a finite string $w$, then an SLP is just a context-free grammar for the language $\{w\}$, which can be seen as a sequence of instructions that construct $w$ from the terminal symbols. For instance, the SLP $S \to AA$, $A \to B B C$, $B \to ba$, $C \to cb$ (where $S, A, B,C$ are nonterminals and $a,b,c$ are terminals) produces the string $babacbbabacb$. String SLPs (s-SLPs for short) are very popular and many results exist that demonstrate their wide-range applicability (see, e.g.,~\cite{BannaiEtAl2021,CaselEtAl2021,Ga21,GanardiGawrychowski2022,GanardiJL21} for some recent publications and the survey~\cite{Loh12survey}). Moreover, SLPs achieve very good compression rates in practice (exponential in the best case)
and are tightly related to dictionary based compression, in particular LZ77 and LZ78 \cite{CharikarLLPPSS05,Ryt03}.

An important point is that the ACD paradigm may lead to substantial running time improvements over the uncompressed setting. Indeed, the algorithm's running time only depends on the size of the compressed input, so the smaller size of the input may directly translate into a lower running time. For example, if the same problem can be solved in linear time both in the uncompressed and in the compressed setting, then in the case that the input can be compressed from size $n$ to size $\bigO(\log(n))$ (this is possible with SLPs in the best case), the algorithm in the compressed setting is exponentially faster. This is not just hypothetically speaking. In the field of string algorithms several fundamental problems are known to show this behaviour. String pattern matching is a prominent example for this~\cite{GanardiGawrychowski2022}. 

Recently, the ACD paradigm has been combined with the enumeration perspective of query evaluation. In~\cite{MunozRiveros2025,SchmidS21,SchmidSchweikardt2022,SchmidS22}, the information extraction framework of document spanners is investigated in the compressed setting, and it has been shown that the results of regular spanners over SLP-compressed text documents can be enumerated with linear preprocessing and constant delay.\footnote{Since document spanners output span-tuples of constant size, output-linear delay is the same as constant delay in this query evaluation setting.}
Applying SLP-based ACD in the framework of document spanners suggests itself, since this is essentially a query model for string data (or sequences), and ACD is most famous in the realm of string algorithms.

An advantage of the grammar-based compression approach is that it can be easily extended to trees. More precisely, by using context-free tree grammars, s-SLPs can be extended to SLPs for ranked trees~\cite{Lohrey15,LohreyMR18,LohreyEtAl2012}. In this paper, we use so-called \emph{forest SLPs} (f-SLPs for short) that were defined in~\cite{GasconLMRS20}. Forest SLPs allow to compress \emph{node-labelled, unranked and ordered} forests, i.e., ordered lists of trees, where every node has an ordered list of children of arbitrary length. A string is the special case where all trees have size one. An f-SLP can compress such a forest in the horizontal as well as the vertical dimension. Using the horizontal dimension, one can, for instance, compress the forest $a a \cdots a$ ($n$ many  $a$-labelled roots without children) by an f-SLP of size $\bigO(\log n)$, whereas compression in the vertical dimension allows to represent the tree $a(a(\cdots a(a) \cdots ))$ (a vertical chain of $n$ $a$-labelled nodes)  by an f-SLP of size $\bigO(\log n)$ (note that such vertical compression cannot be achieved by simple DAG-compression, i.e., folding a tree into its natural DAG-representation (see also \autoref{sec:relatedWork})). Formally, f-SLPs  are based on the formalism of \emph{forest algebras}~\cite{DBLP:conf/birthday/BojanczykW08}. In the forest algebra a forest can be constructed using two operations: the horizontal concatenation $\conch$ and the vertical concatenation $\concv$. The horizontal concatenation simply concatenates two forests horizontally, analogous to the string concatenation. The vertical concatenation of two forests $f_1$ and $f_2$ is only defined if $f_1$ contains exactly one $\ast$-labelled leaf, and $f_1 \concv f_2$ is then obtained by appending $f_2$ at the bottom of $f_1$, namely at the $\ast$-labelled leaf (i.e., the $\ast$-labelled leaf is replaced by $f_2$'s roots). A forest SLP is then a \emph{directed acyclic graph} (DAG) that unfolds into a forest algebra expression. See Sections~\ref{sec-FA} and \ref{subsubsection:FSLPs} for a formal treatment and examples (in particular, \autoref{fig-fa} shows the graphical representation of a forest algebra expression that evaluates to the tree on the right side of \autoref{fig-fslp}, and the left side of \autoref{fig-fslp} shows the f-SLP induced by this forest algebra expression). Since typical tree-structured data is unranked (e.g., XML trees or tree decompositions), forest SLPs are a relevant compression scheme in our setting. 

Forest SLPs share many of the desirable properties of s-SLPs. They cover other popular tree compression schemes like top dags \cite{BilleGLW15,DudekG18,Hubschle-Schneider15} and tree SLPs \cite{GanardiHJLN17,LohreyMM13}. Furthermore, there exist compressors such as  TreeRePair~\cite{LohreyMM13} that can be used to produce forest SLPs and that show excellent compression ratios in practice. Other available grammar-based tree compressors are BPLEX~\cite{BusattoLM08} and CluX~\cite{BottcherHK10}.  

Further motivation for the choice of SLPs as a compression scheme for strings and trees in the context of ACD shall be deferred to the end of this work (see \autoref{sec:SLPsBackground}, where we provide comprehensive background information on SLPs). Instead, let us move on to a summary of our main result.

\subsection{Main result}

Our main result lifts the linear preprocessing and output-linear delay enumeration algorithms for MSO-queries over trees to the case of MSO-queries over SLP-compressed unranked forests. We use the general notation $q[D]$ to denote the result set of query $q$ over data $D$, and our MSO-queries have the form $\Psi(X_1, X_2, \ldots, X_k)$, where $X_1, X_2, \ldots, X_k$ are free set variables. Complexity bounds in the following theorem and in the remainder of the introduction hold with respect to data complexity.

\begin{theorem}\label{mainResultForests}
Fix an MSO-query $\Psi(X_1, X_2, \ldots, X_k)$. For an unranked forest $F$ that is given in compressed form by a forest SLP $\FSLP$, one can enumerate $\Psi[F]$
after linear preprocessing $\bigO(|\FSLP|)$ and with output-linear delay.
\end{theorem}

The algorithm behind this theorem enumerates the answers of an MSO-query with linear preprocessing and output-linear delay (just like the algorithms from~\cite{Bagan06,KazanaS13}), but the input is compressed by a forest SLP $\FSLP$ and the preprocessing is linear only in the size $\card{\FSLP}$ of the forest SLP instead of the size of the decompressed forest $F$.
 Hence, depending on the size of the forest SLP $\FSLP$ in comparison to the data $F$ (which, theoretically, might be logarithmic, and in practical scenarios can be assumed to be rather small), this yields enumeration with optimal delay, but potentially much faster preprocessing. 

Our result also covers and properly extends known results from the literature, e.g., the enumeration algorithms for regular document spanners on compressed and uncompressed strings~\cite{AmarilliEtAl2021,FlorenzanoEtAl2020,MunozRiveros2025, SchmidS21,SchmidSchweikardt2022,SchmidS22}, and the enumeration algorithms for MSO-queries over uncompressed trees~\cite{AmarilliBMN19,Bagan06,Courcelle2009,KazanaS13}. It can also be seen as a meta-theorem in the field of algorithmics on SLP-compressed forests (and therefore also strings): Any enumeration problem on SLP-compressed forests (or strings) that can be formulated in MSO-logic can be solved with output-linear delay after a preprocessing linear in the size of the SLP. This covers practically relevant tasks like enumerating all occurrences of a pattern (described by a single string or a regular expression) in a string, enumerating all tandem repeats $w^k$ in a biosequence (with $k \leq c_1$ and $|w| \leq c_2$ for reasonable constants $c_1, c_2$), enumerating the vertices of an unranked tree that have only children of the same kind, enumerating all pairs of close cousins of a phylogenetic tree, etc. All these problems have MSO-formulations. Hence, \autoref{mainResultForests} yields an algorithm for enumerating 
the query result in the case that the data is SLP-compressed. 

In addition to our main result, we also investigate the dynamic setting where we assume that after having enumerated the answers to a query, our data can be updated and after such an update, we still want to be able to enumerate the answers to our query (now with respect to the updated data), but without having to repeat the costly preprocessing. 

Query evaluation under updates has received a lot of attention over the last decades in database theory, since it covers the practically relevant scenario where we repeatedly evaluate queries over data that is subject to only small changes. With respect to MSO-evaluation over uncompressed trees (and strings),~\cite{AmarilliBMN19,MMN22,NiewerthSegoufin2018} present linear preprocessing and output-linear delay enumeration algorithms with update procedures handling insertions and deletions of leaves (symbols, respectively), and relabelling of vertices (symbols, respectively). The running time of these update procedures is logarithmic in the data size. 
We can show that \autoref{mainResultForests} can be extended by vertex relabelling updates in time that is logarithmic in the uncompressed data; see \autoref{relabellingTheorem-rooted} for the formal result. Applying the techniques of~\cite{AmarilliBMN19,MMN22,NiewerthSegoufin2018}
for the other types of tree updates in the SLP-compressed setting seems rather challenging and we leave this for future research; see \autoref{sec other updates} for a discussion.

\subsection{Proof techniques and novel aspects}

As explained above, algorithmics on SLP-compressed inputs is a large field with many theoretical and algorithmic results as well as practical implementations. However, most algorithms are tailored to particular computational problems and, to the best of our knowledge, our result is the first that treats the issue in form of a meta-theorem that yields an algorithm for any problem definable in MSO-logic. Moreover, the enumeration aspect is usually not a main focus in algorithmics on SLP-compressed inputs. 

Query evaluation over SLP-compressed data can also be seen as an approach to sublinear query evaluation, which is interesting for big data scenarios, where even a linear dependency on the size of the data might be too expensive. More precisely, we spend linear time only once when we compress the data by an SLP. Then we can evaluate an arbitrary query in time that is only linear in the size of the  SLP, which is potentially much smaller (logarithmic in the best case) than the actual data size. To our knowledge, this aspect has not yet been considered in the literature on database theory.

Existing algorithms for linear preprocessing and output linear delay of MSO-queries over uncompressed trees (see~\cite{Bagan06,AmarilliBMN19}) are similar in the sense that in the preprocessing the query (represented by an automaton) and the input tree are combined into a data structure that represents the whole query result. The enumeration phase then enumerates from this data structure the tuples of node sets. On the one hand, the underlying data structure must represent the exponentially large query result in a concise and implicit way, but, on the other hand, it explicitly contains the complete input tree. In the compressed setting, however, the data structure must respect the compression of the input data and it is therefore impossible for it to explicitly contain the nodes of the input tree. Consequently, the enumeration phase must not only derive all the elements of the query result from the data structure, it must also construct (or decompress) the actual nodes of the input tree that appear in the output tuples, and since this cannot be done in a preprocessing, it must be done on-the-fly during the enumeration. Let us give a high level description of how we solve this task, i.e., how we prove \autoref{mainResultForests}.

By several simplification steps, we first reduce the problem from \autoref{mainResultForests} to the following enumeration problem for binary trees: given a directed acyclic graph (DAG) $\DAG$ that unfolds to the binary vertex-labelled tree $T$ and a deterministic bottom-up tree automaton $\mathcal{B}$, enumerate all subsets $S \subseteq \leaves(T)$ such that $\mathcal{B}$ accepts the tree obtained from $T$ by marking all leaves from $S$ with a $1$. Since $T$ is a binary tree, we could solve this problem by using Bagan's algorithm from~\cite{Bagan06}, but $T$ is given by the DAG $\DAG$ and we cannot afford to explicitly construct $T$. Consequently, we have to adapt Bagan's algorithm in such a way that it can be used directly on DAG-compressed trees, which is not a trivial task. 
As mentioned above, we cannot afford to compute the data structure used by Bagan's algorithm, but we can compute a compressed variant of this data structure, which is an edge-labelled DAG. In order to exploit this compressed data structure in a similar fashion as done by the enumeration phase of Bagan's algorithm, we have to be able to enumerate with constant delay the labels of the paths of this DAG, since these edge labels represent the actual nodes of the uncompressed tree.\footnote{Actually, the edges of the DAG will be labelled with the morphisms of a category, but for the sake of an intuitive and high-level explanation we abstract here from this detail.} This enumeration procedure constitutes a non-trivial algorithmic result and is the main component of our algorithm (it will be proven as an independent algorithmic result in \autoref{sec-enumeration-general}, where we also demonstrate that it has further applications of independent interest). 

It is not unlikely that instead of extending Bagan's algorithm from~\cite{Bagan06} in this fashion, we can also extend the algorithm of~\cite{AmarilliBMN19} to the compressed setting, i.e., we may be able to obtain from a compressed input tree in linear time a compressed version of the data structure used by the algorithm from~\cite{AmarilliBMN19}, and then find a way to employ this compressed data structure in the enumeration phase. However, this would most likely lead to similar challenges compared to extending Bagan's algorithm. In particular, it is likely that extensions of other algorithms eventually also require a procedure that enumerates paths in a DAG (simply because in the SLP-compressed setting the nodes of the tree are represented by paths in a DAG and we cannot afford to precompute them in a preprocessing).

\subsection{Further related work}\label{sec:relatedWork}

A vast body of literature is concerned with straight-line programs and their application in the context of algorithmics on compressed data. We will give a comprehensive discussion of those aspects relevant for our work in \autoref{sec:SLPsBackground} towards the end of the paper.

Forest SLPs -- the central compression scheme in our work -- are based on forest algebra expressions. Forest algebras have  been also used in the context of MSO-enumeration on uncompressed trees in~\cite{MMN22} for the purpose of enabling updates of the queried tree in logarithmic time by updating and re-balancing a forest algebra expression for the tree. However, no compression is handled in~\cite{MMN22}. 

Relabelling updates have been also studied recently in the context of dynamic regular membership testing \cite{AmarilliBJP25}.
For a fixed regular set of forests and a forest that undergoes relabelling updates, the goal is to maintain the information whether
the current forest belongs to $L$. It is shown that this can be done time in time $\bigO(\log n / \log \log n)$ per update,  where $n$ is the number of nodes in the forest, and in constant time per update for so-called almost-commutative regular forest languages.
Also in \cite{AmarilliBJP25} the concept of forest algebras plays an important role.

The arguably simplest way of compressing a tree is to fold it into a DAG. In the best case this allows to represent a tree of size $n$ by
a DAG of size $\Theta(\log n)$ (take for instance a perfect binary tree).
In the context of database theory, DAG-compression has been investigated in~\cite{KochBG03,FrickGK03} for XPath and monadic datalog queries, but the enumeration perspective has not been investigated. 

Forest SLPs subsume DAGs in the sense that a DAG for a tree $t$ can be easily translated into a forest SLP for $t$ of asymptotically the same size as the DAG. Moreover, there are also trees, where forest SLPs compress exponentially better than DAGs, e.g., unary trees of the form $a(a(\ldots a(a) \ldots))$. The experimental study of~\cite{LohreyMM13} also shows that in a practical setting, DAG-compression cannot compete with forest SLPs. Further work on the compression performance of DAGs for XML can be found in \cite{Bousquet-MelouL15,LohreyMR17}.

There is a growing body of work on \emph{factorised databases} that investigates the task of producing the result of a conjunctive relational query as a \emph{factorised representation}, which is a compressed form obtained by exploiting redundancies (see~\cite{Olteanu2023, OlteanuSchleich2016}). The motivation is that such query results are often just intermediate data objects that are inputs for further computations within a larger pipeline, and computing them in factorised form can decrease the overall computation time. This approach, however, is slightly different from ours: In factorised databases, the compressibility is entailed by the fact that the object to compress is the result of a conjunctive query, while in our case tree structured input data is given in a compressed form.

\subsection{Organisation of the paper}

We start in \autoref{sec:prelim} with introducing general notations about algebra, trees, forests, directed acyclic graphs and enumeration algorithms. Then, in \autoref{sec-path-enumeration}, we present the algorithm for enumerating paths of DAGs, that will serve as an important building block for our enumeration algorithm for MSO-queries over SLP-compressed forests. We will also briefly present two further applications 
of our path enumeration algorithm in \autoref{sec-further-app}.
\autoref{sec-fslp} will be devoted to thoroughly introducing the concept of forest straight-line programs and a technique that shall be crucial for obtaining the preorder numbers of the nodes of a compressed forest. Basic concepts of monadic second order logic and the corresponding evaluation problem will be given in \autoref{sec-MSO-aut}. Then we present the proof of our main result in \autoref{sec-main-result} and the result on relabelling updates in \autoref{sec:updates}. Finally, we give some more detailed background information about the concept of straight-line programs (including some practical considerations) in \autoref{sec:SLPsBackground} and conclude the paper with \autoref{sec:conclusions}.

\section{Preliminaries}\label{sec:prelim}

\subsection{General notations}

Let $\mathbb{N} = \{1, 2, 3, \ldots\}$ and $[n] = \{1, 2, \ldots, n\}$ for $n \in \mathbb{N}$. By $2^A$ we denote the power set of a set $A$. 
For a binary relation $\to$, we use $\to^*$ to denote its reflexive-transitive closure.

Every finite sequence of elements from a finite alphabet $\Sigma$ is a \emph{word} (or \emph{string}) \emph{over $\Sigma$}. 
With $\Sigma^*$ we denote the set of all words over $\Sigma$ including the empty word $\eword$. For a word $w \in \Sigma^*$, $|w|$ denotes its length (in particular, $\card{\eword} = 0$).

We will occasionally also talk about enumeration with duplicates, which will be formalised with multisets. A \emph{multiset} is a set that can contain duplicates of an element and we use the standard $\multiset{\ldots}$ notation for denoting multisets.
A multiset with elements from a set $A$ can be formalised as a function $f : A \to \mathbb{N}$. 
Multisets will play only a marginal role for our results and we therefore omit a more formal treatment.

\subsection{Categories} \label{sec-category}

A \emph{category} (see, e.g., \cite{BaWells}) is a pair $(A, (M_{a,b})_{a,b \in A},\circ)$ where $A$ is a set of objects, $M_{a,b}$ is the set of morphisms from object $a$ to object $b$ and $\circ$ is a mapping that maps morphisms $\alpha \in M_{a,b}$ and $\beta \in M_{b,c}$ ($a,b,c \in A$) to a morphism $\alpha \circ \beta \in M_{a,c}$ such that the following hold:
\begin{itemize}
\item for all $a,b,c,d \in A$ and $\alpha \in M_{a,b}$, $\beta \in M_{b,c}$ and $\gamma \in M_{c,d}$: $(\alpha \circ \beta) \circ \gamma = \alpha \circ (\beta \circ \gamma)$,
\item for all $a \in A$ there is a morphism $\iota_a \in M_{a,a}$ such that for all  $a,b \in A$ and $\alpha \in M_{a,b}$: $\iota_a \circ \alpha = \alpha \circ \iota_b = \alpha$.
\end{itemize}
Note that a category with only one object is the same thing as a monoid. We use a general category $\mathcal{C}$ to state a general version of our core enumeration problem (see \autoref{sec-path-enumeration}).
In our main application (enumerating results of MSO-queries on compressed trees) only a specific category with two objects
and certain affine functions as morphisms will be used.
No specific results on categories are used in this paper.

\subsection{Trees and forests} \label{sec-trees}

We will work with different types of rooted trees and acyclic graphs in this paper. All graphs (including trees) will be finite.
Forests are sequences of trees, and we should keep in mind that each definition of a certain tree model yields the corresponding concept of forests (which are just sequences of such trees).

\subsubsection{Vertex-labelled ordered trees} \label{forests}

Vertex-labelled ordered trees are trees, where vertices are labelled with symbols from some alphabet $\Sigma$ and may have an arbitrary number of children (i.e., the trees are unranked). Moreover, the children are linearly ordered. A typical example of such trees are XML tree structures. A vertex-labelled ordered tree can be defined as structure $T = (V, E, R, \lambda)$, where $V$ is the set of vertices, $E$ is the edge relation (i.e., $(u,v) \in E$ if and only if $v$ is a child of $u$), $R$ is the sibling relation (i.e., $(u,v) \in R$ if and only if $v$ is the right sibling of $u$), and $\lambda : V \to \Sigma$ is the function that assigns a label to every vertex. In the following, when we speak of a tree, we always mean a vertex-labelled ordered tree. For a tree $T$ and one of its vertices $v$, we use $T(v)$ to denote $T$'s subtree rooted in $v$, i.e., the subtree of $T$ consisting of  all descendants of $v$ including $v$.

A  \emph{forest} is a (possibly empty) ordered sequence of trees; it is also described by a structure 
$(V,E,R,\lambda)$, where the roots of the forest are chained by the sibling relation $R$. Note that a string can be identified with a forest $(V,\emptyset,R,\lambda)$ (this is similar to the classical representation of strings as relational structures with the only difference that for strings we commonly use one unary relation per symbol instead of the labelling function $\lambda$).
We write $\F(\Sigma)$ for the set of all forests with vertex labels from $\Sigma$. The size $|F|$ of a forest $F$ is the number of vertices of $F$.

\pgfkeys{/pgf/inner sep=0.1em}
\begin{figure}
  \centering
  \begin{forest} baseline
    [$\phantom{0} a \; \scriptstyle{0}$, s sep=3mm, for tree={parent anchor=south, child anchor=north}
      [$\phantom{0} b \; \scriptstyle{1} $]
      [$\phantom{0} a \; \scriptstyle{2}$ 
        [$\phantom{0} a \; \scriptstyle{3}$]
      ]]  
    \end{forest} 
    \begin{forest} baseline [$\phantom{0} b \; \scriptstyle{4}$] \end{forest}
    \begin{forest} baseline [$\phantom{0} c \; \scriptstyle{5}$] \end{forest}
    \begin{forest} baseline 
    [$\phantom{0} b \; \scriptstyle{6}$
      [$\phantom{0} c \; \scriptstyle{7}$, s sep=3mm, for tree={parent anchor=south, child anchor=north}
      [$\phantom{0} a \; \scriptstyle{8}$] 
      [$\phantom{0} b \; \scriptstyle{9}$]
     ]]   
      \end{forest}
\caption{\label{fig-val(E)} A forest, where every
vertex is additionally labelled with its preorder number.}
\end{figure}

We also use a term representation for forests, i.e., we write elements of $\F(\Sigma)$ as strings over the alphabet $\Sigma \cup \{ (, )\}$. For example, the forest from \autoref{fig-val(E)} has the term representation $a(ba(a))bcb(c(ab))$ (or $a(b,a(a)), b, c, b(c(a, b))$ with commas for better readability). Note that the occurrences of the symbols from $\Sigma$ in the term representation of the forest $F \in \F(\Sigma)$  
correspond to the vertices of $F$. 

The \emph{preorder number} of a vertex $v$ in a forest $F$ is the position of $v$ in the 
(depth-first left-to-right) preorder enumeration of the vertices, where the root gets preorder number $0$.
In the term representation of $F$, the $i^{\text{th}}$ occurrence of a symbol from $\Sigma$ (starting with $i=0$)
corresponds to the  $i^{\text{th}}$ vertex of $F$ in preorder.
In \autoref{fig-val(E)} the preorder numbers are written next to the vertices. In the term representation
 of the forest $a(ba(a))bcb(c(ab))$ the preorder numbers are
\begin{alignat*}{10}
& a \;\, ( \, &&  b \;\, && a \;\, ( \, && a \, ) \;\, ) \, && b \;\, &&  c \;\,  && b \;\, ( \, && c \;\, ( \, && a \;\,  && b \;\, ) \;\, ) \ .
 \\[-2mm]
& \scriptstyle{0} && \scriptstyle{1} && \scriptstyle{2} && \scriptstyle{3} && \scriptstyle{4} && \scriptstyle{5} && \scriptstyle{6} && \scriptstyle{7} && \scriptstyle{8} && \scriptstyle{9} 
\end{alignat*}

\subsubsection{Vertex-labelled binary trees}

Vertex-labelled binary trees (binary trees for short) are the special case of the trees from the previous paragraph, where every vertex is either
a leaf or has two children (a left and a right child). It is then more common to replace the two relations $E$ (edge relation) and $R$ (sibling relation) 
by the relations
$E_\ell$ (left edges) and $E_r$ (right edges), where $(u,v) \in E_\ell$ (resp., $(u,v) \in E_r$)
 if $v$ is the left (resp., right) child of $u$. We write $E = E_\ell \cup E_r$ for the set of all edges.
Our binary trees have the additional property that 
$\Sigma$ is partitioned into two  disjoint sets $\Sigma_0$ and $\Sigma_2$ labelling leaves and internal vertices, respectively;
see \autoref{fig-DagFoldingExample} (left), where $\Sigma_0 = \{c,d\}$ and $\Sigma_2 = \{a,b\}$ for an example.
We use the above term representation for general trees also for binary trees.
With $\leaves(T)$ we denote the set of leaves of the binary tree $T$. 

Binary trees will be mainly used for describing algebraic expressions over algebras.

\subsubsection{Unordered trees}

Unordered trees are trees without vertex labels and without  an order on the children of a vertex. They will be used 
as auxiliary data structures in our algorithms. An unordered tree will 
be defined as a pair $(V,E)$, where $V$ is the set of vertices and 
$E$ is the edge relation. An unordered forest is a disjoint union of unordered trees.

\subsubsection{Decorated trees}

We also have to consider (ordered as well as unordered) trees, where the vertices and edges are decorated with objects and morphisms, respectively, from a 
category $\mathcal{C} = (A, (M_{a,b})_{a,b \in A},\circ)$. More precisely, a $\mathcal{C}$-decorated tree is equipped with a function $\gamma$ that maps every vertex $v \in V$ of the tree to an object $\gamma(v)$ and 
every edge $e = (u,v) \in E$ to a morphism $\gamma(e) \in M_{\gamma(u), \gamma(v)}$. These edge morphisms can be lifted from edges to paths in the natural way: let $v_1v_2 v_3 \cdots v_{d-1} v_d$ be the unique path from a vertex
$v_1$ to a descendant $v_d$ in the tree, i.e., $(v_i, v_{i+1}) \in E$ for all $i \in [d-1]$. We then define $\gamma(v_1,v_d)$ as the morphism $\gamma(v_1,v_2) \circ \gamma(v_2,v_3) \circ \cdots \circ \gamma(v_{d-1},v_d)$. For a leaf $v$ of the tree $T$ we define $\gamma^*(v) = \gamma(r,v)$,
where $r$ is the root of $T$, and for a set of leaves $S$ we define $\gamma^*(S) = \multiset{ \gamma^*(v) : v \in S}$. 

The outputs of our enumeration algorithms will be sets $\gamma^*(S)$ for certain leaf sets $S$ and a specific category $\mathcal{C}$.

\subsection{Directed acyclic graphs}
\label{sec-DAG}

As commonly defined, a \emph{directed acyclic graph} (DAG for short) is a directed graph $\DAG = (V, E)$ that has no cycles, i.e., there is no non-empty path from a vertex $v$ back to $v$.
We will need DAGs with multiple edges between vertices. For this, $E$ can be taken as a subset $E \subseteq V \times I \times V$ for some 
index set $I$. Then we can have two different edges $(u,i,v)$ and $(u,j,v)$ (with $i \neq j$) from $u$ to $v$. 
The size $|\DAG|$ of $\DAG$ is defined as $|E|+|V|$. 
The \emph{outdegree} (resp., \emph{indegree}) of a vertex $v$ is the number of edges of the form  $(v, i, u)$ (resp., $(u, i, v)$).
Analogously to trees, vertices of a DAG with outdegree $0$ are called leaves. A  \emph{path} (from $v_1$ to $v_n$) in $\DAG$ is a word  $\pi = v_1 i_1  \cdots v_{n-1} i_{n-1} v_n$ such that $n \geq 1$ and $(v_k, i_k, v_{k+1}) \in E$ for all $1 \le k \le n-1$. 
The length of this path $\pi$ is $|\pi| := n-1$.
We write $\omega(\pi)  = v_n$ for the terminal vertex of the path $\pi$.\label{omegaDef}
If $n=1$ (in which case we have $\pi = v_1$) we speak of the empty path at $v_1$. For $v \in V$ and $U \subseteq V$ let $\path_{\DAG}(v,U)$ be the set of all paths from $v$ to some vertex in $U$. We also write $\path_{\DAG}(v)$ for 
$\path_{\DAG}(v,V)$ (the set of all paths that start in $v$).
Paths in $\path_{\DAG}(v,L)$, where $L$ is the set of leaves of $\DAG$ will be also called \emph{$v$-to-leaf paths}.

If in the above definition the index set $I$ is $I = \{\ell, r\}$ then $\DAG$ is called a binary DAG. Then, edges in $E_\ell := E \cap (V \times \{\ell\} \times V)$ are called left edges and edges in $E_\ell := E \cap (V \times \{r\} \times V)$ are called right edges. We additionally assume that every
vertex $v$ is either a leaf or has a (necessarily unique) left and right outgoing edge. Note that $\card{V} \le \card{\DAG} \le 3 \card{V}$ for a 
binary DAG.

The DAGs in this paper will mostly contain some of the following additional components, where $\mathcal{C} = (A, (M_{a,b})_{a,b \in A},\circ)$
is a category and $\Sigma$ is a finite set of vertex labels.
\begin{itemize}
\item A decoration mapping $\gamma$ that assigns an object $\gamma(v)$ to every vertex $v \in V$ and a morphism $\gamma(e) \in M_{\gamma(u),\gamma(v)}$ to every edge $e = (u,i,v) \in E$. 
For a path $\pi = v_1 i_1  \cdots v_{n-1} i_{n-1} v_n$ we define its $\mathcal{C}$-morphism 
$$
\gamma(\pi) = \gamma(v_1, i_1, v_{2}) \circ  \gamma(v_2, i_2, v_3) \circ \cdots \circ   \gamma(v_{n-1}, i_{n-1}, v_{n}).
$$
For the empty path $\pi = v \in V$ at vertex $v$ we set $\gamma(\pi) = \iota_{\gamma(v)}$.\footnote{The reader familiar with category theory may notice 
that $\gamma$ is a functor from the path category of $\DAG$ to the category $\mathcal{C}$.}
The resulting structure $\DAG = (V,E,\gamma)$ is called a  \emph{$\mathcal{C}$-decorated DAG}.
\item A vertex labelling function $\lambda : V \to \Sigma$, which results in a vertex-labelled DAG.
We will need this concept only for a binary DAG $\DAG$. In this case, $\Sigma$ is the disjoint union of $\Sigma_0$ and $\Sigma_2$ 
and $\lambda(v) \in \Sigma_0$ for every leaf of $\DAG$ and $\lambda(v) \in \Sigma_2$ for every non-leaf vertex.
\end{itemize}
When a DAG contains one of the components $\gamma$ or $\lambda$, it will have the above meaning.

\subsection{Enumeration algorithms} \label{sec-enumeration-general}

We use the standard RAM model with a special assumption about handling the morphisms of the category $\mathcal{C}$ of $\mathcal{C}$-decorated DAGs (discussed towards the end of this section), and a restriction for the register length (discussed at the end of \autoref{subsubsection:FSLPs}). 

An enumeration problem is a function $E$ that maps an input $I$ to a finite set $E(I)$ of objects. We can assume that $I$ and the objects in $E(I)$ 
are written in RAM registers. An enumeration algorithm $A$ for $E$ is an algorithm that computes on input $I$ a sequence 
 $(s_1, s_2, \ldots, s_m, s_{m+1})$, where  $E(I) = \{ s_1, \ldots, s_m \}$, $s_i \neq s_j$ for all $i \neq j$ and 
  $s_{m+1} = \mathsf{EOE}$ is the \emph{end-of-enumeration} marker. The algorithm produces this sequence in order, i.e., the algorithm only starts with
  the computation of $s_{i+1}$, once it finishes outputting $s_i$.
 The \emph{preprocessing time} of $A$ on input $I$ is the time when the algorithm starts with outputting  $s_1$.
The preprocessing time  of $A$ is the maximum preprocessing time over all possible inputs $I$ of length at most $n$ (viewed as a function of $n$).

Besides the preprocessing time, the other important time measure for an enumeration algorithm $A$ is its \emph{delay}, which measures 
the maximal time between the computation of two consecutive outputs $s_i$ and $s_{i+1}$. In our situation, the output objects $s_i$ will be very large
(potentially much larger than the input $I$). Therefore, it does not make sense to measure the delay with respect to the input length $I$. In such situations, 
the notion of \emph{output-linear delay} makes sense. It requires that for every output sequence $(s_1, s_2, \ldots, s_m, s_{m+1})$ produced by the algorithm $A$,
if $t_i$ ($1 \leq i \leq m+1$) is the time when $A$ starts with outputting $s_i$, then $t_{i+1}-t_i = \bigO(|s_{i}|)$ for every $1 \leq i \leq m$.
If every output $s_i$ has constant size (which for the RAM model means that it occupies a constant number of registers), then output-linear delay is the same as constant delay.
The gold standard in the area of enumeration algorithms is (i) \emph{linear preprocessing} (i.e., the preprocessing time is $\bigO(|I|)$) and 
(ii) \emph{output-linear delay}.

We are interested in enumeration algorithms that enumerate subsets $U \subseteq V$ of a forest $F = (V,E,R,\lambda)$.
These subsets represent the results of a query $\mathcal{Q}$ (that is given by an MSO-formula or a tree automaton; see \autoref{sec-MSO-aut}).
The input is $F$, while the query  $\mathcal{Q}$ is fixed and not part of the input, i.e., we measure in data complexity. The special feature of this work is that the input forest $F$ is not given explicitly, but in a potentially highly compressed form, and the enumeration algorithm must be able to handle this compressed representation rather than decompressing it. This aspect shall be explained in detail in \autoref{sec-fslp}.

Since we also deal with $\mathcal{C}$-decorated DAGs, where morphisms from a category $\mathcal{C}$ are assigned to edges, we need the following assumption for the proof of our main result (recall that $\omega(\pi)$ is the terminal vertex of a path $\pi$; see \autoref{sec-DAG}):

\begin{assumption} \label{assumption-category}
Let $\mathcal{C} = (A, (M_{a,b})_{a,b \in A},\circ)$ be a category. If a $\mathcal{C}$-decorated DAG $\DAG = (V,E,\gamma)$ is part of the input then
\begin{itemize}
\item vertices from $V$ as well as
\item all morphisms $\gamma(\pi)$ for $\pi$ a path in $\DAG$ 
\end{itemize}
fit into a single register of our RAM. 
For two such morphisms $\gamma(\pi_1)$
and $\gamma(\pi_2)$ such that $\omega(\pi_1)$ is the first vertex of $\pi_2$, the morphism $\gamma(\pi_1) \circ \gamma(\pi_2)$
can be computed in constant time on the RAM. 
\end{assumption}

For the specific setting that arises in the proof of \autoref{mainResultForests}, this assumption will be justified later.
Only in \autoref{sec-further-app}, where we discuss some further applications of our path enumeration
technique from \autoref{sec-path-enumeration}, we will deviate from \autoref{assumption-category}.

While this is not our algorithmic focus, it will occasionally also be convenient to talk about enumeration with duplicates, which we formalise as enumerating a multiset. More precisely, we say that an algorithm $A$ on input $I$ enumerates a multiset $S$ if and only if the output sequence is $(s_1, s_2, \ldots, s_m, \mathsf{EOE})$, $|S| = m$ and $S = \multiset{s_1, s_2, \ldots, s_m}$. All other notions defined above apply in the same way also for enumerations of multisets. 

\section{Path Enumeration in DAGs} \label{sec-path-enumeration}

In this section, we provide a proof of the following result about path enumeration in decorated DAGs. This algorithmic result will be used in \autoref{sec-main-result} as a crucial building block for our enumeration algorithm for MSO-queries over SLP-compressed forests (see \autoref{mainResultForests}). Further applications that are discussed in \autoref{sec-further-app} demonstrate the 
independent interest of the result.

\begin{theorem} \label{thm-enumerate-paths}
Fix a category $\mathcal{C} = (A, (M_{a,b})_{a,b \in A},\circ)$.
Let $\DAG = (V,E,\gamma)$ be a $\mathcal{C}$-decorated (not necessarily binary) DAG such that \autoref{assumption-category} holds,
and let $V_0 \subseteq V$ be a distinguished set of target vertices. In time $\bigO(|\DAG|)$ one can compute a data structure that allows to enumerate for a given source vertex $s \in V$ in constant delay the multiset 
$\multiset{ \langle \omega(\pi), \gamma(\pi) \rangle : \pi \in \path_{\DAG}(s,V_0)}$.\footnote{In general, we might have $\langle\omega(\pi),\gamma(\pi)\rangle = \langle\omega(\pi'),\gamma(\pi')\rangle$ for 
different paths $\pi, \pi' \in \path_{\DAG}(s,V_0)$ (although this does not happen for the decoration used later on in the application of this result in the context of \autoref{mainResultForests}).
Therefore, we have to work with multisets.}
\end{theorem}
Let us write $M$ for the union of all the morphism sets $M_{a,b}$ from \autoref{thm-enumerate-paths}. In the following, we let $s \in V$ be the source vertex that we get as input according to the statement of \autoref{thm-enumerate-paths}, i.e., the vertex for which we wish to enumerate the multiset $\multiset{ \langle \omega(\pi), \gamma(\pi) \rangle : \pi \in \path_{\DAG}(s,V_0)}$. Note that according to \autoref{thm-enumerate-paths}, the data structure that allows enumeration should work for any possible $s \in V$, which means that our preprocessing must be independent from the choice of $s$.

\subsection{Preprocessing} \label{sec preproc path enum}

Let us start with some preprocessing for the DAG $\DAG = (V,E,\gamma)$. We first remove as long as possible  vertices $v \in V$ of outdegree zero that do not belong to the target set $V_0$ (and we remember these removed vertices so that in case the given source vertex $s$ is one of them, we can simply output the empty list). We can therefore assume that all vertices of outdegree zero belong to $V_0$. If there is a vertex $v \in V_0$ of non-zero outdegree, we can add a copy $v'$
together with a new edge $(v,i,v')$ and extend the decoration by $\gamma(v') = \gamma(v)$ and  $\gamma(v,i,v') = \iota_{\gamma(v)}$ (recall that $\iota_a$ is the identity
morphism for object $a$). Moreover, we remove $v$ from $V_0$ and 
add $v'$ to $V_0$.
By this, we can assume that $V_0$ consists of all vertices of outdegree zero (the leaves of the DAG).

We next eliminate vertices of outdegree one.
In time $\bigO(\card{\DAG})$ we first determine the set $V_1$ of vertices of outdegree one. 
 For every vertex $u \in V_1$ we then compute the unique vertex $f(u)$ such that the outdegree of $f(u)$ is not $1$ and $f(u)$ is reached
 from $u$ by the (unique) path $\pi_u$ consisting of edges $(v_1, i, v_2) \in E$ with $v_1 \in V_1$. We also compute the morphism $\gamma(\pi_u)$ for this
 path. This can be done bottom-up in time $\bigO(\card{\DAG})$ as follows: For every edge $(u, i, v)$ with $u \in V_1$ we set $f(u) = v$ and $\gamma(\pi_u) = \gamma(u, i, v)$ if $v \notin V_1$ (this includes the case where $v$ is a leaf), and we set $f(u) = f(v)$ and $\gamma(\pi_u) = \gamma(u, i, v) \circ \gamma(\pi_v)$ if $v \in V_1$. We then replace every edge $(v,i,u) \in E$ with $u \in V_1$ by $(v,i,f(u))$ and set $\gamma(v,i,f(u)) = \gamma(v,i,u)\circ\gamma(\pi_u)$. 
 After this step, there is no edge that ends in a vertex of outdegree one. In particular, if a vertex has outdegree one,
 its indegree is zero. We can then remove all vertices of outdegree one and their outgoing edges from $\DAG$ and store
 them together with their decorations in a separate table. In case that the source vertex $s$ has outdegree one and $(s,i,v)$ is its unique outgoing edge, we run the enumeration algorithm for $v$ instead of $s$
 and multiply $\gamma(s,i,v)$ on the left to every $\mathcal{C}$-morphism that is printed in the enumeration for $v$.
 We can now assume that all vertices of the DAG $\DAG$ have outdegree zero or at least two.
 The set $V_0$ still consists of all vertices of outdegree zero.

We next transform $\DAG$ into a binary DAG $\DAG_b = (V_b, E_b, \gamma_b)$, where $V \subseteq V_b$ and every vertex $v \in V_b$ has either outdegree zero or two.
For this, take a vertex $u_1 \in V$ with outdegree at least three (recall that all vertices have outdegree either $0$ or at least $2$). Let 
$(u_1, i_1, v_1), \ldots, (u_1, i_d, v_d)$ be all outgoing edges of $u$ ($d \geq 3$).
We then add new vertices $u_2, \ldots, u_{d-1}$ to $V_b$ 
and add the following edges to $E_b$: 
\begin{itemize}
\item all right edges $(u_k, r, u_{k+1})$ for $1 \leq k \leq  d-2$,
\item all left edges $(u_k, \ell, v_k)$ for $1 \leq k \leq  d-1$ and
\item $(u_{d-1}, r, v_d)$.
\end{itemize}
The decorations of the new vertices and edges are defined as follows:
\begin{itemize}
\item $\gamma_b(u_k) = \gamma(u_1)$ for all $2 \leq k \leq d-1$,
\item $\gamma_b(u_k, r, u_{k+1}) = \iota_{\gamma(u_1)}$ for $1 \leq k \leq  d-2$,
\item $\gamma_b(u_k, \ell, v_k) = \gamma(u_1, i_k, v_k)$ for $1 \leq k \leq  d-1$ and 
\item $\gamma_b(u_{d-1}, r, v_d) = \gamma(u_1, i_{d}, v_d)$.
\end{itemize}
The binary DAG $\DAG_b = (V_b, E_b, \gamma_b)$ can be easily computed in time $\bigO(\card{\DAG})$. Obviously, $V \subseteq V_b$ and the set of leaves of $\DAG_b$ is still $V_0$. Moreover, there is a one-to-one correspondence between the $s$-to-$V_0$ paths in $\DAG$ and the $s$-to-$V_0$ paths in $\DAG_b$, and the $\mathcal{C}$-morphisms of corresponding such paths are the same. Hence, we can now work with the binary DAG $\DAG_b$ and, for simplicity, we omit the subscript $b$, i.e., we write $\DAG$, $V$, $E$ and $\gamma$ instead of $\DAG_b$, $V_b$, $E_b$ and $\gamma_b$. We also define the set $V_2 = V \setminus V_0$ of internal vertices.

For a vertex $v \in V_2$ we will write in the following $v[\ell]$ and $v[r]$ for the left and right child of $v$, respectively, i.e., $(v, \ell, v[\ell]), (v, r, v[r]) \in E$.

\subsection{Enumeration} \label{subsec path enumeration} 

Recall that $s$ is our source vertex of $\DAG$ in \autoref{thm-enumerate-paths}.
For a word $\pi \in \{\ell,r\}^*$ we define
the vertex $s[\pi]$ and the $\mathcal{C}$-morphism $\gamma[s,\pi]$ inductively as follows:
We start with $s[\eword] = s$ and $\gamma[s,\eword] = \iota_{\gamma(s)}$.
Let us now assume that 
$s[\pi]$ and $\gamma[s,\pi]$ are already defined. If the vertex $s[\pi]$ has outdegree zero then for $d \in \{\ell,r\}$, 
$s[\pi d]$ and $\gamma[s,\pi d]$ are not defined. Otherwise, we define
$s[\pi d] = (s[\pi])[d]$ and $\gamma[s,\pi d] = \gamma[s,\pi] \circ \gamma(s[\pi],d,s[\pi d])$.
Finally, let $L_s \subseteq \{\ell,r\}^*$ be the set of all words $\pi \in \{\ell,r\}^*$
such that $s[\pi]$ is defined and $s[\pi] \in V_0$. The set $L_s$ is in a one-to-one-correspondence with the paths from $s$ to $V_0$, and
we will speak of paths for elements of $L_s$ in the following.

The goal of the enumeration algorithm can therefore be formulated as follows: enumerate the pairs
$\langle s[\pi],\gamma[s,\pi] \rangle$, where $\pi$ ranges over all words from $L_s$. 

Let us illustrate this with an example. \autoref{fig-ExampleDagForPathEnum} shows a possible input DAG for \autoref{thm-enumerate-paths} (after the preprocessing described above, i.e., it is a binary DAG without vertices of outdegree $1$ and $V_0$ is the set of leaves), and let us assume that the vertex $s$ is the vertex $3$. 
For the category $\mathcal{C}$ we take the monoid $(\mathbb{N},+)$, and we call $\gamma(e)$ the weight of the edge $e$.
Our task is now to enumerate all $\langle s[\pi], \gamma[s,\pi] \rangle$, where $\pi$ describes a path from $3$ to one of the leaves $11$, $12$ and $13$. There are $4$ individual paths from $3$ to $11$, which yield the pairs $\langle 11, 16 \rangle, \langle 11, 10 \rangle, \langle 11, 17 \rangle, \langle 11, 11 \rangle$, $8$ individual paths from $3$ to $12$, which yield the pairs $\langle 12, 12 \rangle, \langle 12, 6 \rangle, \langle 12, 13 \rangle, \langle 12, 7 \rangle$ (if we go via vertex $5$) and the pairs $\langle 12, 13 \rangle, \langle 12, 18 \rangle, \langle 12, 15 \rangle, \langle 12, 20 \rangle$ (if we go via vertex $6$), and $4$ individual paths from $3$ to $13$, which yield the pairs $\langle 13, 11 \rangle, \langle 13, 16 \rangle, \langle 13, 13 \rangle, \langle 13, 18 \rangle$. Note that the pair $\langle 12, 13 \rangle$ occurs twice, due to the path $(3, 0, 5), (5, 4, 7), (7, 8, 9), (9, 1, 12)$ with weight $13$ and the path $(3, 8, 6), (6, 2, 8), (8, 0, 10), (10, 3, 12)$ with weight $13$ (note that we write here paths as sequences of edges $(u, \gamma, v)$, where $\gamma$ is the weight of the edge $(u, v)$). As explained before, we do not have to take care of such duplicates, since in the application of the enumeration procedure for the proof of \autoref{mainResultForests} no duplicated entries appear.

\begin{figure}
\begin{center}

\scalebox{1}{
\includegraphics{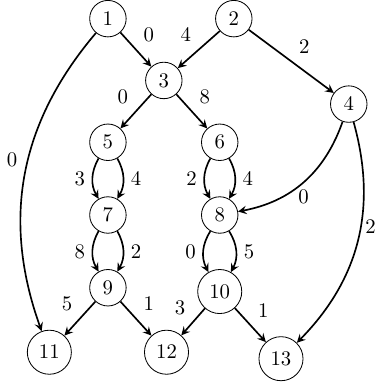}
}

\end{center}
\caption{A possible input DAG for \autoref{thm-enumerate-paths} after the preprocessing, i.e., the DAG is binary, it has no vertices of outdegree $1$ and the set $V_0$ is the set $\{11, 12, 13\}$ of leaves. Note that the edge decorations are integers labelling the edges.}
\label{fig-ExampleDagForPathEnum}
\end{figure}

Let us now discuss a simple algorithm that, on the one hand, produces all these desired pairs, but, on the other hand, \emph{does not} achieve constant delay (our actual algorithm will be an improvement of this algorithm). We just start in vertex $s$ and then carry out a depth-first search in order to explore all paths from $s$ to a leaf. While producing these paths, we can maintain their $\mathcal{C}$-morphisms, which allows us to output at every vertex $v \in V_0$ the pair consisting of $v$ and the current $\mathcal{C}$-morphism.
More formally, we can initialise variables $v := s$, $\gamma := \iota_{\gamma(s)}$ and $\pi := \eword$, so that $\pi$ is a path from $s$ to $v$ with morphism $\gamma$ (recall that we represent paths by words over $\{\ell, r\}$ as explained above). Now the depth-first search is done by moving from the current vertex $v$
 to its left child $v := v[\ell]$, updating the morphism by $\gamma := \gamma \circ \gamma(v, \ell, v[\ell])$ and the path by $\pi:= \pi \ell$. In addition we also have to store the triple $(v[r], \gamma \circ \gamma(v, r, v[r]), \pi r)$ on a stack to take care of it later. We repeat this until we reach a leaf, which means that we can produce an output (formally, we output $\langle v, \gamma \rangle$). Then we pop a triple $(v', \gamma', \pi')$ from the stack, we set $v := v'$, $\gamma := \gamma'$ and $\pi := \pi'$ and proceed as before. Obviously, this yields an enumeration of the desired pairs $\langle s[\pi],\gamma[s,\pi]\rangle$, but the delay can be non-constant (it can be proportional to the maximal length of a path in the DAG).

\label{page-omega_r}

For solving this problem, we need some more notations. For a vertex $v$ we denote with $\omega_r[v] \in V_0$ the unique leaf vertex that is reached from $v$ by following right edges: $\omega_r[v] = v$ if $v \in V_0$ and $\omega_r[v] = \omega_r[ v[r] ]$ if $v \in V_2$. Moreover, we define $\gamma_r[v]$ as the morphism of the unique path of right edges from $v$ to $\omega_r[v]$: $\gamma_r[v] = \iota_{\gamma(v)}$ if $v \in V_0$ and $\gamma_r[v] = \gamma(v, r, v[r]) \circ \gamma_r[ v[r] ]$ if $v \in V_2$. These data can be precomputed in time $\bigO(\card{\DAG})$ by a bottom-up computation on the DAG $\DAG$.

\SetAlgoLongEnd
\begin{algorithm}[t]
\SetKwComment{Comment}{(}{)}
\SetKwInput{KwGlobal}{variables}
\SetKw{KwPrint}{print}
\SetKw{KwStop}{stop}
\KwGlobal{$v \in V$, $\gamma \in M$, $\mathsf{stack} \in (V_2 \times M$ \colorbox{lightgray}{$\times \, \{\ell,r\}^*\!$}  $)^*$, $\mathsf{flag} \in \{0,1\}$, \colorbox{lightgray}{$\pi \in \{\ell,r\}^*$}} 
$v := s$  ; $\gamma := \iota_{\gamma(s)}$ ;  $\mathsf{stack} := \eword$  ;   $\mathsf{flag} := 1$ ;  \colorbox{lightgray}{$\pi := \eword$}\;
\While{{\sf true}}{
\If{$\mathsf{flag} = 1$\label{line:ifflag1}}{\KwPrint $\langle\omega_r[v], \gamma \circ \gamma_r[v]\rangle$\label{line:printright}}
$\mathsf{flag} := 1$\label{line:flag1}\;
\uIf{$v \in V_2$}
{ \If{$v[r] \in V_2$\label{line:vinV2}}{$\mathsf{stack.push} (v[r], \gamma \circ \gamma(v, r, v[r]), \colorbox{lightgray}{\text{$\pi r$}})$\label{line:push}}
  $v := v[\ell]$ ; $\gamma := \gamma \circ \gamma(v, \ell, v[\ell])$ ; \colorbox{lightgray}{$\pi:= \pi \ell$} \label{line:moveleft}}
\uElseIf{$\mathsf{stack} \neq \eword$}
{$(v,\gamma, \colorbox{lightgray}{\text{$\pi$}}) := \mathsf{stack.pop}$\label{line:pop}\;
 $\mathsf{flag} := 0$\label{line:flag0}}
\Else{\KwPrint {\sf EOE} \;
  \KwStop}
}
\caption{path\_enumeration$(s)$.  \label{path-enumerate-with-pi}}
\end{algorithm}

Our new algorithm (Algorithm~\ref{path-enumerate-with-pi}) has variables $v$, $\pi$ and $\gamma$ 
with the same meaning as in the above depth-first search. In 
each iteration of the main while loop, the algorithm behaves as follows:
\begin{enumerate}
\item If the current path $\pi$ ends with $\ell$ or is empty, then we print the pair $\langle \omega_r[v], \gamma \circ \gamma_r[v]\rangle$.
In other words: We extend the current path $\pi$ maximally to the right until we reach a leaf. Then we print the pair corresponding to this path.
\item \label{point--ii} If the current path $\pi$ ends with $r$ then the pair $\langle \omega_r[v], \gamma \circ \gamma_r[v]\rangle$ is \emph{not} printed.
This can be justified as follows: write $\pi = \xi r^k$, where $\xi$ does not end with $r$ ($\xi$ can be empty). 
Then, in a previous iteration of the while loop, the path variable $\pi$ had the value $\xi$,
and in this iteration the pair $\langle \omega_r[v], \gamma \circ \gamma_r[v]\rangle$ was already printed
(because $\omega_r[v]$ is also the leaf obtained from extending $\xi$ maximally to the right).
\end{enumerate}
The above behavior is achieved with the flag in Algorithm~\ref{path-enumerate-with-pi}.
Note that the current path $\pi$ ends with $r$ if and only if 
we pop from the stack, i.e., Line~\ref{line:pop} is executed. We then set the flag to zero in Line~\ref{line:flag0}. Moreover, when
we start a new iteration of the while loop we only print the pair $\langle \omega_r[v], \gamma \circ \gamma_r[v] \rangle$ if the flag is 1 (Line~\ref{line:ifflag1}~and~\ref{line:printright}).
Line~\ref{line:flag1}, where the flag is set to 1 is explained in a moment.

This modified algorithm with the flag variable is still correct, i.e., it enumerates the pairs $\langle s[\pi], \gamma[s,\pi]\rangle$ for $\pi \in L_s$.
In order to achieve constant delay it is crucial that the algorithm pushes a triple $(v[r], \gamma \circ \gamma(v, r, v[r]), \pi r)$ in Line~\ref{line:push}
on the stack
only if $v[r]$ is not a leaf of the DAG; see Line~\ref{line:vinV2}. First of all, this restriction 
does not harm the correctness of the algorithm: Assume that $v[r]$ is a leaf so that Line~\ref{line:push} is not executed.
The pair $\langle v[r], \gamma \circ \gamma(v, r, v[r])\rangle$ has been printed in a previous iteration of the while loop; see the above argument
in Point \eqref{point--ii}. Moreover, since $v[r]$ is a leaf, there
is no reason to return later to the vertex $v[r]$. Hence, it is not necessary to push $(v[r], \gamma \circ \gamma(v, r, v[r]), \pi r)$ on the stack.

To see that the restriction in Line~\ref{line:vinV2} is needed for constant delay, assume for a moment that we would push
a triple $(v[r], \gamma \circ \gamma(v, r, v[r]), \pi r)$ on the stack in Line~\ref{line:push} also if $v[r]$ is a leaf.
Consider for instance the case, where $L_s = \{\ell^n \} \cup \{\ell^i r : 0 \leq i \leq n-1\}$. In the first $n$ iterations, the algorithm prints the pairs $\langle s[\ell^i r],\gamma[s,\ell^i r] \rangle$ ($0 \leq i \leq n-1$) followed by
$\langle s[\ell^n],\gamma[s,\ell^n] \rangle$. Moreover, it pushes the triples $(s[\ell^i r],\gamma[s,\ell^i r], \ell^i r)$ for 
$0 \leq i \leq n-1$ on the stack. These triples will then be popped again from the stack in $n$ iterations but nothing is printed in these iterations. Hence, the delay would be not constant. 

By the above discussion, whenever Algorithm~\ref{path-enumerate-with-pi} pops the new path $\pi$ from the stack (Line~\ref{line:pop}), $\pi$ ends with $r$,
the flag is set to zero, and moreover, $s[\pi] \in V_2$, which means that $\pi \ell$ is a valid path of the DAG. In the next iteration of the while loop, nothing is printed (due to Line~\ref{line:ifflag1}).
Moreover, since $v \in V_2$, the variable $\pi$ will be set to $\pi \ell$ in Line~\ref{line:moveleft}. Hence, we can safely reset the flag to 1 in Line~\ref{line:flag1} (the new $\pi$
does not end with $r$).
Moreover, in the next iteration of the while loop a pair will be printed in Line~\ref{line:printright}. Hence, there cannot be two consecutive iterations of the while loop, where no pair is printed. This shows that Algorithm~\ref{path-enumerate-with-pi} works with constant delay.
 
Finally, notice that we can omit the code with gray background dealing with the variable $\pi$ in Algorithm~\ref{path-enumerate-with-pi}. It does not influence the control flow of the algorithm.

\subsection{Further applications} \label{sec-further-app}

Our \autoref{thm-enumerate-paths} serves as an important component in our algorithm for enumerating MSO-queries on SLP-compressed forests (\autoref{mainResultForests}). Nevertheless, as we shall briefly discuss in this section, it is rather general and can be used (with some slight variations) to solve other relevant enumeration problems.
The reader, who is only interested in  the proof of \autoref{mainResultForests} can skip Sections~\ref{sec free} and \ref{sec perm} and
continue with \autoref{sec-fslp}.

\subsubsection{Free monoids} \label{sec free}

A special case of \autoref{thm-enumerate-paths} that deserves further clarification is the one where $\mathcal{C}$ is a free monoid 
$\Sigma^*$ for some finite alphabet $\Sigma$ and the edges of the DAG $\DAG$ are decorated with elements from $\Sigma \cup \{\varepsilon\}$.
In this case the elements $\gamma(\pi)$ are words that can be as long as the path $\pi$.
\autoref{assumption-category} would mean that such words fit into a single RAM register, which can be hardly
justified. On the other hand, the free monoid case is relevant since it covers an important setting in the field of information extraction,
namely, it yields an enumeration algorithm for so-called annotation transducers.

An \emph{annotation transducer} is an NFA $\mathcal{T}$ whose transitions are labelled by pairs $(a,x) \in \Sigma \times (\Gamma \cup \{\diamond\})$ for an input alphabet $\Sigma$, and set of markers $\Gamma$, and the  \emph{empty marker} $\diamond \notin \Gamma$. An annotation transducer reads an input over $\Sigma$ by means of the first components of the transition labels just like a normal NFA, but due to the markers, it is interpreted as a query that maps an input string $w \in \Sigma^*$ to a set of output words $\mathcal{T}(w)$ as follows:
Take all words $(a_1, x_1) \cdots (a_n, x_n) \in (\Sigma \times \Gamma)^*$ that are accepted by $\mathcal{T}$ and 
such that  $w = a_1 \cdots a_n$. For each such word 
we only keep the word consisting of all pairs $(i, x_i)$ such that
$x_i \in \Gamma$ and put it into $\mathcal{T}(w)$. For example, if $\Sigma = \{a, b\}$, $\Gamma = \{x, y\}$ and
$\mathcal{T}$ accepts $(a, \diamond) (b, y) (a, \diamond) (b, \diamond) (b, x) (a, \diamond)$ as well as $(a, \diamond) (b, y) (a, x) (b, \diamond) (b, \diamond) (a, y)$ then $\mathcal{T}(ababba)$ contains the output words $(2, y) (5, x)$ and $(2, y) (3, x) (6, y)$. An important problem in the context of information extraction is to enumerate the set $\mathcal{T}(w)$ with a linear preprocessing (in data complexity) and output-linear delay (i.e., the delay depends linearly on the length of the next output word). See~\cite{BourhisEtAl2021,GawrychowskiEtAl2024,MunozRiveros2025} for variants of this problem, and~\cite{AmarilliEtAl2021,FaginEtAl2015,Schmid2024,SchmidSchweikardt2022,SchmidSchweikardt2024} for general information about information extraction. 

The above setting can be easily expressed in the setting of \autoref{thm-enumerate-paths}. We let $\mathcal{C}$ be the free monoid  generated by $\{1, 2, \ldots, |w|\} \times \Gamma$. We combine $w$ and $\mathcal{T}$ into a $\mathcal{C}$-decorated DAG $\mathcal{D}$ with a source vertex $s$ and a sink vertex $v$ such that the labels of the $s$-to-$v$-paths are exactly the output words from $\mathcal{T}(w)$. 
Constructing $\mathcal{D}$ is straightforward and similar to the product automaton construction. Thus, invoking the algorithm of \autoref{thm-enumerate-paths} on $\mathcal{D}$, vertex $s$ and $V_0 = \{v\}$ enumerates the output words of $\mathcal{T}(w)$. In case that $\mathcal{T}$ is \emph{unambiguous}, i.e., there are no two different accepting runs that produce on the same input word $w$ 
the same output word, the enumeration contains no duplicates. If we apply \autoref{thm-enumerate-paths} as stated above, then the delay is constant, but this assumes that the enumerated output words fit into single RAM registers, which is unrealistic, considering that they might be as long as the input string $w$. We now explain how to adapt Algorithm~\ref{path-enumerate-with-pi} such that we still get output-linear delay for the free
monoid case.

Let $\Sigma^*$ be a free monoid finitely generated by the alphabet $\Sigma$, and let $\DAG = (V, E, \gamma)$ be a DAG with $\gamma(e) \in \Sigma \cup \{\eword\}$ for every edge $e$ (since we have a category with a single object, we can omit the values 
$\gamma(v)$ for $v \in V$). We assume that all elements $x \in \Sigma \cup \{\eword\}$ fit into a single register of our RAM. We will also call $\gamma(e)$ the \emph{edge label} of an edge $e$ and $\gamma(\pi)$ the \emph{path label} of a path $\pi$ (i.e., the concatenation of the labels of its edges). 

The preprocessing slightly differs from \autoref{sec preproc path enum}. As before, we make sure that $V_0$ consists of all vertices of outdegree zero, and we compute the set $V_1$ of vertices with outdegree one. Next, we replace every edge $(v, i, u)$ with $u \in V_1$ by the edge $(v, i, f(u))$, where $f(u)$ is the unique vertex with $f(u) \notin V_1$ and $f(u)$ is reached from $u$ by the unique path $\pi_u$ of edges $(v_1, j, v_2)$ with $v_1 \in V_1$. The difference is that we do not label $(v, i, f(u))$ with $\gamma(v, i, u) \cdot \gamma(\pi_u)$ as before, but with the 
triple $(\gamma(v,i,u), u, f(u))$. Such a label $(\gamma(v,i,u), u, f(u))$ has constant size in our RAM model 
 and indicates that to $\gamma(v,i,u)$ we still have to append the label of the unique path from $u$ to $f(u)$. After this modification, every vertex has an outdegree of zero or of at least two, and, in the same way as before, we can transform the DAG into an equivalent DAG $\DAG_b$ whose vertices have all outdegree zero or two. Moreover, every edge is now labelled with an element from $\Sigma \cup \{\eword\}$ or a triple $(x,u, v)$, where $x \in \Sigma \cup \{\eword\}$ and $u$ and $v$ are vertices such that in the original DAG $\DAG$ there is a unique path from $u$ to $v$ with only vertices of outdegree one (except $v$). In particular, every edge label can be stored in a constant number of RAM registers. Obviously, the labels $(x, u, v)$ need to be turned into actual path labels later on, but we will take care of this at the end of our explanation. Let $\Delta$ be the set of all edge labels of $\DAG_b$ and define $\gamma(\zeta) \in \Sigma^*$ for $\zeta \in \Delta$
as follows: $\gamma(x) = x$ for $x \in \Sigma \cup \{\eword\}$ and $\gamma(x, u,v) = x w$ where $w \in \Sigma^*$ is the word 
labelling the unique path from $u$ to $v$ in $\mathcal{D}$. Let $\Delta_{\neq \eword} = \{ \zeta \in \Delta : \gamma(\zeta) \neq \eword \}$.
This set can be easily precomputed.

Algorithm~\ref{path-enumerate-with-pi} must now be adapted as follows. As an additional data structure, we need a trie $\mathsf{T}$ that stores path labels of $\DAG_b$. More precisely, $\mathsf{T}$ is a tree whose edge labels are elements from $\Delta_{\neq \eword}$.
A given node $\alpha$ of $\mathsf{T}$ represents the word $\gamma_{\mathsf{T}}(\alpha) = \gamma(\zeta_1) \cdots \gamma(\zeta_k) \in \Sigma^*$,
where $\zeta_1 \cdots \zeta_k$ is the sequence of edge labels in $\mathsf{T}$ from the root to $\alpha$.
 Note that $\gamma_{\mathsf{T}}(\alpha) \neq \eword$ if $\alpha$ is not the root.
 Initially, $\mathsf{T}$ consists of a single vertex. 
 The variable $\gamma$ in Algorithm~\ref{path-enumerate-with-pi} is replaced in our adapted version of the algorithm by
 a variable $\alpha$ that stores a node of $\mathsf{T}$ (initially it is the unique node of $\mathsf{T}$). The variable $\alpha$ should
 be seen as a succinct representation of the word $\gamma_{\mathsf{T}}(\alpha) \in \Sigma^*$ (which would be the value
 of the variable $\gamma$ in Algorithm~\ref{path-enumerate-with-pi}).
 In particular, at every time instance, $\gamma_{\mathsf{T}}(\alpha)$ will be the label of a path from the source vertex $s$ to $v$  in $\DAG_b$.

 In our adaptation of 
Algorithm~\ref{path-enumerate-with-pi} we push in
Line~\ref{line:push} the pair $(v[r], \alpha')$ on the stack, where $\alpha'$ is defined as follows: Let $\zeta \in \Delta$ be the label of the edge $(v, r, v[r])$. If $\gamma(\zeta) = \eword$
then $\alpha' = \alpha$ and $\mathsf{T}$ remains unchanged. If $\gamma(\zeta) \neq \eword$, then $\alpha'$ is a new vertex in the 
trie and we add the edge $(\alpha, \zeta, \alpha')$ to $\mathsf{T}$ (which can be done in constant time).\footnote{There might be already a node $\alpha''$ in the trie such that
$\gamma_{\mathsf{T}}(\alpha'') = \gamma_{\mathsf{T}}(\alpha')$. This does not cause problems.}
Finally, $\alpha$ is set to $\alpha'$.
Line~\ref{line:moveleft} is modified analogously except that we take the edge label of $(v,\ell, v[\ell])$ for $\zeta$ in the above
definition of the new node $\alpha'$ of $\mathsf{T}$. Moreover, we do not push on the stack but set $v$ to $v[\ell]$ and $\alpha$ to $\alpha'$.
Line~\ref{line:printright} requires a few more changes, since we have not precomputed the values $\omega_r(v)$ and $\gamma_r[v]$
(all the words $\gamma_r[v]$ cannot be produced explicitly in linear time).
 We first walk in the trie $\mathsf{T}$ from the root to the node $\alpha$ and output all edge labels along this path. Then we go to vertex $v$ of $\DAG_b$, we move along the right edges and thereby output each edge label $\zeta \in \Delta_{\neq \eword}$
 until we reach a vertex from $V_0$ that we will also output. 

This algorithm correctly enumerates the labels of all $s$-to-$V_0$ paths in $\DAG_b$. In order to
enumerate the path labels of all $s$-to-$V_0$ paths in the original DAG $\DAG$ in output-linear delay (i.e., 
a delay proportional to the length of the produced word), we have to resolve two small problems: Firstly, the labels of $s$-to-$V_0$ paths in $\DAG_b$ contain the succinct edge labels of the form $\zeta = (x, u, v) \in \Delta_{\neq \eword}$.
While producing the output in Line~\ref{line:printright}, such a label $\zeta$ has to be replaced by the word $\gamma(\zeta)$ by
producing $x$ followed by the sequence of edge labels along the unique path from $u$ to $v$ in $\DAG$.
 Now we enumerate the correct labels, but there is a slight problem with the delay: the path from $u$ to $v$ in $\DAG$ may contain 
  a long subpath of $\eword$-labelled edges, which do not contribute to the final output word. 
  The solution is simple: In the preprocessing, we compute shortcuts for these subpaths of $\eword$-labelled edges: For this we remove all edges in $\DAG$ except the outgoing edges of vertices with outdegree one, which gives us a forest (with edges pointing towards the roots) in which we can compute the required shortcuts in linear time. This problem of walking through sequences of $\eword$-labelled edges can also happen when we move along right edges of $\DAG_b$ in our modified Line~\ref{line:printright}. Such a path of right edges may contain a long subpath consisting of edges with labels
  from $\Delta \setminus \Delta_{\neq \eword}$. The solution is the same: In a preprocessing step we remove all edges in $\DAG_b$ except the right edges and then we compute shortcuts as before.

We summarise that by a minor modification of Algorithm~\ref{path-enumerate-with-pi}, \autoref{thm-enumerate-paths} can be adapted to the special case where $\mathcal{C}$ is a free monoid (see Corollary~\ref{thm-enumerate-paths-monoid-corollary-1} below). For this, we only need the assumption that edge labels of $\DAG$ fit into single registers of the RAM. As explained above, this covers the relevant case of enumerating the output words of an annotation transducer.

\begin{corollary} \label{thm-enumerate-paths-monoid-corollary-1}
Fix a free monoid $\Sigma^*$. Let $\DAG = (V,E,\gamma)$ be a (not necessarily binary) DAG such that $\gamma(e) \in \Sigma \cup \{\varepsilon\}$ for all edges $e \in E$
and all edge labels fit into single registers of the RAM. Let $V_0 \subseteq V$ be a distinguished set of target vertices. In time $\bigO(|\DAG|)$ one can compute a data structure that allows to enumerate for a given source vertex $s \in V$ in output-linear delay the multiset $\multiset{ \langle \omega(\pi), \gamma(\pi) \rangle : \pi \in \path_{\DAG}(s,V_0)}$.
\end{corollary}

\subsubsection{Permutation groups} \label{sec perm}

In this section we briefly discuss another application of \autoref{thm-enumerate-paths} to permutation groups.
The set $\mathsf{Sym}(n)$ of all permutations on $\{1,2,\ldots, n\}$ is a group under the operation of composition (the so-called
symmetric group of degree $n$). A permutation group $G$ is simply a subgroup of $\mathsf{Sym}(n)$ (denoted by $G \leq \mathsf{Sym}(n)$), which is usually given by
a generating set $A \subseteq \mathsf{Sym}(n)$. Thus, $G$ is the closure $\langle A \rangle$ of the set $A$ under composition. 
There are many algorithmic results for permutation groups; see, e.g., the monograph \cite{seress03}. Among others, there are polynomial 
time algorithms for testing membership in permutation groups and computing the size of a permutation group.

We can use \autoref{thm-enumerate-paths} to enumerate a given permutation group $G = \langle A \rangle \leq \mathsf{Sym}(n)$
with delay $\Theta(n)$ after a polynomial preprocessing time. Like many other permutation group algorithms,
our enumeration method is based
on the stabiliser chain $1 = G_n \leq G_{n-1} \leq \cdots \leq G_1 \leq G_0 = G$, where $G_i$ the subgroup of $G$ consisting  
of all $g \in G$ such that $g(a)=a$ for all $1 \leq a \leq i$. It is known that one can compute in polynomial time from $A$ for every
$0 \leq i \leq n-1$ a set of right coset representatives $R_i$ of $G_{i+1}$ in $G_i$; see \cite[Chapter~4]{seress03}. Thus, for every $g \in G_i$ there is a unique $r \in R_i$
such that $g \in G_{i+1} r$. The union $R = \bigcup_{0 \le i \leq n-1} R_i$ is called a strong generating set of $G$. It has the property 
that every element of $g \in G$ can be uniquely written as $g = r_{n-1} r_{n-2} \cdots r_0$ with $r_i \in R_i$ for all $0 \leq i \leq n-1$.
Hence, in order to enumerate $G$, we can first construct in polynomial time a strong generating set $R = \bigcup_{0 \le i \leq n-1} R_i$ and a corresponding DAG
$\mathcal{D}$ consisting of nodes $0,1, \ldots, n$ and $|R_i|$ many edges from $i+1$ to $i$ that are labelled with the elements from $R_i$
($0 \leq i \leq n-1$). We then use Algorithm~\ref{path-enumerate-with-pi} to enumerate all products labelling the paths from the source vertex $n$ to the target vertex $0$. This will not produce duplicates.

For complexity considerations, we assume a RAM of register length $\Theta(\log n)$ so that 
elements from $\{1,\ldots, n\}$ can be stored in a constant number of registers.
Hence, a permutation $g \in \mathsf{Sym}(n)$ is stored in $\Theta(n)$ registers and the multiplication of two permutations needs
time $\Theta(n)$. Under this assumption, the delay of Algorithm~\ref{path-enumerate-with-pi} is $\Theta(n)$, which is also the time needed to output
a single permutation. In other words: the enumeration algorithm works in output-linear delay.

\begin{corollary} \label{thm-enumerate-paths-monoid-corollary-2}
Given a set of permutations $A \subseteq \mathsf{Sym}(n)$, one can enumerate on a RAM with register length $\Theta(\log n)$
the permutation group $\langle A \rangle$ after polynomial preprocessing in delay $\Theta(n)$.
\end{corollary}

\section{Straight-Line Programs for Strings, Trees and Forests}\label{sec-fslp}

We now come back to the main topic of this paper: MSO-query enumeration in compressed forests.
In this section we will introduce \emph{forest straight-line programs} (f-SLPs for short), which is our formalism for unranked forest compression.
For better motivation we start with the simpler concept of \emph{string straight-line programs} (s-SLPs for short), which are a special
case of f-SLPs. Our definitions of s-SLPs and f-SLPs are equivalent to what is commonly found in the literature, but we choose a more algebraic point of view that is more convenient for our applications of SLPs. 

In \autoref{sec:SLPsBackground}, we provide some general background information on straight-line programs and their role in theoretical computer science. We also provide several references to the literature.

\subsection{Folding trees and unfolding DAGs} \label{folding}

A simple but important concept for SLPs are the operations of folding a labelled tree into a DAG and unfolding a DAG 
from a distinguished source vertex. To define this formally, let $T = (V, E_\ell, E_r, \lambda)$ be a vertex-labelled binary tree. Recall that every vertex $u$ of $T$ induces the subtree $T(u)$ rooted by $u$.
We can define an equivalence relation $\equiv$ on $V$ by $u \equiv v$ if and only if $T(u)$ and $T(v)$ are isomorphic (as vertex-labelled binary trees).
Then the \emph{DAG-folding} of $T$ is the quotient graph $T/_\equiv$.  Its vertices are the equivalence classes $[v]_\equiv$ ($v \in V$) and the label of $[v]_\equiv$ is $\lambda(v)$.
Moreover, if $(u,v) \in E_d$ ($d \in \{\ell,r\}$) 
then there is a DAG-edge $([u]_\equiv,d,[v]_\equiv)$. Intuitively, we merge all vertices of $T$ where isomorphic subtrees are rooted. See \autoref{fig-DagFoldingExample} for an example.

\begin{figure}
\begin{center}
\scalebox{1.4}{\includegraphics{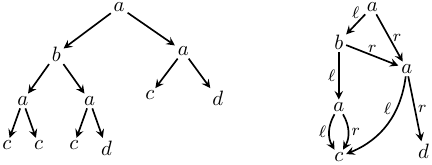}}
\end{center}
\caption{A binary tree $T$ with labels from $\Sigma_0 = \{c,d\}$ and $\Sigma_2 = \{a,b\}$
(left side) and its DAG-folding (right side) with edge labels $\ell$ and $r$ indicating left and right edges. The distinct names of the vertices are omitted for readability, i.e., we only show the labels. Observe that $T$ has the following $6$ distinct subtrees: $c, d, a(cc), a(cd), b(a(cc)a(cd)), a(b(a(cc)a(cd))a(cd))$; thus, its DAG-folding has $6$ vertices.}
\label{fig-DagFoldingExample}
\end{figure}

In the same way, we can define the DAG-folding of ordered trees with larger arity (we just need a larger index set $I$ for the DAG-edges since there can be more than
two edges between two vertices).  The DAG-folding of a syntax tree of an algebraic expression (which is a vertex-labelled ranked tree) is often called an algebraic circuit.
It can be easily seen that there is a one-to-one correspondence between the vertices of a tree $T$ and the paths of its DAG-folding that start in the root (i.e., the unique vertex of in-degree $0$). For example, in \autoref{fig-DagFoldingExample}, the first and second $c$-labelled leaves of $T$ are represented by the paths $a \xrightarrow{\ell} b \xrightarrow{\ell} a \xrightarrow{\ell} c$ and $a \xrightarrow{\ell} b \xrightarrow{\ell} a \xrightarrow{r} c$, respectively, in its DAG-folding.

The DAG-folding $T/_\equiv$ can be seen as a compressed representation of the tree $T$. DAG-compression is used in many different
areas of computer science (see \cite{DowneyST80,FlajoletSS90,FrickGK03} for further details), but it has also limitations. 
For instance, the size of $T/_\equiv$ is always lower
bounded by the height of $T$. In \autoref{sec:FSLP} we will introduce a compression scheme for trees (and forests) that overcomes this
limitation.

For our applications, more important than the DAG-folding is the opposite operation, i.e., the \emph{unfolding} of a DAG. We first define this concept for a DAG $\DAG = (V,E)$ without any further structure, where $E \subseteq V \times I \times V$ for an index set
$I$.  Then the tree $\unfold_{\DAG}(v)$ for a vertex $v$ (the unfolding of $\DAG$ from vertex $v$) has the vertex set
 $\path_{\DAG}(v)$ consisting of all paths 
 that start in $v$ and end in an arbitrary vertex. For two paths $\pi, \pi' \in \path_{\DAG}(v)$,
there is an edge $(\pi, \pi')$ if there are $i \in I$ and $v' \in V$ with $\pi' = \pi i v'$. 
If $\DAG$ is a binary DAG then $\unfold_{\DAG}(v)$ becomes a binary tree in a natural way by 
declaring the edge $(\pi, \pi i v')$ to be a left (resp., right) edge if $i = \ell$ (resp., $i = r$).
In addition, if the DAG $\DAG$ is $\mathcal{C}$-decorated (for a category $\mathcal{C}$) with the decoration mapping $\gamma$ then 
$\unfold_{\DAG}(v)$ becomes a $\mathcal{C}$-decorated tree by extending $\gamma$ to 
$\unfold_{\DAG}(v)$ as follows: $\gamma(\pi) = \gamma(\omega(\pi))$ and 
$\gamma(\pi, \pi i v') =  \gamma(\omega(\pi),i,v')$ (recall that $\omega(\pi)$ is the terminal vertex of $\pi$; see 
\autoref{sec-DAG}).
Similarly, if $\DAG$ is vertex-labelled with the labelling function $\lambda : V \to \Sigma$ then 
$\unfold_{\DAG}(v)$ becomes also vertex-labelled by setting $\lambda(\pi) = \lambda(\omega(\pi))$ for a path $\pi$.

Let $\DAG = (V, E, \lambda)$ be the vertex-labelled binary DAG (without a decoration mapping $\gamma$) on the right of \autoref{fig-DagFoldingExample} and let $v_0$ be its topmost vertex.
 Obviously, according to the above definition, $\unfold_{\DAG}(v_0)$ is the tree $T$ on the left of \autoref{fig-DagFoldingExample}.

\subsection{Straight-line programs for strings}\label{sec:sSLPs}

Let us denote for the moment the concatenation operation for strings over a finite alphabet $\Sigma$ by $\conch$ (usually, we write $uv$ for 
$u \conch v$, but in this section it will be convenient to treat string concatenation more prominently as a binary operation on strings). Then $(\Sigma^*, \conch)$ is an algebraic structure that is usually called the free monoid over $\Sigma$. An expression (or syntax tree) $T$ over $(\Sigma^*, \conch, (a)_{a \in \Sigma})$ (the free monoid with all alphabet symbols added as constants) is a binary tree with leaves labelled by symbols from $\Sigma$ and inner vertices labelled by $\conch$. This expression $T$ naturally evaluates in 
$(\Sigma^*, \conch, (a)_{a \in \Sigma})$ to a string $\valX{T}$. It is the word obtained by traversing the leaves of $T$ from left to right and thereby writing the labels of the leaves.

Usually, string straight-line programs are defined as context-free grammars in Chomsky normal form that derive exactly one string; see \autoref{sec:SLPsBackground}. For our extension to forest straight-line programs the following 
more algebraic definition is useful.
 A \emph{string straight-line program} (s-SLP for short) 
is a (binary vertex-labelled) DAG $\SLP = (V,E,\lambda)$ such that 
$\lambda(A) \in \Sigma$  for every leaf $A$ of $\SLP$ and $\lambda(A) = \conch$ for every non-leaf vertex $A$.
In other words, every tree $\unfold_{\SLP}(A)$ for $A \in V$ is a syntax tree over $(\Sigma^*, \conch, (a)_{a \in \Sigma})$.
We follow here the tradition that vertices of an s-SLP are denoted by capital letters (in the grammar-like definition of an s-SLP
they correspond to non-terminals). We then define for every vertex $A \in V$ the string 
 $\derivsub{\SLP}{A} = \valX{\unfold_{\SLP}(A)}$.

Often in the literature one adds to an s-SLP $\SLP = (V,E,\lambda)$ a distinguished start vertex $A \in V$ (the start non-terminal in the grammar-like
formalization) and
defines $\valX{\SLP} = \valX{\unfold_{\SLP}(A)}$ (the string \emph{represented} or \emph{compressed} by the s-SLP $\SLP$).
We then also say that $\SLP$ is an \emph{s-SLP for $\valX{\SLP}$}. 
For example, let $w = a a b b a b a a b a a b b a b \in \Sigma^*$. Then 
\begin{equation*}
T = ((((a \conch (a \conch b)) \conch b)) \conch (a \conch b)) \conch ((a \conch (a \conch b)) \conch (((a \conch (a \conch b)) \conch b) \conch (a \conch b))) 
\end{equation*}
is an expression over $(\Sigma^*, \conch, a,b)$ with $\valX{T} = w$ (due to associativity, parentheses can be removed, but they determine the structure of $T$, which is obviously crucial for s-SLPs). The syntax tree $T$ is shown on the left of \autoref{fig-strSLPExample}, and the right shows the DAG-folding of $T$ and therefore an s-SLP for $w$.

\begin{figure}
	\centering
	\scalebox{1.4}{\includegraphics{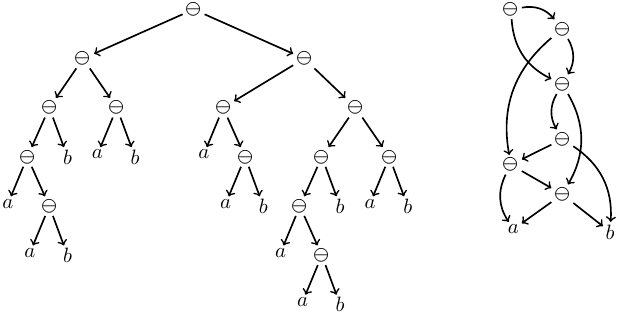}}
	\caption{A syntax tree  over $(\Sigma^*, \conch, \eword)$ (left side) and the corresponding s-SLP (right side). The labels $\ell$ and $r$ for left and right edges are implicitly represented by the drawing of the s-SLP (i.e., the left edge of a vertex is always drawn to the left of the right edge).}
\label{fig-strSLPExample}
\end{figure}

\subsection{Straight-line programs for unranked forests}\label{sec:FSLP}

Straight-line programs for trees or, more generally, forests can be defined in a similar way as for strings. We only have to replace the algebra for strings (i.e., the free monoid over $\Sigma^*$) by a suitable algebra $A$ for forests. We start with the definition of this algebra.

\subsubsection{Forest algebra} \label{sec-FA}

Recall that $\F(\Sigma)$ is the set of all forests with vertex labels from $\Sigma$.
Let us fix a distinguished symbol $\ast \not\in \Sigma$.
The set of forests $F \in \F(\Sigma \cup \{\ast\})$ such that $\ast$ has a unique
occurrence in $F$ and this occurrence is at a leaf vertex is denoted by $\F_{\ast}(\Sigma)$.
Elements of $\F_{\ast}(\Sigma)$ are called \emph{forest contexts}. Note that $\F(\Sigma) \cap \F_{\ast}(\Sigma) = \emptyset$.
Following~\cite{DBLP:conf/birthday/BojanczykW08}, we define the \emph{forest algebra}
as the 2-sorted algebra
$(\F(\Sigma) \cup \F_{\ast}(\Sigma), \conch, \concv, \varepsilon, \ast)$,
where $\varepsilon \in \F(\Sigma)$ is the empty forest, $\ast \in \F_{\ast}(\Sigma)$ is the empty forest context, and
$\conch$ (\emph{horizontal concatenation}) and $\concv$ (\emph{vertical concatenation}) are partially defined binary
operations on $\F(\Sigma) \cup \F_{\ast}(\Sigma)$ that are defined as follows, where we view forests and forest contexts as parenthesised expressions built from the 
binary symbols $\concv$ and $\conch$, and the constants from $\Sigma \cup \{ \ast \}$ (see also \autoref{forests}):
\begin{itemize}
\item For $F_1, F_2 \in \F(\Sigma) \cup \F_{\ast}(\Sigma)$ such that $F_1 \in \F(\Sigma)$ or $F_2 \in \F(\Sigma)$,
we set $F_1 \conch F_2 = F_1 F_2$ (i.e., we concatenate the corresponding sequences of trees).
\item For  $F_1 \in \F_{\ast}(\Sigma)$ and $F_2 \in \F(\Sigma) \cup \F_{\ast}(\Sigma)$, $F_1 \concv F_2$ is obtained by replacing in $F_1$ the unique
occurrence of $\ast$ by  $F_2$.
\end{itemize}
Consider $a(b \ast) \concv (a(bc) \conch b(ccb)) = a(b \ast) \concv a(bc)b(ccb) = a(b a(bc)b(ccb))$ as an example for these operations.

Note that $(\F(\Sigma), \conch)$ and $(\F_{\ast}(\Sigma), \concv)$ are monoids with the neutral elements $\varepsilon$ (the empty forest) and
$\ast$ (the empty forest context) and that $(\Sigma^*, \conch)$ is a submonoid of $(\F(\Sigma), \conch)$.
For $a \in \Sigma$, we write $a_\ast$ for the forest context $a(\ast)$ which consists
of an $a$-labelled root with a single child labelled with $\ast$. Note that $a = a_\ast \concv \varepsilon$.
More generally, for a forest $F$, $a_* \concv F$ yields the tree obtained from the forest $F$ by adding an $a$-labelled
root vertex on top of the forest $F$.
In \cite{DBLP:conf/birthday/BojanczykW08} the forest algebra is introduced as a two sorted algebra with the two sorts $\F(\Sigma)$ and $\F_{\ast}(\Sigma)$.
Our approach with partially defined concatenation operators is equivalent.

A \emph{forest algebra expression} is an expression over the algebra $(\F(\Sigma) \cup \F_{\ast}(\Sigma), \conch, \concv, \varepsilon, \ast)$ with atomic subexpressions
of the form $a$ and $a_\ast$ for $a \in \Sigma$. Such an expression can be identified with a vertex-labelled binary tree, where
every internal vertex is labelled with the operator $\conch$ or $\concv$
and every leaf is labelled with a symbol $a$ or $a_\ast$ for $a \in \Sigma$. 
Not all such trees are valid in the sense that they evaluate to an element from $\F(\Sigma) \cup \F_{\ast}(\Sigma)$, e.g., $a_\ast \conch a_\ast$ is not valid, since it would produce a forest with two occurrences of $\ast$. 
We define  \emph{valid forest algebra expressions} and the \emph{type} $\tau(T) \in \{0,1\}$ of a  valid forest algebra expression inductively as follows:
\begin{itemize}
\item For every $a \in \Sigma$, $a$ and $a_{\ast}$ are valid and $\tau(a)=0$ and $\tau(a_{\ast})=1$.
\item If $T_1$ and $T_2$ are valid and $\tau(T_1) + \tau(T_2) \leq 1$ then $T_1 \conch T_2$ is valid 
and $\tau(T_1 \conch T_2) = \tau(T_1) + \tau(T_2)$.
\item If $T_1$ and $T_2$ are valid and $\tau(T_1)=1$ then $T_1 \concv T_2$ is valid 
and $\tau(T_1 \concv T_2) = \tau(T_2)$.
\end{itemize}
We will only consider valid forest algebra expressions
in the following. We write $\EXP(\Sigma)$ for the set of all valid forest algebra expressions.
Elements of $\EXP(\Sigma)$ will be denoted with $\expr$ in the following. Moreover, we set $\EXP_i(\Sigma) = \{ \expr \in \EXP(\Sigma) : \tau(\expr) = i \}$ for $i \in \{0,1\}$.

With $\valX{\expr} \in \F(\Sigma) \cup \F_\ast(\Sigma)$ we denote the forest or forest context obtained by evaluating $\expr$ in the forest algebra.
If $\tau(\expr) = 0$ then $\valX{\expr} \in \F(\Sigma)$ and if  $\tau(\expr) = 1$ then $\valX{\expr} \in \F_\ast(\Sigma)$.
The empty forest $\varepsilon$ and the empty forest context $\ast$ are not allowed in forest algebra expressions, which is not a restriction as long as we only want to produce non-empty forests and forest contexts; see \cite[Lemma~3.27]{GanardiJL21}.
Note that the vertices of $\valX{\expr}$ can be mapped bijectively on the leaves of $\expr$. 
For the rest of the paper we always assume that the vertices of $\valX{\expr}$ are the leaves of $\expr$.
For a leaf $v$ of $\expr$ we write $\po(v)$ for preorder number of the corresponding vertex of 
$\valX{\expr}$.

\autoref{fig-fa} shows a forest algebra expression $\expr$, which evaluates to the forest on the right of \autoref{fig-fslp}. 
Every vertex of $\deriv{\expr}$ is labelled with its preorder number in blue and the vertex with preorder number $k$ is the leaf
with number $k$ in \autoref{fig-fa}. If $v$ is the leaf of $\expr$ identified by the green path in \autoref{fig-fa}, then
$\po(v) = 14$.

\begin{figure}
	\centering
	\scalebox{1.4}{\includegraphics{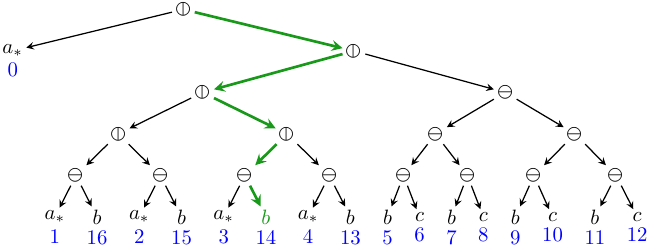}}
	\caption{A forest algebra expression.}
\label{fig-fa}
\end{figure}

\subsubsection{Forest straight-line programs}\label{subsubsection:FSLPs} 

Based on the forest algebra, we can now define forest straight-line programs in an analogous way as we did for string straight-line programs (forest straight-line programs were introduced in~\cite{GasconLMRS20} in a more grammar-like way that is nevertheless equivalent to the following approach). A \emph{forest straight-line program} (over $\Sigma$), f-SLP for short, is a binary DAG $\FSLP = (V, E, \lambda)$  such that $\unfold_{\FSLP}(A) \in \EXP(\Sigma)$ for all $A \in V$. We can assign to every vertex $A \in V$ its type $\tau(A) = \tau(\unfold_{\FSLP}(A))$.
We also say that $A$ is a \emph{forest vertex} if $\tau(A)=0$. Moreover, we define $\derivsub{\FSLP}{A} = \valX{\unfold_{\FSLP}(A)}$. If $\FSLP$ is clear from the context then we write 
$\deriv{A}$ instead of $\derivsub{\FSLP}{A}$.
Note that an s-SLP is an f-SLP $\FSLP = (V, E, \lambda)$ such that $\lambda(A) \in \{ \conch \} \cup \Sigma$ for every $A \in V$.

\autoref{fig-fslp} shows an f-SLP $\FSLP$ on the left, where $A$ is the topmost vertex. Then the forest
 $\deriv{A}$ (actually, a tree) is shown on the right. Every vertex is additionally labelled with its preorder number in blue.
 The forest algebra expression 
$\unfold_{\FSLP}(A)$ is shown in \autoref{fig-fa}. 
The green path $\pi$ in $\FSLP$ (resp., $\unfold_{\FSLP}(A)$) determines the green $b$-labelled vertex in the tree $\deriv{A}$. 
Recall that we identify the path $\pi$ with a leaf of $\unfold_{\FSLP}(A)$ (namely the leaf labelled with 14 in 
\autoref{fig-fa}).

\begin{figure}
\begin{center}
\scalebox{1.4}{\includegraphics{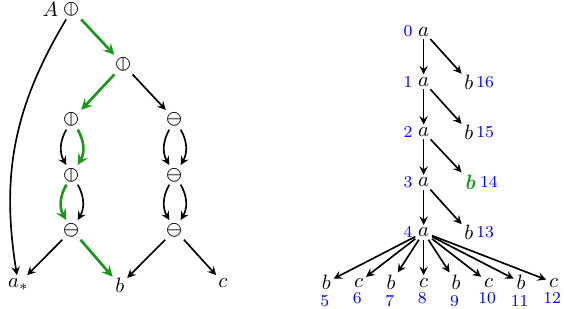}}
\end{center}
\caption{An example f-SLP $\FSLP$ (left side) that describes the tree $\deriv{A}$ on the right side. The forest algebra expression that corresponds to the unfolding $\unfold_{\FSLP}(A)$ is shown in \autoref{fig-fa}.}
\label{fig-fslp}
\end{figure}

Occasionally, (and when the underlying f-SLP $\FSLP = (V,E,\lambda)$ is clear from the context)
we will write $N_A$ for the number of vertices of $\derivsub{\FSLP}{A}$ ($A \in V$).
Since for every vertex $A \in V$ with left child $B$ and right child $C$ we have
$N_A \leq N_B + N_C$, we get the following simple fact:

\begin{lemma} \label{lemma size FSLP}
For every f-SLP $\FSLP = (V,E,\lambda)$ and every vertex $A \in V$ we have $N_A \leq 2^{\card{V}} \leq 2^{\card{\FSLP}}$.
\end{lemma} 

We make the assumption that RAM-algorithms for an f-SLP $\FSLP = (V,E,\lambda)$ have registers 
 of word length $\bigO(\card{\FSLP})$. This is a standard assumption in the area of algorithms for compressed data
 (see for instance also \cite{BilleLRSSW15,GanardiJL21}).
With Lemma~\ref{lemma size FSLP} it allows to store for every vertex $A \in V$ and every vertex $v$ of $\derivsub{\FSLP}{A}$ 
 the preorder number of $v$ in a single register.

We introduced f-SLPs without a distinguished root vertex $S$. In the literature such a root vertex $S$ is often added to
an f-SLP $\FSLP$ and one defines $\deriv{\FSLP} = \derivsub{\FSLP}{S}$.
Let us call such an $\FSLP$ with a distinguished root vertex a \emph{rooted f-SLP}.
Rooted f-SLPs will only appear in \autoref{sec:relabellingUpdates}, all other f-SLPs in this paper are unrooted.

\subsection{Computing preorder numbers}\label{sec:preorderNumbers}

Recall that for a forest vertex $A$ of an f-SLP $\FSLP$ the vertices of the forest $\derivsub{\FSLP}{A}$ correspond to $A$-to-leaf paths in $\FSLP$ (i.e., leaves of $\unfold_{\FSLP}(A)$). Hence, to every such path $\pi$ we assign the preorder number $\po(\pi)$ of the corresponding vertex
of $\derivsub{\FSLP}{A}$.
In the enumeration produced by our algorithm for \autoref{mainResultForests}, every vertex of $\derivsub{\FSLP}{A}$
will be represented by its preorder number (see \autoref{main-preorder} for a more precise version of 
\autoref{mainResultForests}). Hence, it is crucial that, after some preprocessing, we can compute $\po(\pi)$ along the path $\pi$. In this section, we shall discuss the details of this technique, which, to the best knowledge of the authors, has not already been reported in the literature on f-SLPs.

For the further consideration we fix a forest algebra expression $\expr \in \EXP_0(\Sigma)$.
Recall that the vertices of $\deriv{\expr}$ are the leaves of $\expr$ and that for a leaf $v$ of $\expr$,
$\po(v)$ denotes the preorder number of the corresponding node of $\deriv{\expr}$. These are the blue numbers in \autoref{fig-fa}.
 When moving from the root $r$ of $\expr$ down to the leaf $v$ we can compute the preorder number $\po(v)$. 
To see this, consider the path from the root $r$ of the expression tree $\expr$ to the leaf $v$ and let $u$ be the current vertex.
Recall that $\expr(u)$ is the subtree rooted in $u$. In accordance with our notation for f-SLPs we 
write $\valX{u} \in \F(\Sigma) \cup \F_\ast(\Sigma)$ for $\valX{\expr(u)}$ and $\tau(u) \in \{0,1\}$ for $\tau(\expr(u))$ (the type of $u$).
In order to compute $\po(v)$ we have to compute for every $u$ on the path from $r$ to $v$
the so-called \emph{preorder data} $\pre(u)$. It is a single natural number in case $\tau(u)=0$
and it is a pair of natural numbers in
case $\tau(u)=1$.
 The meaning of these numbers is the following, where we identify $\valX{u}$ with the set of vertices of $\valX{\expr}$
 that belong to the subforest (resp., subcontext) $\valX{u}$. If $\pre(u) = x \in \mathbb{N}$ (resp., 
 $\pre(u) = (x,y) \in \mathbb{N} \times \mathbb{N}$) then
\begin{itemize}
\item 
$x$ is the smallest preorder number of the vertices in $\valX{u}$ and
\item $y$ (which only exists if $\tau(u)=1$)
is the size of the subforest of $\valX{\expr}$ that replaces the special symbol $\ast$ in the subcontext $\valX{u}$.
\end{itemize}
For a leaf $v$ of $\expr$ with $\tau(v)=0$
we have $\pre(v)=\po(v)$ and for a leaf $v$ with $\tau(v)=1$
the first component of $\pre(v)$ is $\po(v)$.

In order to compute the preorder data for every vertex of $\expr$, we first have to compute the \emph{leaf size} $s(u)$ and the \emph{left size} $\ell(u)$ for every vertex $u$ of $\expr$:
\begin{itemize}
\item $s(u)$ is the number of leaves in the subexpression rooted in $u$
(where a leaf labelled with $a_\ast$ also counts as one leaf).
This is the same as the number of vertices of $\valX{u}$ without the unique $\ast$-labelled vertex
in case $\tau(u)=1$.
\item $\ell(u)$ is only defined if $\tau(u)=1$ and  is the number of vertices of $\valX{u}$ that are in preorder
smaller than the unique occurrence of the special symbol $\ast$ in $\valX{u}$. Thus, it is the preorder number of $\ast$
in $\valX{u}$.

\end{itemize}
The leaf sizes are computed bottom-up as follows: if $u$ is a leaf of $\expr$ then $s(u)=1$ and if 
$u$ has the left (resp., right) child $v_1$ (resp., $v_2$) then $s(u) = s(v_1)+s(v_2)$.

The left size for a leaf $u$ of $\expr$ with $\tau(u)=1$ (i.e., $u$ is labelled with a symbol $a_\ast$) is $\ell(u)=1$.
Now assume that $u$ has the left (resp., right) child $v_1$ (resp., $v_2$) and that $u$ is labelled with the operator
$\conc \in \{\conch, \concv\}$:
\begin{itemize}
\item Case $\conc = \conch$, $\tau(v_1)=0$, and $\tau(v_2)=1$: 
$\ell(u) = s(v_1) + \ell(v_2)$
\item Case $\conc = \conch$, $\tau(v_1)=1$, and $\tau(v_2)=0$: 
$\ell(u) = \ell(v_1)$
\item Case $\conc = \concv$, and $\tau(v_1)=\tau(v_2)=1$: 
$\ell(u) = \ell(v_1) + \ell(v_2)$
\end{itemize}
Finally, the preorder data are computed top-down as follows: For the root vertex $r$ of $\expr$ we set $\pre(r) = 0$ (recall that $\tau(\expr)=0$).
Now assume that $u$ is an internal vertex of $\expr$ with  left (resp., right) child $v_1$ (resp., $v_2$).
Moreover, let $u$ be labelled with the operator $\conc \in \{\conch, \concv\}$ and let
$\pre(u) = x$ in case $\tau(u)=0$ and 
$\pre(u) = (x,y)$ in case $\tau(u)=1$:
\begin{enumerate}[label=Case (\arabic*)\,\,, leftmargin=2.1cm]
\item $\conc = \conch$,  and $\tau(v_1)=\tau(v_2)=0$: 
$\pre(v_1) = x \text{ and } \pre(v_2) = x+s(v_1)$
\item $\conc = \conch$,  $\tau(v_1)=0$, and $\tau(v_2)=1$:
$\pre(v_1) = x \text{ and }  \pre(v_2) = (x+s(v_1), y)$
\item $\conc = \conch$, $\tau(v_1)=1$, and $\tau(v_2)=0$: 
$\pre(v_1) = (x,y) \text{ and }  \pre(v_2) = x+s(v_1) + y$
\item $\conc = \concv$, $\tau(v_1)=1$, and $\tau(v_2)=0$: 
$\pre(v_1) = (x,s(v_2)) \text{ and }  \pre(v_2) = x+\ell(v_1)$
\item $\conc = \concv$ and $\tau(v_1)=\tau(v_2)=1$:
$\pre(v_1) = (x,y+s(v_2)) \text{ and }  \pre(v_2) = (x+\ell(v_1),y)$
\end{enumerate}
Every edge $e = (u,v)$ in the expression tree $\expr$ can be labelled with the function $f_e$
that represents the effect on the preorder data. We call this function the \emph{preorder effect} of the edge $e$.
It depends on the operator computed in vertex $u$, on the types $\tau(u)$, $\tau(v)$, the values $s(v)$, $\ell(v)$, $s(v')$ and $\ell(v')$  (where $v'$ is the other child of $u$),
and whether $v$ is the left or right child of $u$. 
If $\tau(u)=i$ and $\tau(v)=j$ for $i,j \in \{0,1\}$ then $f_e : \mathbb{N}^{i+1} \to \mathbb{N}^{j+1}$.
The functions $f_e$ are defined in \autoref{fig-effect-pre} (note that these functions correspond to the $5$ cases from above that show how the preorder data is computed).

\begin{figure}
    \begin{center}
    	\scalebox{1.4}{\includegraphics{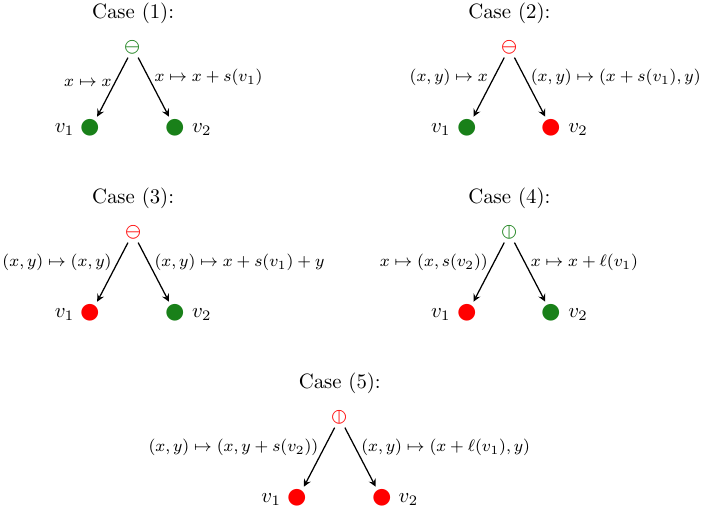}}
    \end{center}
    \caption{The effect of edges on the preorder data. Green (resp., red) vertices have type $0$ (resp., $1$).}
	\label{fig-effect-pre}
\end{figure}

\begin{example} \label{ex-preorder-effect}
In \autoref{fig-decorated-expression-tree} we reconsider the forest algebra expression from  \autoref{fig-fa}.
Green (resp., red) vertices have type $0$ (resp., $1$).
Every vertex is labelled 
with its leaf size (in black) and its left size (in red).
The edges along the green path are labelled with their preorder effects.
By composing these preorder effects one obtains the total preorder effect
$x \mapsto x+14$ for the whole path.
\end{example}

\begin{figure}
	\centering
	\scalebox{1.4}{\includegraphics{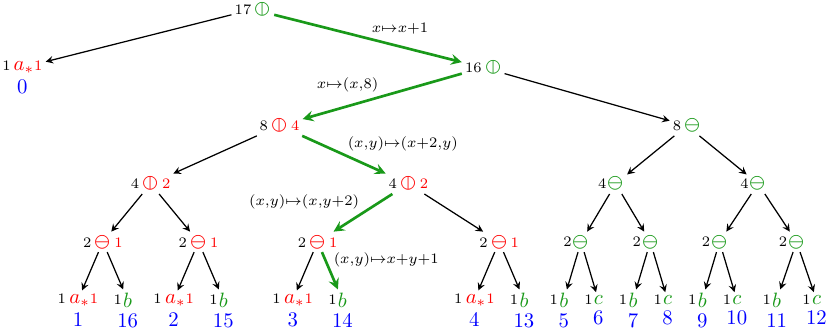}}
	\caption{The forest algebra expression $\expr$ from \autoref{fig-fa}, where
	 green (resp., red) vertices have type $0$ (resp., $1$).
	Every vertex is in addition labelled with its leaf size (in black on the left of the vertex), and its left size (in red on the right of the vertex, if the vertex has type $1$). Leaves are labelled with the preorder numbers of the corresponding vertices in $\deriv{\expr}$ in blue.}
	\label{fig-decorated-expression-tree}
\end{figure}

If we take the closure of all possible preorder effects (for all possible integer constants $s(v_1)$, $s(v_2)$, $\ell(v_1)$ that appear in 
\autoref{fig-effect-pre}) under function composition, we obtain a category $\mathcal{C}_{\text{pre}}$ with two objects $0$ and $1$ (the possible types of vertices in $\expr$). The 
set of morphisms  $M_{i,j}$ from object $i$ to object $j$ ($i,j \in \{0,1\}$) are the affine functions from $\mathbb{N}^{i+1}$ to $\mathbb{N}^{j+1}$
of the following form,  where $c,d \in \mathbb{N}$ are constants and $x,y \in \mathbb{N}$ are the arguments:
\begin{alignat}{4}
& M_{0,0}                     &                &  M_{0,1}             &                    &  M_{1,0}            &                      &  M_{1,1}  \nonumber \\
x     & \mapsto x+c       & \qquad x & \mapsto (x+c,d) & \qquad (x,y) & \mapsto x+c      & \qquad  (x,y) & \mapsto (x+c,y+d) \label{preorder-morph1}\\
       &                           &               &                            & \qquad (x,y) & \mapsto x+y+c  & \qquad  (x,y) & \mapsto (x+y+c,d) \label{preorder-morph2}
\end{alignat}
Note that the definition of the leaf size $s(\cdot)$ and left size $\ell(\cdot)$ make also sense for vertices of an f-SLP $\FSLP$ (since these values are computed bottom-up). More precisely, if $\FSLP$ is an f-SLP and $B$ is a vertex of $\FSLP$ then 
$s(B)$ is the number of vertices of $\deriv{B}$ (without the unique $\ast$-labelled vertex) and $\ell(B)$ is the 
preorder number of the unique $\ast$-labelled vertex in $\deriv{B}$  (which is only defined if $\tau(B)=1$).
Moreover, for an edge  $e = (B,i,C)$ of $\FSLP$, we can define the function $f_e$ in the same way as the function 
$f_e$ for an edge $e = (u,v)$ of a forest algebra expression. The reason is that this function $f_e$ only depends on the types $\tau(u)$, $\tau(v)$, the values $s(v)$, $\ell(v)$, $s(v')$ and $\ell(v')$  (where $v'$ is the other child of $u$), and whether $v$ is the left or right child of $u$. 
Hence, we obtain a $\mathcal{C}_{\text{pre}}$-decorated f-SLP with the decoration mapping $\gamma(e) = f_e$ for every edge $e$.
This will be the decoration of f-SLPs that we always choose implicitly in the rest of the paper.

Recall from the beginning of this section that our goal is to compute for 
 a forest vertex $A$ of an f-SLP $\FSLP$  and an $A$-to-leaf path $\pi$ in $\FSLP$
  the preorder number $\po(\pi)$ while walking along the path $\pi$. 
This problem is subsumed by computing the affine function $\gamma(\pi)$ (which belongs to
$M_{0,0} \cup M_{0,1}$ since $A$ is a forest vertex) in the 
$\mathcal{C}_{\text{pre}}$-decorated f-SLP $\FSLP$. As can be seen from above, this function
$\gamma(\pi)$ is a mapping of the form $x \mapsto x+c$ or
$x \mapsto (x+c,d)$. In both cases, the  preorder number $\po(\pi)$ is $c$.
Note that the green path in \autoref{fig-decorated-expression-tree} leads to the leaf with preorder number 14 and indeed
(as remarked in Example~\ref{ex-preorder-effect}) the total preorder effect of the green path is $x \mapsto x+14$.

We have to argue that \autoref{assumption-category} is satisfied for an input f-SLP $\FSLP$ (viewed as a $\mathcal{C}_{\text{pre}}$-decorated DAG with the above decoration mapping $\gamma$). Recall from \autoref{subsubsection:FSLPs} that we assume a RAM with register length $\bigO(\card{\FSLP})$. This clearly allows to store vertices of $\FSLP$ in a single RAM register.
Consider now a path $\pi$ in the f-SLP $\FSLP$ and the corresponding affine function $\gamma(\pi)$. It can be stored by at most two 
natural numbers (the constants $c$ and $d$ in \eqref{preorder-morph1} and \eqref{preorder-morph2}).
We claim that these numbers are bounded by $2^{\bigO(\card{\FSLP})}$ and therefore have bit length
$\bigO(\card{\FSLP})$. First, recall that the constants appearing in the functions $f_e$ for edges $e$ of 
$\FSLP$ are the numbers $s(B)$ and $\ell(B)$ (for $B$ a vertex of $\FSLP$).
These numbers are bounded by $2^{\card{\FSLP}}$ by Lemma~\ref{lemma size FSLP}. It is easy to see that the 
constants appearing in $\gamma(\pi)$ are sums of at most
$|\pi| \le \card{\FSLP}$ many of these numbers $s(B)$ and $\ell(B)$. Hence, they are bounded by $\card{\FSLP} \cdot
2^{\card{\FSLP}} \leq 2^{2\card{\FSLP}}$. Finally, computing the composition of two preorder effects 
$f$ and $g$ involves a constant number of additions. This shows that
\autoref{assumption-category} is satisfied.

\section{MSO Logic and Automata over Trees and Forests} \label{sec-MSO-aut}

In \autoref{sec-fslp} we talked about the compressed representation of the queried forest.
In this section we discuss the formalisms for representing queries over forests.

\subsection{Monadic second order logic}  \label{sec-MSO}

We consider formulas of \emph{monadic second order logic} (MSO) \cite{CouEng2012} that are interpreted over forests $F = (V, E, R, \lambda)$ as described in \autoref{forests}.
Since any first-order variable $x$ (that takes elements from the universe as values) can be replaced by a set variable $X$ (an MSO-formula can express that $X$ is a singleton set), we can restrict to MSO-formulas where all free variables are set variables. 

If $\Psi$ is an MSO-formula over the signature of unranked forests with free set variables $X_1, \ldots, X_k$ (written as $\Psi(X_1, \ldots, X_k)$) and $S_1, \ldots, S_k \subseteq V$ are vertex sets of some forest $F = (V, E, R, \lambda)$, then we write $(F,S_1, \ldots, S_k) \models \Psi$ if the formula $\Psi$ holds in the forest $F$ if the variable $X_i$ is set to $S_i$. Hence, we can interpret MSO-formulas $\Psi(X_1, \ldots, X_k)$ as \emph{MSO-queries} that, for a given forest $F$, define a \emph{result set} $\Psi[F] = \{(S_1, \ldots, S_k) : (F,S_1, \ldots, S_k) \models \Psi\}$. 

To make the exposition less technical, we further restrict to MSO-formulas with a single free set variable $X$. This is a common restriction that can be found elsewhere in the literature (see, e.g.,~\cite{Bagan09, FrickGK03}) and it is without loss of generality. For the sake of a self-contained exposition, we will show next how to reduce the number of set variables in MSO-formulas to one. 

Consider an MSO-formula $\Psi(X_1, \ldots, X_k)$ with $k \ge 2$ and let $F = (V, E, R, \lambda)$ be a forest with vertex labels from $\Sigma$.
We then take the new set of vertex labels $\Sigma' = \Sigma \times [k]$ and define the new forest $F' = (V',E',R',\lambda')$ with 
\begin{itemize}
\item $V' = V \times [k]$, 
\item $E' = \{ ((u,k),  (v,i)) : (u,v) \in E, i \in [k] \}$,
\item $R' = \{ ((u,i), (u,i+1)) : u \in V, i \in [k-1] \} \cup \{ ((u,k), (v,1)) : (u,v) \in R \}$, and
\item $\lambda'(u,i) = (\lambda(u),i)$.
\end{itemize}
So, intuitively, we add to each vertex $k-1$ siblings on the left; the original vertex $u$ corresponds to the vertex $(u,k)$.

It  is then straightforward to translate $\Psi(X_1, \ldots, X_k)$ into an MSO-formula $\Psi'(X)$ with a single free set variable such that 
for all subsets $S_1, \ldots, S_k \subseteq V$ we have 
$(F,S_1, \ldots, S_k) \models \Psi$ if and only if $(F',S') \models \Psi'$, where
$S' = \{ (v,i) : v \in S_i, i \in [k] \}$.
Hence, an algorithm that enumerates the set 
$\Psi'[F']$  with output-linear delay directly yields an algorithm that 
enumerates the set $\Psi[F]$ with output-linear delay. This is also true if we insist that forest vertices are represented by their preorder numbers:
If the algorithm for enumerating $\Psi'[F']$ outputs the set $S'$, where
every $m \in S'$ is a preorder number in $F'$, then the algorithm for enumerating $\Psi[F]$ outputs
the tuple $(S_0, \ldots, S_{k-1})$, where $S_i = \{ \lfloor m/k \rfloor : m \in S', m \bmod k = i \}$.

As mentioned before, our task is to enumerate the result set of an MSO-query on an f-SLP-compressed forest.
Consequently, we have to explain how the above reduction to the one-variable case can be done on the level of f-SLPs.
For this we will add additional vertices to an f-SLP $\FSLP = (V,E,\lambda)$ such that the resulting f-SLP $\FSLP'$ has the following property
for every forest vertex $A \in V$: $\derivsub{\FSLP'}{A}$ is the forest obtained by applying the above transformation $F \mapsto F'$
 to $F=\derivsub{\FSLP}{A}$. Moreover, the size of $\FSLP'$ should be linear in the size of $\FSLP$. This, however, is easy:
 We can assume that in $\FSLP$ there is for every $a \in \Sigma$ exactly one vertex 
 $A_a$ labelled with $a$ and exactly one vertex $A_{a_\ast}$ labelled with  $a_\ast$. We obtain $\FSLP'$
 by adding in total $\bigO(|\Sigma| \cdot k)$ many new vertices such that $\derivsub{\FSLP'}{A_a} = (a,1) (a,2) \cdots (a,k)$ 
 and $\derivsub{\FSLP'}{A_{a_\ast}} = (a,1) \cdots (a,k-1) (a,k)_\ast$.
 This yields an f-SLP $\FSLP'$ of size $\bigO(\card{\FSLP})$ with the above property.

\subsection{Tree automata} \label{sec-ta}

We consider two types of tree automata: deterministic bottom-up tree automata that work on binary trees and nondeterministic
stepwise tree automata that work on general unranked forests. Since they should implement queries on trees and forests, they will be interpreted as selecting vertices from trees or forests (this aspect is explained in more detail later on).

\subsubsection{Deterministic bottom-up tree automata}\label{sec:dBUTA}

A \emph{deterministic bottom-up tree automaton} (over the alphabets $\Sigma_0$ and $\Sigma_2$) is a 6-tuple $\mathcal{B} = (Q, \Sigma_0, \Sigma_2, \delta_0, \delta_2, Q_f)$, where $Q$ is a finite set of states, $\Sigma_0$ is the set of leaf vertex labels, $\Sigma_2$ is the set of labels for internal vertices, $Q_f \subseteq Q$ is the set of final states, $\delta_0 : \Sigma_0 \to Q$ assigns states to leaves of a tree, and $\delta_2 : Q \times Q \times \Sigma_2 \to Q$ assigns states to internal vertices depending on the vertex label and the states of the two children. For a given binary tree $T$ we define the state $\mathcal{B}(T)$ as the unique state to which $\mathcal{B}$ evaluates the tree $T$. It is inductively defined as follows, where $a \in \Sigma_0$ and  $f \in \Sigma_2$:
\begin{itemize}
\item $\mathcal{B}(a) = \delta_0(a)$ and
\item $\mathcal{B}(f(T_1,T_2)) = \delta_2( \mathcal{B}(T_1), \mathcal{B}(T_2), f)$ for binary trees $T_1$ and $T_2$.
\end{itemize}
The binary tree $T$ is accepted by $\mathcal{B}$ if and only if $\mathcal{B}(T) \in Q_f$. With $L(\mathcal{B})$ we denote the set of binary trees accepted by $\mathcal{B}$. We use the acronym dBUTA for deterministic bottom-up tree automaton.

As an example, consider a dBUTA $\mathcal{B}_\tau$ with $\Sigma_0 = \{a, a_\ast : a \in \Sigma\}$, $\Sigma_2 = \{\conch, \concv\}$, $Q = \{0,1,\mathsf{failure}\}$, $Q_f = \{0, 1\}$ and the transition functions $\delta_0$ and $\delta_2$ are defined as follows: 
\begin{align*}
& \delta_0(a) = 0, \delta_0(a_\ast) = 1 \\
& \delta_2(0,0,\conch) = 0,  \delta_2(0,1,\conch) = \delta_2(1,0,\conch) = 1, \delta_2(1,1,\conch) = \mathsf{failure} \\
& \delta_2(1,0,\concv) = 0,  \delta_2(1,1,\concv) = 1, \delta_2(0,0,\concv) = \delta_2(0,1,\concv) = \mathsf{failure} \\
& \delta_2(p,q,\concv) = \mathsf{failure} \text{ if }  p = \mathsf{failure} \text{ or } q = \mathsf{failure}
\end{align*}
Obviously, $\mathcal{B}_\tau$ accepts exactly the set of valid forest algebra expressions.

\subsubsection{Stepwise tree automata} 

Stepwise tree automata are an automaton model for forests that is equivalent to MSO-logic \cite{CarmeNT04}.\footnote{Stepwise tree automata are defined in such a way that they can run on forests, so they should be called stepwise forest automata, but we prefer to use the existing terminology.} We follow the definition from \cite{MMN22}.
A \emph{nondeterministic stepwise tree automaton} (nSTA for short) over the alphabet $\Sigma$ 
is a tuple $\mathcal{A} = (Q, \Sigma, \delta, \iota, q_0, q_f)$ with the following
properties:
\begin{itemize}
\item $Q$ is a finite set of states,
\item $\delta \subseteq Q \times Q \times Q$ is the transition relation,
\item $\iota : \Sigma \to 2^Q$ assigns a set of local initial states to each alphabet symbol,
\item $q_0$ is the global initial state, and
\item $q_f$ is the global final state.
\end{itemize}
Let $F = (V, E, R, \lambda) \in \F(\Sigma)$ be a forest with root vertices $v_1, \ldots, v_k \in V$, where $v_1$ is the left-most root and $v_k$ 
is the right-most root. If $F$ is a tree,
we have $v_1 = v_k$. For states $q_1, q_2 \in Q$, a $(q_1,q_2)$-run of $\mathcal{A}$ on the forest $F$ is given by three
mappings $\rho_0 : V \to Q$ (called $\lambda_{\mathsf{pre}}$ in \cite{MMN22}), $\rho_1 : V \to Q$ (called $\lambda_{\mathsf{self}}$ in \cite{MMN22}), and $\rho_f : V \to Q$ 
(called $\lambda_{\mathsf{post}}$ in \cite{MMN22}) such that the following conditions hold (an intuitive explanation follows below):
\begin{itemize}
\item $\rho_0(v_1) = q_1$,
\item $\rho_f(v_k) = q_2$,
\item $\rho_0(v) \in \iota(\lambda(u))$ if $v$ is the first child of $u$,
\item $\rho_0(v) = \rho_f(u)$ if $u$ is the left sibling of $v$ (this includes the case where $u = v_i$ and $v = v_{i+1}$ for some $1 \leq i \leq k-1$),
\item $\rho_1(v) \in \iota(\lambda(v))$ if $v$ is a leaf, 
\item $\rho_1(v) = \rho_f(u)$ if $u$ is the last child of $v$, and
\item $(\rho_0(v), \rho_1(v), \rho_f(v)) \in \delta$ for all vertices $v$.
\end{itemize}
A forest $F$ is accepted by $\mathcal{A}$ if its has a $(q_0, q_f)$-run. With $L(\mathcal{A})$ we denote the set of forests accepted by $\mathcal{A}$.

Let us explain this model on an intuitive level. To this end, we first observe that all the triples $(q, r, s)$ from $\delta$ can be interpreted as a string automaton $M_{\mathcal{A}}$ over the alphabet $Q$, i.e., $(q, r, s)$ means that we can change from $q$ to $s$ by reading $r$. Now an nSTA processes an unranked tree by treating each sibling-sequence as a string that is then processed by this string automaton $M_{\mathcal{A}}$. More precisely, let $v_1, v_2, \ldots, v_k$ be the children (ordered from left to right) of some vertex $u$. Each sibling $v_i$ gets some state $\rho_1(v_i)$, which is either from $\iota(\lambda(v_i))$ if $v_i$ is a leaf, or it is propagated from $v_i$'s last child $w_i$ via the condition $\rho_1(v_i) = \rho_f(w_i)$. Now we read the string $\rho_1(v_1) \rho_2(v_2) \ldots \rho_k(v_k)$ with $M_{\mathcal{A}}$ as follows. We start in some state $\rho_0(v_1) \in \iota(\lambda(u))$ (recall that $u$ is $v_1$'s parent). Then reading $\rho_1(v_1)$ changes the state from $\rho_0(v_1)$ to $\rho_f(v_1)$ ($= \rho_0(v_2)$), reading $\rho_1(v_2)$ changes the state from $\rho_0(v_2)$ to $\rho_f(v_2)$ ($= \rho_0(v_3)$), reading $\rho_1(v_3)$ changes the state from $\rho_0(v_3)$ to $\rho_f(v_3)$ ($= \rho_0(v_4)$) and so on until we reach the state $\rho_f(v_k)$, which then serves as the state associated to the parent vertex $u$ of the siblings $v_i$, i.e., $\rho_f(v_k) = \rho_1(u)$. 

Consequently, a computation of an nSTA can be seen as a preorder traversal of the trees of the forest: We read a sequence of siblings as described above until we reach a sibling $v$ that is not a leaf. Then we go down one step and process $v$'s children and so on. Whenever we finish reading a sequence of siblings, we have determined the $\rho_1(\cdot)$-state of its parent vertex $u$ and we can therefore continue reading the sequence of siblings that contains $u$ and so on until we end up with the $\rho_f(\cdot)$-state of the root of the rightmost tree of the forest. 

In this work, stepwise tree automata will only serve as an intermediate model that, in the process of our algorithm, will be transformed into a deterministic bottom-up tree automaton. Hence, we will not further discuss the model of nSTA and refer the reader to~\cite{MMN22} for further details and explanations.

\subsubsection{Representing MSO-queries by tree automata}\label{sec:representingTreeAutomata}

Nondeterministic stepwise tree automata can represent queries on forests as follows. 
For a forest $F$ and a subset $S$ of its vertices, we identify the pair $(F,S)$ with the forest that is obtained from $F$ by relabelling every $a$-labelled vertex $v$ of $F$ ($a \in \Sigma$) with $(a, \beta) \in \Sigma \times \{0,1\}$, where $\beta = 1$ if and only if $v \in S$. Intuitively, $(F,S)$ represents the forest $F$ from which the vertices in $S$ have been selected (or the forest $F$ together with a possible query result $S$). 
Our nSTAs become vertex-selecting, by taking $\Sigma \times \{0,1\}$ as the set of vertex labels. Such an nSTA $\mathcal{A}$ selects the vertex set $S$ from a forest $F \in \F(\Sigma)$ if and only if
$(F,S) \in L(\mathcal{A})$.

Our dBUTAs only need the ability to select leaves of binary trees, which means that we define them over the alphabets $\Sigma_0 \times \{0, 1\}$ (for leaf vertices) and $\Sigma_2$ (for internal vertices), i.e., we run them on pairs $(T,S)$, where $T$ is a binary tree and $S \subseteq \leaves(T)$.

In the following, we assume that all nSTAs and dBUTAs are vertex-selecting in the  above sense.
For the forest $F = (V,E,R,\lambda)$ and an nSTA $\mathcal{A}$
we write 
\[ \select(\mathcal{A},F) = \{ S \subseteq V :  (F,S) \in L(\mathcal{A}) \}\] 
for the set of vertex sets selected by the nSTA $\mathcal{A}$.
Similarly, for a binary tree $T$ and a dBUTA $\mathcal{B}$ we define
\[ \select(\mathcal{B},T) = \{ S \subseteq \leaves(T) :  (T,S) \in L(\mathcal{B}) \}.
\]
It is known that MSO-formulas (that are interpreted over forests) can be translated into equivalent automata (and vice versa). 
More precisely, we use the following well-known fact:

\begin{theorem}[cf.~\cite{CarmeNT04}] \label{thm-MSO->automata}
From an MSO-formula $\Psi(X)$ one can construct an nSTA
 $\mathcal{A}$
such that for every forest $F \in \F(\Sigma)$ with vertex set $V$ we have 
$\select(\mathcal{A},F) = \Psi[F]$.
\end{theorem}
Our main goal is to enumerate the set of query results $\Psi[F]$.
By \autoref{thm-MSO->automata} this is equivalent to enumerating all $S$ such that 
$(F,S)$ is accepted by an nSTA. 
We will therefore ignore MSO-logic in the following and directly start from an nSTA.

\section{Enumerating MSO Queries Over Forest SLP} \label{sec-main-result}

In this section, we present our enumeration algorithm for MSO-queries on f-SLP-compressed forests. In particular, we provide a proof for \autoref{mainResultForests}. As already mentioned before \autoref{mainResultForests} we assume the data complexity
setting, where the size of the query $\Psi$ is assumed to be constant.

In \autoref{sec-MSO}, we have already discussed that we can restrict to MSO-queries with only one free set variable. Moreover, \autoref{thm-MSO->automata} mentioned in \autoref{sec:representingTreeAutomata} means that instead of an MSO-query, we can directly start with an nSTA $\mathcal{A}$. The size of the latter can be non-elementary in the size of the MSO-query $\Psi$ but in the data
complexity setting this is not an issue.
In summary, this means that in order to prove \autoref{mainResultForests} it is sufficient to prove the following theorem.

\begin{theorem} \label{main-preorder}
From an nSTA $\mathcal{A}$ over the alphabet $\Sigma \times \{0,1\}$  with $m$ states and an f-SLP $\FSLP$ one can compute in preprocessing time $\card{\FSLP} \cdot 2^{\bigO(m^4)}$ a data structure that allows to enumerate for a given forest vertex $A$ of $\FSLP$ the set $\select(\mathcal{A},\derivsub{\FSLP}{A})$ with output-linear delay. In the enumeration, every vertex of $\derivsub{\FSLP}{A}$ is represented by its preorder number.
\end{theorem}
Strictly speaking, for a set $S \in \select(\mathcal{A},\derivsub{\FSLP}{A})$ the delay in \autoref{main-preorder} is $f(m) \cdot |S|$  for 
a function $f$ ($f$ can be bounded by $\bigO(m^4)$ if states of $\mathcal{A}$ fit into $\bigO(1)$ registers), but, as already remarked above,
$m$ is considered to be a constant in the data complexity setting.

Compared to \autoref{mainResultForests}, the formulation of \autoref{main-preorder} is more general, since it does not only apply to the forest described by the whole f-SLP, but to every forest described by any forest vertex $A$, i.e., to every forest $\derivsub{\FSLP}{A}$, where $A$ is a forest vertex of the f-SLP. In this setting, one should view the f-SLP $\FSLP$ as the specification of a collection of forests
$\derivsub{\FSLP}{A}$ (for every  forest vertex $A$ of $\FSLP$). After the preprocessing the user can choose for which of these forests
the query results shall be enumerated.
We stress the fact that the vertex $A$ in \autoref{mainResultForests} is not known during the preprocessing phase, i.e., the data structure computed in the preprocessing enables enumeration for every  given forest vertex $A$ of $\FSLP$.

Next, we will argue that in order to prove \autoref{main-preorder}, it is sufficient to prove \autoref{main2} below, which talks about path enumeration for arbitrary  vertex-labelled DAGs
with decorations from a category $\mathcal{C}$.

The input data of our MSO-evaluation problem can be interpreted on three different levels: (1) the actual f-SLP $\FSLP$  together with a 
forest vertex $A$ (see left side of \autoref{fig-fslp}), (2) the forest algebra expression 
$\expr = \unfold_{\FSLP}(A) \in \EXP_0(\Sigma)$ (see \autoref{fig-fa}), and (3) the unranked forest $\valX{\expr} = \derivsub{\FSLP}{A}$ that $\expr$ evaluates to (see right side of \autoref{fig-fslp}). Obviously, (1) is our actual input, while (3) is the structure on which the MSO-query (in form of an nSTA) should be evaluated. Let us ignore perspective (1) for a moment and focus on the forest algebra expression $\expr$ and the unranked forest $\valX{\expr}$ that it describes (that these objects are not explicitly given and that we cannot afford to construct them is a problem we have to deal with later). 

Our task is to enumerate all those vertex-sets $S$ of $\valX{\expr}$ such that $(\valX{\expr}, S) \in L(\mathcal{A})$ (recall that $(\valX{\expr}, S)$ is the variant of $\valX{\expr}$ with all vertices from $S$ being marked). By definition, the vertices of $S$ uniquely correspond to some leaf-set $S'$ of $\expr$. Hence, we can define $(\expr, S') \in \EXP_0(\Sigma \cup \Sigma \times \{0,1\})$ as the forest algebra expression obtained from $\expr$ by marking the leaves from $S'$ with $1$ and all other leaves with $0$. More precisely, every leaf labelled with some $a \in \Sigma$ is relabelled to $(a, 1)$ or $(a, 0)$ depending on whether or not it is from $S'$, and every leaf labelled with $a_\ast$ for some $a \in \Sigma$ is relabelled to $(a, 1)_\ast$ or $(a, 0)_\ast$ depending on whether or not it is from $S'$.
In particular, we have that $\valX{(\expr, S')} = (\valX{\expr}, S)$. Our task is therefore to enumerate all leaf-sets $S'$ of $\expr$ such that $\valX{(\expr, S')} \in L(\mathcal{A})$. Thereby, for every leaf $v$ of $\expr$ we output the preorder number $\po(v)$ of the vertex $v$ in $\valX{\expr}$.

In order to enumerate all leaf-sets $S'$ of $\expr$ such that $\valX{(\expr,S')} \in L(\mathcal{A})$,
we can actually turn $\mathcal{A}$ into a dBUTA $\mathcal{B}$ such that for every forest algebra expression $\expr$ and every subset $S'\subseteq \leaves(\expr)$ we have: $\valX{(\expr,S')} \in L(\mathcal{A})$ if and only if $(\expr,S') \in L(\mathcal{B})$. This means that our problem reduces to the following general enumeration problem: Given a dBUTA $\mathcal{B}$ and a vertex-labelled binary tree $T$, enumerate all leaf-sets $S$ of $T$ such that $(T, S) \in L(\mathcal{B})$. 
The required dBUTA $\mathcal{B}$ can be computed from $\mathcal{A}$ by known techniques (see~\cite{MMN22}). For completeness, we give a full proof of the following result in Appendix~\ref{sec:proofOfthm-MNN}. Let $\Sigma_0 = \{ a,a_\ast : a \in \Sigma\}$ and $\Sigma_2 = \{\conch, \concv \}$.

\begin{theorem}[cf.~\cite{MMN22}] \label{thm-MNN}
From an nSTA $\mathcal{A}$ over $\Sigma$ with $m$ states one can construct
a dBUTA  $\mathcal{B}$ over $\Sigma_0, \Sigma_2$ with $2^{m^2} + 2^{m^4} +1$ states
such that $L(\mathcal{B}) = \{ \expr \in \EXP(\Sigma) : \valX{\expr} \in L(\mathcal{A})\}$.
\end{theorem}

As explained above, with \autoref{thm-MNN}, our problem reduces to enumerating all leaf-sets $S$ of a vertex-labelled binary tree $T$ such that $(T, S) \in L(\mathcal{B})$, where $\mathcal{B}$ is a given dBUTA. For this problem, we could now use one of the existing algorithms from the literature, for example Bagan's algorithm from~\cite{Bagan06}. However, now is the time to remember that above we have ignored the fact that we do not have the explicit vertex-labelled binary tree $T$ but only a DAG $\DAG$ together with a vertex $v_0$ such that $T = \unfold_{\DAG}(v_0)$. In our setting, $\DAG$ is an f-SLP $\FSLP$ and $v_0=A$ is a vertex
of $\FSLP$.

Let us fix an f-SLP $\FSLP$ and a forest vertex $A$.
The vertices of the forest $\derivsub{\FSLP}{A}$ (or, equivalently, the leaves of $\unfold_{\FSLP}(A)$)
are the $A$-to-leaf paths in the DAG $\FSLP$. Moreover, every such path $\pi$ will be represented in the output of the enumeration
phase by the preorder number $\po(\pi)$ of the corresponding vertex of $\derivsub{\FSLP}{A}$.
We have seen in \autoref{sec:preorderNumbers} that $\po(\pi)$ can be directly obtained from $\gamma(\pi)$, where $\gamma$
is the $\mathcal{C}_{\text{pre}}$-decoration of the f-SLP $\FSLP$ from
\autoref{sec:preorderNumbers}. 
This means that it suffices to prove the following result that is stated for an arbitrary category $\mathcal{C}$ (instead of 
$\mathcal{C}_{\text{pre}}$):

\begin{theorem} \label{main2}
Fix a category $\mathcal{C}$. From a dBUTA  $\mathcal{B}$ with $m$ states and a vertex-labelled $\mathcal{C}$-decorated binary DAG $\DAG = (V, E, \lambda, \gamma)$ that satisfies \autoref{assumption-category}, one can compute in preprocessing time $\bigO(|\DAG| \cdot m^2)$ a data structure that allows to enumerate for a given vertex $v_0 \in V$ the multiset $\multiset{\gamma^*(S) : S \in \select(\mathcal{B}, \unfold_{\DAG}(v_0))}$ with output-linear delay.\footnote{Strictly speaking, output-linear delay only holds under the natural assumption that also states of $\mathcal{B}$ fit
into a constant number of RAM registers. In our application of \autoref{main2} this is clearly satisfied since $m$ only depends on the MSO-query $\Psi$, whose size is a constant with respect to data complexity.} 
\end{theorem}

Let us explain some aspects of this theorem. First, recall that 
in \autoref{sec-trees} we defined $\gamma^*(S) = \multiset{ \gamma^*(v) : v \in S}$, where the $\mathcal{C}$-morphism $\gamma^*(v)$ of a leaf $v$ of $\unfold_{\DAG}(v_0)$ is just the $\mathcal{C}$-morphism defined by the unique $v_0$-to-$v$ path  in $\unfold_{\DAG}(v_0)$ (which is the composition of the morphisms assigned to the edges along the path). 
In \autoref{main2}, we talk about enumerating the \emph{multiset} $\multiset{\gamma^*(S) : S \in \select(\mathcal{B}, \unfold_{\DAG}(v_0))}$ (see the end of \autoref{sec-enumeration-general} for our formalisation of enumeration with duplicates), since for a general decoration function $\gamma$, there may exist different sets $S_1, S_2 \in \select(\mathcal{B}, \unfold_{\DAG}(v_0))$ such that $\gamma^*(S_1)=\gamma^*(S_2)$. We stress the fact that in our application of \autoref{main2} in the proof of \autoref{main-preorder}, we will use the decoration function 
from \autoref{sec:preorderNumbers}, which satisfies $\gamma^*(S_1) \neq \gamma^*(S_2)$ whenever $S_1 \neq S_2$
(since different leaves of a forest algebra expression tree $\expr$ have different preorder numbers in $\deriv{\expr}$).
In other words, we have
$\multiset{\gamma^*(S) : S \in \select(\mathcal{B}, \unfold_{\DAG}(v_0))} = \{\gamma^*(S) : S \in \select(\mathcal{B},  \unfold_{\DAG}(v_0))\}$.

Once we have proven \autoref{main2}, we obtain \autoref{main-preorder} by using for $\mathcal{B}$ the dBUTA of 
\autoref{thm-MNN}, specializing $\mathcal{C}$ to the category $\mathcal{C}_{\text{pre}}$ and taking for $\DAG$ an f-SLP $\FSLP$, whose decoration function $\gamma$ is defined in \autoref{sec:preorderNumbers}. As argued at the end of \autoref{sec:preorderNumbers}, $\mathcal{C}_{\text{pre}}$
and $\FSLP$ satisfy \autoref{assumption-category}.

Our algorithm for \autoref{main2} will be an extension of Bagan's algorithm~\cite{Bagan06}, which handles the case where the input tree has no decoration mapping $\gamma$ and is explicitly given instead of a DAG. Therefore, it will be necessary to first explain Bagan's original algorithm in some detail, which we do in \autoref{sec-bagan-tree} below. After this, we show in \autoref{sec-bagan-dag} how to extend Bagan's algorithm to the case where the input tree is given by a DAG. For this, our path enumeration algorithm from \autoref{sec-path-enumeration} 
(see \autoref{thm-enumerate-paths}) will be a crucial component.

\subsection{Bagan's algorithm for explicit binary trees} \label{sec-bagan-tree}

In this section, we discuss Bagan's algorithm~\cite{Bagan06}, which proves the following result.

\begin{theorem}[Bagan~\cite{Bagan06}] \label{BaganThm}
From a dBUTA $\mathcal{B}$ with $m$ states and a vertex-labelled binary tree $T = (V,E_\ell, E_r, \lambda)$ one can compute in preprocessing time $\bigO(|T| \cdot m^2)$ a data structure that allows to enumerate $\select(\mathcal{B}, T)$ with output-linear delay.
\end{theorem}

\begin{figure}
\begin{center}
\scalebox{1.4}{\includegraphics{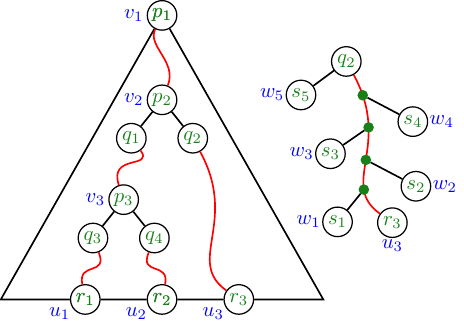}}
\end{center}
\caption{A witness tree: automaton states are in green, node names (if written) are in blue. The path from $q_2$ to $r_3$ is shown on the right with
the nodes that branch off from the path.}
	\label{witness-tree2}
\end{figure}

We will now explain this algorithm in greater detail (and in a slightly different version compared to~\cite{Bagan06}), since it is needed for our extension to the DAG-compressed case in \autoref{sec-bagan-dag}. Our treatment of Bagan's algorithm will be largely self-contained. Note that in this subsection, we are dealing with input trees without decorations.

Let $\mathcal{B}$ and $T$ be as described in \autoref{BaganThm}, i.e., $\mathcal{B} = (Q, \Sigma_0 \times \{0,1\}, \Sigma_2, \delta_0, \delta_2, Q_f)$ is a dBUTA, and $T = (V,E_\ell, E_r, \lambda)$ is a vertex-labelled binary tree as described in \autoref{sec-trees}. Recall that for a vertex $v \in V$ we denote by $T(v)$ the subtree rooted in $v$. For $S \subseteq \leaves(T)$ and $v \in V$ we define $S_v = S \cap \leaves(T(v))$. Recall also that $(T,S)$ denotes the tree obtained from marking all leaves in $S$ with a $1$. Then $\mathcal{B}(T,S)$ is the state of 
$\mathcal{B}$ to which $(T,S)$ evaluates to.
Our goal is to enumerate the set  $\select(\mathcal{B}, T) \setminus \{\emptyset\}$.
Whether $\emptyset \in \select(\mathcal{B},T)$ can be checked in the preprocessing.

\subsubsection{Witness trees} \label{sec-generalIdea}

 Bagan's algorithm enumerates all non-empty sets $S \in \select(\mathcal{B},T)$
together with a witness that $(T,S) \in L(\mathcal{B})$ holds. Consider such a non-empty set $S \in \select(\mathcal{B},T)$.
A first step towards a witness for $(T,S) \in L(\mathcal{B})$
 is to replace in the tree $T$ every node $v \in V$ 
by the pair $(v,q)$ (a so-called \emph{configuration}), where $q$ is the unique state 
$q = \mathcal{B}(T(v),S_v)$ at which $\mathcal{B}$ arrives in node $v$. Let us call this tree the \emph{configuration tree}.

The delay for producing a solution $S$ must be in $\bigO(|S|)$ (i.e., we need output-linear delay), but the configuration tree has size $|T|$ and is therefore too big. A next step towards a solution is to prune the configuration tree by keeping only those nodes that are on a path from the root to a leaf
from $S$. This yields a tree with only $|S|$ many leaves that we call the \emph{pruned configuration tree}.
 It is shown in \autoref{witness-tree2}
for an example, where $S = \{ u_1, u_2, u_3\}$. The original node names of $T$ are written in blue, automaton states
are written in green. The nodes of the pruned configuration tree
are from the following set of so-called \emph{active
configurations}:
\begin{eqnarray*}
\Conf^a(T) & = & \{ (v,q) \in V \times Q : 
\mathcal{S}^a(v,q) \neq \emptyset \} \text{ with } \\
\mathcal{S}^a(v,q) & = & \{ S \subseteq \leaves(T(v)) : S \neq \emptyset, q = \mathcal{B}(T(v),S) .
\end{eqnarray*}
The pruned configuration tree is still too big (it does not fit in space $\bigO(|S|)$)
because it may contain long paths of unary nodes (nodes with a single child except for the last node on
the path). In \autoref{witness-tree2} these are the red paths. The edges on these paths can be described as follows.
The configurations that were removed in the pruning are from the set
\[ 
\Conf^\emptyset(T)  = \{ (v,q) \in V \times Q :  q = \mathcal{B}(T(v),\emptyset) \}.
\]
The configurations $(w_1,s_1), \ldots, (w_5,s_5)$ in \autoref{witness-tree2} on the right are from this set. 
On the set of active configurations $\Conf^a(T)$ we define a new edge relation $\to$ as follows:
For active configurations $(u,p), (v,q) \in \Conf^a(T)$ with $u$ internal in $T$ and labelled with $f \in \Sigma_2$, there is an edge
$(u,p) \to (v,q)$ if there is $(v',q') \in \Conf^\emptyset(T)$ such that
one of the following two cases holds:
\begin{itemize}
\item $(u,v) \in E_\ell$, $(u,v') \in E_r$ and $\delta_2(q,q',f) = p$,
\item $(u,v) \in E_r$, $(u,v') \in E_\ell$ and $\delta_2(q',q,f) = p$.
\end{itemize}
Then all the edges of the unary paths in the pruned 
configuration tree (the red paths in \autoref{witness-tree2})
are of the above form $(u,p) \to (v,q)$. The configuration 
$(v',q') \in \Conf^\emptyset(T)$ is a configuration to which an additional edge branches off from the red unary paths
(configurations $(w_1,s_1), \ldots, (w_5,s_5)$ in \autoref{witness-tree2}).
We define the graph
$$
T \otimes \mathcal{B}  = (\Conf^a(T), \to).
$$
Since $\mathcal{B}$ is deterministic and $T$ is a tree, $T \otimes \mathcal{B}$ is an unordered forest.

 The final idea is to contract the red paths
in \autoref{witness-tree2} to single edges; this results in a tree of size $\bigO(|S|)$, which is called
a \emph{witness tree} $W$. To formally define (and construct) witness trees it is useful to define a further set of configurations, 
the so-called \emph{useful configurations}: An active configuration $(v,q) \in \Conf^a(T)$ is useful if either $v$ is a leaf in $T$
or $v$ has the children $v_1, v_2$ in $T$ and there exist states $q_1, q_2 \in Q$ such that $\delta_2(q_1, q_2, \lambda(v)) = q$
and $(v_1,q_1), (v_2,q_2) \in \Conf^a(T)$. In other words:
\begin{eqnarray*}
\Conf^u(T) & = & \{ (v,q) \in V \times Q : 
\mathcal{S}^u(v,q) \neq \emptyset \} \text{ with } \\
\mathcal{S}^u(v,q) & = & \{ S \in \mathcal{S}^a(v,q) : S_{v'} \neq \emptyset \text{ for every child $v'$ of $v$} \} .
\end{eqnarray*}
Note that $\Conf^u(T) \subseteq \Conf^a(T)$.
In a witness tree, all leaves and binary nodes are  useful configurations. In \autoref{witness-tree2} these are the 
configurations $(u_1,r_1), (u_2,r_2), (u_3,r_3)$ and $(v_2,p_2), (v_3, p_3)$.

Let us now give the formal definition of witness trees. For this and the enumeration of witness trees
it is convenient to define the following two set of successor tuples:
For $(v,q) \in \Conf^u(T)$ with $(v,v_1) \in E_\ell$, $(v,v_2) \in E_r$ and $\lambda(v) = f \in \Sigma_2$, let 
\begin{equation}
\suc_T^a(v,q) = \{ (v_1, q_1, v_2, q_2) : \;  (v_1,q_1), (v_2, q_2) \in \Conf^a(T), \delta_2(q_1,q_2,f)=q \} \label{def-suc-a}.
\end{equation}
 For $(v,q) \in \Conf^a(T)$ let 
\begin{equation} \label{suc-u}
\suc_T^u(v,q) = \{ (v',q') \in \Conf^u(T) : (v,q)  \xrightarrow{*} (v',q') \}.
\end{equation}
Note that the set $\suc_T^a(u,p)$ contains only pairs of active configurations, and the set $\suc_T^u(u,p)$ contains only useful configurations (and also contains $(u, p)$ in the case that it is useful). 

\begin{definition} \label{def-witness} \rm
A \emph{witness tree} $W$ for an active configuration $(v,q) \in \Conf^a(T)$ is a tree with root $(v,q)$. 
It is constructed recursively: 
\begin{itemize}
\item If $v$ is a leaf in $T$ then $(v,q)$ is the only vertex of $W$.
\item Assume that $v$ is not a leaf in $T$. Then in $W$ the root $(v,q)$ has a single child $(v',q') \in \suc_T^u(v,q) \subseteq \Conf^u(T)$ (we may have $(v',q') = (v,q)$ in which case we introduce a copy of the vertex $(v,q)$). If $v'$ is a leaf of $T$, then $(v',q')$ is a leaf of $W$. Otherwise, $(v',q')$ has a left child $(v_1,q_1) \in \Conf^a(T)$ and a right child $(v_2,q_2) \in \Conf^a(T)$  in $W$ such that $(v_1,q_1, v_2, q_2) \in \suc_T^a(v',q')$. The construction then continues from $(v_1,q_1)$ and $(v_2,q_2)$ in the same way as for $(v,q)$.
\end{itemize}
\end{definition}
For a witness tree $W$, let 
$$
S(W) = \{  v \in \leaves(T) :  \text{ $(v,q)$ is a leaf of $W$ for some $q \in Q$} \}.
$$
The main lemma about witness trees is:
\begin{lemma} \label{lemma-witness-tree} 
The following holds for every $(v,q) \in \Conf^a(T)$:
\begin{itemize}
\item Every witness tree $W$ for $(v,q)$ contains at most $4 |S(W)|-2$ many nodes.
\item $\mathcal{S}^a(v,q) = \{ S(W) : \text{$W$ is a witness tree for $(v,q)$} \}$ 
\item For every $S \in \mathcal{S}^a(v,q)$ there is a unique witness tree $W$ for $(v,q)$ with $S = S(W)$.
\end{itemize}
\end{lemma}
Hence, there is a one-to-one correspondence between witness trees for $(v,q) \in \Conf^a(T)$ and the leaf-sets in $\mathcal{S}^a(v,q)$. 
For the proof of this lemma we can use the following lemmas from~\cite{Bagan06}, where $\biguplus$ indicates
that the sets in the union are pairwise disjoint:

\begin{lemma} \label{lemma-disjoint1}
For $(v,q) \in \Conf^a(T)$ we have
\begin{equation} \label{eq-Sb}
\mathcal{S}^a(v,q) = \biguplus_{(v',q') \in \, \suc_T^u(v,q)}  \mathcal{S}^u(v',q').
\end{equation}
\end{lemma}

\begin{lemma} \label{lemma-disjoint2}
For $(v,q) \in \Conf^u(T)$ with
$(v, v_1) \in E_\ell$ and $(v, v_2) \in E_r$ we have
\begin{equation} \label{eq-Sa}
\mathcal{S}^u(v,q) = \biguplus_{(v_1, q_1, v_2, q_2) \in \, \suc_T^a(v,q)} 
\{ B_1 \cup B_2 : B_1 \in \mathcal{S}^a(v_1,q_1), B_2 \in \mathcal{S}^a(v_2,q_2) \}.
\end{equation}
\end{lemma}
The unions in  \eqref{eq-Sb} and \eqref{eq-Sa} go over pairwise disjoint sets since $\mathcal{B}$ is deterministic.

\begin{proof}[Proof of Lemma~\ref{lemma-witness-tree}]
The first statement holds, since $W$ has exactly $|S(W)|$ many leaves (since there do not exist different leaves of the form $(v,q)$ and $(v,q')$)
and it arises from a binary tree by inserting a unary node above every node, which doubles the number of nodes.
The other two statements follow easily from Lemmas~\ref{lemma-disjoint1} and \ref{lemma-disjoint2}. For the last point,
the disjointness of the unions in \eqref{eq-Sb} and \eqref{eq-Sa} is important.
\end{proof}

\subsubsection{Enumerating witness trees} \label{sec:enumWitnessTrees}

We can now describe Bagan's enumeration algorithm \cite{Bagan06} with the concept of witness trees. Let $r$ be the root of our binary tree $T$.
The goal is to enumerate all sets $S \subseteq \leaves(T)$ such that $\mathcal{B}(T,S) \in Q_f$.
After every solution $S$ the algorithm prints the separator symbol $\#$. For every solution $S$, time $\bigO(|S|)$
can be spend (since the delay should be output-linear). We start with checking whether $(r,q_f) \in \Conf^\emptyset(T)$ holds for some $q_f \in Q_f$ (this is part of the preprocessing phase). 
If this holds, then the algorithm starts the enumeration phase with printing a $\#$ (for the empty set).
Then all non-empty solutions $S \in \bigcup_{q_f \in Q_f} \mathcal{S}^a(r,q_f)$ have to be enumerated. For every $S \in \mathcal{S}^a(r,q_f)$, the algorithm prints a list of all elements of $S$ (viewed as numbers from $[0,\ell-1]$).

The algorithm runs over all final states $q_f \in Q_f$. Since $\mathcal{B}$ is deterministic, for two different final states $q_f, q'_f \in Q_f$ the sets $\mathcal{S}^a(r,q_f)$ and $\mathcal{S}^a(r,q'_f)$ are disjoint. 
Let us fix a final state $q_f \in Q_f$ for the further consideration. 
By Lemma~\ref{lemma-witness-tree}, it suffices to enumerate all sets $S(W)$, where $W$ is a
witness tree for $(r,q_f)$. For this, it suffices to enumerate witness trees themselves.
Thereby, every witness tree $W$ has to be produced in time $\bigO(|W|) = \bigO(|S(W)|)$ (see Lemma~\ref{lemma-witness-tree}).
To this end, we define a kind of lexicographical order on witness trees. For this, we have to fix some linear orders: 
For every configuration $(v,q) \in \Conf^a(T)$ we fix and precompute an arbitrary linear order on $\suc_T^u(v,q)$.
Moreover, for every configuration $(v,q) \in \Conf^u(T)$ we additionally fix and precompute an arbitrary linear order on $\suc_T^a(v,q)$.

Consider next a witness tree $W$ and a node $(v,q)$ of $W$. We say that $(v,q)$ is a maximal node if one of the following three cases holds:
\begin{itemize}
\item $(v,q)$ is a leaf of $W$.
\item $(v,q)$ is a unary node whose unique child is the largest $(v',q') \in \suc_T^u(v,q)$.
\item $(v,q)$ is a binary node with left (resp., right) child $(v_1,q_1)$ (resp., $(v_1,q_2)$) and $(v_1,q_1,v_2,q_2)$ is the largest 
4-tuple in $\suc_T^a(v,q)$.
\end{itemize}
If in the second (resp., third) point we take the smallest $(v',q') \in \suc_T^u(v,q)$ (resp., the smallest $(v_1,q_1,v_2,q_2) \in \suc_T^a(v,q)$), 
then we speak of a minimal node. Leaves of a witness tree are maximal as well as minimal.

We say that the witness tree $W$ is maximal (resp., minimal) if all nodes of $W$ are maximal (resp., minimal). 
Intuitively, this means that we construct the witness tree according to Definition~\ref{def-witness}, but at each
extension step (where the children of a node are defined) we take the largest (resp., smallest) available choice.
Clearly, there is a unique maximal (resp., minimal) witness tree for every $(v,q) \in \Conf^a(T)$.

The enumeration algorithm for $\mathcal{S}^a(r,q_f)$ starts with producing the unique minimal witness tree $W_0$ for $(r,q_f)$ in time $\bigO(|W_0|)$. 
For a single enumeration step, assume that $W$ is the previously produced witness tree for $(r,q_f)$. If $W$ is maximal, then the enumeration moves on to producing the witness trees for $(r, q'_f)$, where $q'_f$ is the next final state, or if $q_f$ was the last final state, the algorithm stops and prints {\sf EOE}.
Otherwise, we produce the lexicographically
next witness tree $W'$ as follows: Let $w_1, w_2, \ldots, w_n$ be the set of nodes of $W$ listed in preorder (left-to-right depth-first order).\footnote{Every order on the nodes of the witness tree would be suitable, as long as (i) one can traverse the nodes in the chosen order in constant time per node and 
(ii) the parent node of a node $w$ comes before $w$.} We can assume that this list was produced in the previous enumeration step in time 
$\bigO(|W|)$.
Let $w_i$ be the last non-maximal node in the list, i.e., $w_{i+1}, \ldots, w_n$ are maximal.
Also $w_i$ has been computed in the previous enumeration step. We then copy all nodes $w_1, \ldots, w_{i-1}$ together with their children 
to $W'$. The edges between these copied nodes are also copied from $W$ to $W'$.
In this way, we obtain a partial witness tree for $(r,q_f)$. Note that the parent node of $w_i$ belongs to 
$\{w_1, \ldots, w_{i-1}\}$. Hence, $w_i$ is also copied to $W'$. 
 Let $w_i = (v,q)$. Since leaves of $W$ are always maximal, $(v,q)$ is either a unary or a binary inner node in $W$.
 We then extend the partial witness tree at node $(v,q)$ by taking the next largest choice (compared to the choice taken in $W$). More precisely,
 if $(v,q)$ is unary and its unique child is $(v',q')$ in $W$, then we add in $W'$ an edge from $(v,q)$ to 
the configuration that comes after $(v',q')$ in our fixed linear order on $\suc_T^u(v,q)$.
Now assume that $(v,q)$ is binary and let $(v_1,q_1)$ (resp., $(v_2,q_2)$) be the left (resp., right) child of $(v,q)$ 
in $W$. Let $(v_1, q'_1, v_2, q'_2)$ be the 4-tuple that comes after $(v_1,q_1,v_2,q_2)$ in our fixed order on $\suc_T^a(v,q)$. Then
$(v_1,q'_1)$ becomes the left child and $(v_2,q'_2)$ becomes the right child of $(v,q)$ in $W'$.

In the last step, it remains to complete the partial witness tree $W'$ constructed so far to a (complete) witness tree by extending $W'$
below leaves $(v,q)$ of $W'$ such that $v$ is \emph{not} a leaf in $T$.
During this extension we follow the recursive definition of witness trees (Definition~\ref{def-witness})
but always choose the smallest element from $\suc_T^u(v,q)$ (resp., $\suc_T^a(v,q)$) if $(v,q)$ is the leaf where we currently extend the partial witness tree. 
Since we spend constant time for each node of the final witness tree $W'$,
it follows that $W'$ can be constructed in time $\bigO(|W'|)$.

\subsubsection{Preprocessing} \label{sec bagan preproc}

It remains to argue that all preprocessing can be done in time linear in $|T|$. For this, we need the following lemma:

\begin{lemma} \label{remark-precompute2}  
The sets $\Conf^a(T)$, $\Conf^u(T)$, $\Conf^\emptyset(T)$, and the forest $T \otimes \mathcal{B}$ can be computed bottom-up on the tree $T$
in time $\bigO(|T| \cdot |Q|^2)$.
\end{lemma}

\begin{proof}
We first show how to compute $\Conf^a(T)$, $\Conf^u(T)$, and $\Conf^\emptyset(T)$.
Recall that $\delta_0$ maps from $\Sigma_0 \times \{0,1\}$ to $Q$.
For an $a$-labelled leaf $v$ of $T$ and $q \in Q$ we have:
\begin{itemize}
\item $(v,q) \in \Conf^a(T)$ iff  $(v,q) \in \Conf^u(T)$ iff $q = \delta_0(a,1)$,
\item $(v,q) \in \Conf^\emptyset(T)$ iff $q = \delta_0(a,0)$.
\end{itemize}
Assume now that $(v,v_1) \in E_\ell$, $(v, v_2) \in E_r$ and $\lambda(v) = f \in \Sigma_2$. We have:
\begin{itemize}
\item $(v,q) \in \Conf^\emptyset(T)$ if and only if  there are $(v_1,q_1), (v_2,q_2) \in \Conf^\emptyset(T)$ with
$\delta_2(q_1,q_2,f)=q$.
\item $(v,q) \in \Conf^a(T)$ if and only if  there are $(v_1,q_1), (v_2,q_2) \in \Conf^a(T) \cup \Conf^\emptyset(T)$ with
$\{ (v_1,q_1), (v_2,q_2)\}  \cap \Conf^a(T) \neq \emptyset$ and $\delta_2(q_1,q_2,f)=q$.
\item $(v,q) \in \Conf^u(T)$ if and only if  there are
$(v_1,q_1), (v_2,q_2) \in \Conf^a(T)$ with $\delta_2(q_1,q_2,f)=q$.
\end{itemize}
This allows to compute for each node $v \in V$ the set of states $q \in Q$ such that $(v,q) \in \Conf^x(T)$ for $x \in \{a,u,\emptyset\}$.
To obtain the time bound $\bigO(|T| \cdot |Q|^2)$, we iterate for each $f$-labelled node $v$ over all state pairs $(q_1, q_2) \in Q \times Q$, compute
$q := \delta_2(q_1, q_2, f)$ and then add the configuration $(v,q)$ to the set $\Conf^a(T)$, $\Conf^u(T)$, or $\Conf^\emptyset(T)$,
depending on the membership of $(v_1, q_1)$ and $(v_2, q_2)$ in these sets.

In order to compute $T \otimes \mathcal{B}$, we proceed again bottom-up on $T$. Assume that $(v,v_1) \in E_\ell$, $(v, v_2) \in E_r$ and $\lambda(v) = f$. We iterate over all state pairs $(q_1, q_2) \in Q \times Q$, and compute $q := \delta_2(q_1,q_2,f)$. We then add the edge $(v, q) \to (v_1, q_1)$ if  $(v,q), (v_1, q_1) \in \Conf^a(T)$ and  $(v_2, q_2) \in \Conf^\emptyset(T)$. Similarly, we add the edge $(v, q) \to (v_2, q_2)$ if 
$(v,q), (v_2, q_2) \in \Conf^a(T)$ and  $(v_1, q_1) \in \Conf^\emptyset(T)$.
This procedure needs time $\bigO(|T| \cdot |Q|^2)$. 
\end{proof}
Let us assume now that the sets  $\Conf^a(T)$, $\Conf^u(T)$, $\Conf^\emptyset(T)$ and the forest $T \otimes \mathcal{B}$ 
have been precomputed (in time  $\bigO(|T| \cdot |Q|^2)$). \label{pageRefComputeLinOrder}
The sets $\suc_T^a(v,q)$ (for $(v,q) \in \Conf^u(T)$) are pairwise disjoint and their  union has size $\bigO(|T| \cdot |Q|^2)$.
We can compute  in time $\bigO(|T| \cdot |Q|^2)$ this union together with a linear order where every set $\suc_T^a(v,q)$ forms an interval.
 For the sets $\suc_T^u(v,q)$ for  $(v,q) \in \Conf^a(T)$ the situation is not so clear.
These sets have size $\bigO(|T|\cdot |Q|)$ but they are not disjoint. Nevertheless, one can compute a global linear order on
the set $\Conf^u(T)$ such that every set $\suc_T^u(v,q)$ is an interval of this global linear order; see \cite{Bagan06}.\footnote{We do not need the argument from
\cite{Bagan06}, since at this step, we anyway need another solution 
 for our extension of Bagan's algorithm to DAGs; see \autoref{sec-bagan-dag}.}
Then, every $\suc_T^u(v,q)$ can be represented by the smallest and largest configuration of the corresponding
interval. This is good enough for the above enumeration algorithm.

\subsection{Extending Bagan's algorithm to DAG-foldings of binary trees} \label{sec-bagan-dag}

In this section we prove \autoref{main2}.
For this, we have to take care of the setting where the input binary tree $T$ is not given explicitly (as in the previous section), but by a vertex-labelled $\mathcal{C}$-decorated binary DAG $\DAG = (V, E, \lambda, \gamma)$ and a distinguished vertex $v_0$ (where the latter is not known during the preprocessing phase) with
$\mathcal{C}$ being a category that satisfies the computational assumptions from \autoref{sec-enumeration-general}.
More precisely, for $T = \unfold_{\DAG}(v_0)$ we want to enumerate all sets $\gamma^*(S) = \multiset{ \gamma^*(v) : v \in S}$ with $S \in \select(\mathcal{B}, T)$. Since we are interested in enumerating $\gamma^*(S)$ for every $S \in \select(\mathcal{B}, T)$ and $\gamma^*(S) = \gamma^*(S')$ is possible for distinct sets $S, S' \in \select(\mathcal{B}, T)$, our task is to enumerate the \emph{multiset} $\multiset{\gamma^*(S) : S \in \select(\mathcal{B}, T)}$.\footnote{Recall that in our application of this algorithm with the decoration mapping $\gamma$ from \autoref{sec:preorderNumbers}, we have
$\gamma^*(S) \neq \gamma^*(S')$ whenever $S \neq S'$.}
This generalises Bagan's original algorithm in two regards: Firstly, we have to deal with the $\mathcal{C}$-morphisms, i.e., instead of enumerating sets of leaves, we have to enumerate the sets of the corresponding $\mathcal{C}$-morphisms, and, secondly, we have to deal with the situation that the input tree is given by a DAG.

Handling the $\mathcal{C}$-morphisms is more or less straightforward. Since $\DAG$ is $\mathcal{C}$-decorated, every tree $T = \unfold_{\DAG}(v_0)$ is a $\mathcal{C}$-decorated tree. 
The $\mathcal{C}$-decoration of $T$ yields a $\mathcal{C}$-decoration of $T \otimes \mathcal{B}$ in the natural way. More precisely,
$T \otimes \mathcal{B}$ becomes the unordered and $\mathcal{C}$-decorated forest $(\Conf^a(T), \to, \gamma)$, where we set $\gamma(v,q) = \gamma(v)$ for every $(v,q) \in \Conf^a(T)$ and $\gamma((u,p),(v,q)) = \gamma(u,v)$ for every
edge $((u,p), (v,q))$ of $T$. Now, the $\mathcal{C}$-morphism of a path from some vertex $(u,p)$ to a vertex $(v,q)$ in $T \otimes \mathcal{B}$ is exactly the 
$\mathcal{C}$-morphism of the unique path from $u$ to $v$ in $T$. 

Similarly, we can decorate witness trees. By definition, for every edge $((u,p),(v,q))$ in a witness tree $W$, there is a unique path $\pi$ from $u$ to $v$ in the tree $T$ and we define
$\gamma((u,p), (v,q)) = \gamma(\pi)$. These decorated witness trees share the crucial properties of their undecorated counterparts; in particular, Lemma~\ref{lemma-witness-tree} still holds. Moreover, if the configuration $(v,q)$ is a leaf of a witness tree for the configuration $(u,p)$, then the $\mathcal{C}$-morphism of $(v,q)$ (which by definition
is the $\mathcal{C}$-morphism of the path from the root $(u,p)$ to the leaf $(v,q)$) is equal to the $\mathcal{C}$-morphism of the leaf $v$ in the tree $T(u)$. Consequently, from a $\mathcal{C}$-decorated witness tree $W$ that represents a leaf-set $S$, we can easily obtain $\gamma^*(S)$ from $W$'s $\mathcal{C}$-decoration (i.e., we can compute $\gamma^*(S)$ in time $\bigO(|W|)$ by one top-down traversal of $W$, or while constructing the decorated witness tree).

Consequently, in order to enumerate all sets $\gamma^*(S)$ with $S \in \select(\mathcal{B}, \unfold_{\DAG}(v_0))$, it is sufficient to enumerate all decorated witness trees, but now we have to do this in the setting where the input tree is compressed by $\DAG$. We will now explain how this is possible.

Recall that the vertices of $T = \unfold_{\DAG}(v_0)$ are paths $\pi \in \path_{\DAG}(v_0)$ that start in vertex $v_0$ of $\DAG$ and end in an arbitrary vertex of $\DAG$. 
Consequently, we will denote vertices of a tree $T= \unfold_{\DAG}(v_0)$ in the following with 
$\pi, \pi'$, etc., whereas vertices of $\DAG$ will be denoted with $u$, $v$, etc.
Consider now two paths $\pi \in \path_{\DAG}(v_0)$ and $\pi' \in \path_{\DAG}(v'_0)$  
with the same terminal vertex $\omega(\pi) = \omega(\pi')$ (the start vertices may differ) and let 
$T= \unfold_{\DAG}(v_0)$ and $T'= \unfold_{\DAG}(v'_0)$. Since 
$\omega(\pi) = \omega(\pi')$, the subtrees $T(\pi)$ and $T'(\pi')$ rooted in $\pi$ and $\pi'$, respectively, are isomorphic. 
The following lemma is a direct consequence of this fact:

\begin{lemma}
Let $\pi$, $\pi'$, $T$, and $T'$ be as above. Then for all states $q$ of $\mathcal{B}$ and
all $x \in \{a,u,\emptyset\}$ we have: $(\pi,q) \in \Conf^x(T)$ if and only if $(\pi',q) \in \Conf^x(T')$.
\end{lemma}

This allows to define for the DAG $\DAG$ configuration sets $\Conf^x(\DAG) \subseteq V \times Q$: $(v,q) \in \Conf^x(\DAG)$ if and only if $(\pi,q) \in \Conf^x(\unfold_{\DAG}(v_0))$, where  $v_0 \in V$ and the path $\pi$ are such that $\pi$ is a path from $v_0$ to $v$.
This definition does not depend on the choice of $v_0$ and $\pi$.
The sets $\Conf^x(\DAG)$ can be computed by a bottom-up parse of $\DAG$ in exactly the same way as the sets $\Conf^x(T)$ in the proof
of Lemma~\ref{remark-precompute2}. We therefore obtain:

\begin{lemma} \label{remark-precomputeCompressSets} 
The sets  $\Conf^a(\DAG)$, $\Conf^u(\DAG)$ and $\Conf^\emptyset(\DAG)$ can be computed  in time $\bigO(|\DAG| \cdot |Q|^2)$.
\end{lemma}
Based on the sets $\Conf^a(\DAG)$, $\Conf^u(\DAG)$ and $\Conf^\emptyset(\DAG)$ we also define a $\DAG$-version 
$\DAG \otimes \mathcal{B} = (\Conf^a(\DAG), E', \gamma')$
of the forest $T \otimes \mathcal{B}$ from
\autoref{sec-generalIdea}, which is a $\mathcal{C}$-decorated DAG:  
The edges in $E'$ use the index set $I = \{\ell, r\}$ in order to distinguish multiple edges between two vertices. To define $E' \subseteq \Conf^a(\DAG) \times \{\ell, r\} \times \Conf^a(\DAG)$, let $d \in \{\ell,r\}$ and $(u,p), (v,q) \in \Conf^a(\DAG)$ such that $\lambda(u) = f \in \Sigma_2$. Then, there is an edge $((u,p), d, (v,q))  \in E'$  iff there is $(v',q') \in \Conf^\emptyset(\DAG)$ such that one of the following two cases holds
(recall that $E \subseteq V \times \{\ell,r\} \times V$ since $\DAG$ is a binary DAG):
\begin{itemize}
\item $d = \ell$, $(u,\ell,v), (u, r, v') \in E$ and $\delta_2(q,q',f) = p$,
\item $d = r$, $(u,r,v), (u, \ell, v') \in E$ and $\delta_2(q',q,f) = p$.
\end{itemize}
The decoration mapping $\gamma'$ is inherited from $\DAG$: we set $\gamma'(u,p) = \gamma(u)$ for $(u,p) \in \Conf^a(\DAG)$ and
$\gamma'((u,p),d,(v,q)) = \gamma(u,d,v)$ for an edge $((u,p),d,(v,q)) \in E'$.

The following lemma is shown in the same way as the corresponding statement for 
the forest $T \otimes \mathcal{B}$ in Lemma~\ref{remark-precompute2}.

\begin{lemma} \label{remark-precomputeCompressDAGStructure} 
The DAG $\DAG \otimes \mathcal{B}$ can be computed bottom-up on the DAG $\DAG$ in time $\bigO(|\DAG| \cdot |Q|^2)$.
\end{lemma}

We need the following fact, which follows directly from the definition of the forest $T \otimes \mathcal{B}$ and the DAG $\DAG \otimes \mathcal{B}$.

\begin{lemma} \label{lemma-TB-DB}
Let $\pi \in \path_{\DAG}(v_0)$ be a vertex of $T = \unfold_{\DAG}(v_0)$ with $v = \omega(\pi)$. Moreover, let $(v,d,v') \in E$ be an edge of the DAG $\DAG$ and let $q, q' \in Q$.
Then $(\pi, q) \to (\pi d v',q')$ in the forest $T \otimes \mathcal{B}$ if and only if $((v, q), (v',q'))$ is an edge in $\DAG \otimes \mathcal{B}$.
\end{lemma}

Consider now a witness tree $W$ for $T= \unfold_{\DAG}(v_0)$ with root $(v_0, q) \in \Conf^a(T)$ (recall that the root of $T$ is the empty path $v_0$ in $\DAG$). The vertices of $W$ 
are pairs $(\pi, q) \in \Conf^a(T)$. Of course, the names of the vertices of $W$ are not important; it is perfectly fine to enumerate for every witness tree $W$ an isomorphic copy.
But for doing this, we do not have to know the name $(\pi,q)$ of a vertex, when we construct the children of $(\pi,q)$ according to Definition~\ref{def-witness}. Only the pair $(\omega(\pi), q) \in \Conf^a(\DAG)$ is
important for this. This means that instead of storing the pair $(\pi,q)$ we can store an abstract vertex that is labelled with the pair $(\omega(\pi),q)$.
Let us explain this in more detail. 

In the following, let $\pi \in \path_{\DAG}(v_0)$ be a vertex in $T$ and let $v = \omega(\pi)$.
There are two extension steps in the definition of witness trees from Definition~\ref{def-witness}:
\begin{itemize}
\item For a configuration $(\pi,q) \in \Conf^a(T)$ choose a vertex $(\pi',q') \in \suc^u_T(\pi,q)$ as the unique child of $(\pi,q)$ in a witness tree.
\item For a configuration $(\pi,q) \in \Conf^u(T)$ choose a 4-tuple $(\pi \ell v_1, q_1, \pi r v_2, q_2) \in \suc^a_T(\pi,q)$.
Here, $v_1$ is the left child of $v$ in $\DAG$ and $v_2$ is the right child of $v$ in $\DAG$. By definition of $\suc^a_T(\pi,q)$,
we have $(\pi \ell v_1, q_1), (\pi r v_2, q_2) \in \Conf^a(T)$
and these two configurations become the two children of $(\pi,q)$ in the witness tree.
\end{itemize}
These two steps can be done without actually unfolding $\DAG$ into $T$.
For the second step this is easy to see: For a configuration $(v,q) \in \Conf^u(\DAG)$ with $(v, \ell, v_1), (v, r, v_2) \in E$ and $\lambda(v) = f$ we define
\begin{equation} \label{def sucDAGa}
\sucDAG^a(v,q) = \{ (v_1, q_1, v_2, q_2) : \;  (v_1,q_1), (v_2, q_2) \in \Conf^a(\DAG), \delta_2(q_1,q_2,f)=q \}\,.
\end{equation}
We can compute linear orders for the sets $\sucDAG^a(v,q)$ in time $\bigO(|\DAG| \cdot|Q|^2)$ analogously to the uncompressed setting; see  \autoref{sec bagan preproc}. More precisely, for every vertex $v \in V$ of $\DAG$ labelled with $f$ that has the left child $v_1$ and the right child $v_2$, we proceed as follows. For all $q_1, q_2 \in Q$ with $(v_1,q_1), (v_2, q_2) \in \Conf^a(\DAG)$, we compute $q = \delta_2(q_1, q_2, f)$ (we then have $(v,q) \in \Conf^u(\DAG)$) and add $(v_1, q_1, v_2, q_2)$ to $\sucDAG^a(v,q)$. Clearly, we can do this in such a way that we store each set $\sucDAG^a(v,q)$ as a list of its elements in some order.

Clearly, there is a natural bijection between the sets $\suc^a_T(\pi,q)$ and $\sucDAG^a(\omega(\pi,q))$.
Hence, if we want to extend a partially constructed witness tree in an abstract vertex (let us call it $\hat{v}$) that is labelled with $(v,q) \in \Conf^u(T)$ (which means that $\hat{v}$ represents 
a vertex $(\pi,q)$ with $\omega(\pi)=v$) then we have to choose a 4-tuple $(v_1, q_1, v_2, q_2)  \in \sucDAG^a(v,q)$ and add to the witness tree two abstract vertices, namely
a left child of $\hat{v}$ that is labelled with $(v_1, q_1)$ (it represents $(\pi \ell v_1,q_1)$) and a right child of 
$\hat{v}$ that is labelled with $(v_2, q_2)$ (it represents $(\pi r v_2,q_2)$). The corresponding edges are decorated with $\gamma(v,v_1)$ and $\gamma(v,v_2)$.

The first extension step in the witness tree construction, where one has to choose a vertex $(\pi',q') \in \suc^u_T(\pi,q)$ for  $(\pi,q) \in \Conf^a(T)$, is a bit more subtle.
Recall that $\suc^u_T(\pi,q)$ is the set of all pairs $(\pi',q') \in \Conf^u(T)$ that can be reached from $(\pi,q)$ in the forest $T \otimes \mathcal{B}$.
Hence, one has to choose an arbitrary path $\xi = v_1 d_1 v_2 d_2 \cdots v_{k-1} d_{k-1} v_k$ in the DAG $\DAG$ and states $q_1, \ldots, q_k$
such that the following holds:
\begin{itemize}
\item $v_1 = \omega(\pi)$ and $q_1 = q$,
\item there is an edge from $(\pi d_1 v_2 \cdots d_{i-2} v_{i-1} d_{i-1} v_i, q_i)$ to $(\pi d_1 v_2 \cdots d_{i-1} v_{i} d_{i} v_{i+1}, q_{i+1})$ in the
forest  $T \otimes \mathcal{B}$ for all $1 \leq i \leq k-1$, and
\item $(\pi d_1 v_2 \cdots d_{k-2} v_{k-1} d_{k-1} v_k, q_k) \in \Conf^u(T)$.
\end{itemize}
By Lemma~\ref{lemma-TB-DB} we can equivalently choose a path from $(\omega(\pi), q)$ to a vertex $(v',q') \in \Conf^u(\DAG)$ in the DAG 
$\DAG \otimes \mathcal{B}$.
In other words, there is a canonical bijection between the set  $\suc^u_T(\pi,q)$ for $(\pi,q) \in \Conf^a(T)$ with $v = \omega(\pi)$
 and the set
$$
\sucDAG^u(v,q) := \{ \xi : \xi \text{ is a path from $(v,q)$ to a configuration $(v',q') \in  \Conf^u(\DAG)$ in } \DAG \otimes \mathcal{B} \}.
$$
Therefore, if we want to extend a partially constructed witness tree in an abstract vertex $\hat{v}$ that is labelled with $(v,q) \in \Conf^a(\DAG)$ 
 then we have to choose a path $\xi \in \sucDAG^u(v,q)$ ending in $(v',q') \in  \Conf^u(\DAG)$ 
 and add to the witness tree a new abstract vertex $\hat{v}'$ labelled with $(v',q')$ as the unique child of $\hat{v}$. The edge from $\hat{v}$ to 
 $\hat{v}'$ is decorated with the morphism $\gamma(\xi)$. Notice that the latter depends on the path $\xi$ and not just its terminal vertex.

By the previous discussion, the enumeration of witness trees for $T = \unfold_{\DAG}(v_0)$ works in the same way as the enumeration of witness trees for an explicitly
given tree $T$ in \autoref{sec:enumWitnessTrees}, with the only difference that we use the above sets $\sucDAG^a(v,q)$ (for $(v,q) \in \Conf^u(\DAG)$)
and $\sucDAG^u(v,q)$ (for $(v,q) \in \Conf^a(\DAG)$). We have already remarked that for the sets $\sucDAG^a(v,q)$ one can compute some
fixed linear orders and use these orders for the enumeration phase (in the same way as we did in \autoref{sec:enumWitnessTrees} for an explicitly given tree $T$).

However, the sets $\sucDAG^u(v,q)$ cannot be constructed explicitly (they may have size exponential in the DAG $\DAG$). Instead of explicitly computing them in the preprocessing, we will invoke the path enumeration algorithm of \autoref{thm-enumerate-paths} for enumerating them on demand as follows. 

In the enumeration of the $\mathcal{C}$-decorated witness trees for $T = \unfold_{\DAG}(v_0)$, we use \autoref{thm-enumerate-paths} for the DAG $\DAG \otimes \mathcal{B}$ and the set $V_0 = \Conf^u(\DAG)$. The preprocessing from \autoref{thm-enumerate-paths} is obviously carried out in the overall preprocessing of the whole enumeration algorithm. Moreover, this preprocessing is independent of the vertex $v_0$ in $T = \unfold_{\DAG}(v_0)$.
If during the enumeration of witness trees we created a new abstract vertex $\hat{v}$ that is labelled with $(v, q) \in \Conf^a(\DAG)$ and that should have a single child in the witness tree,
then we start the enumeration phase from \autoref{thm-enumerate-paths} for the DAG $\DAG \otimes \mathcal{B}$, the target set $V_0 = \Conf^u(\DAG)$, and $s = (v, q)$. As soon as we receive
an element $\langle (v',q'), \alpha\rangle$ (here $(v',q') \in \Conf^u(\DAG)$ and $\alpha$ is a morphism from our category $\mathcal{C}$) we create the child  $\hat{v}'$ of $\hat{v}$, label
$\hat{v}'$ with $(v',q')$ and label the edge from  $\hat{v}$ to  $\hat{v}'$ with $\alpha$.  We then freeze the enumeration and resume it later when we 
have to modify the outgoing edge for $\hat{v}$. 

We should emphasise that the preprocessed data structure in \autoref{thm-enumerate-paths} is persistent in the sense that it is not destroyed during an enumeration phase for a certain source vertex $s$. In fact, at each time instant during the witness tree enumeration, there are in general several active (but frozen) enumeration phases for different start vertices
$s$. Therefore, persistence is important. 

Notice that if the DAG $\DAG$ fulfills \autoref{assumption-category} then the same holds for the DAG 
$\DAG \otimes \mathcal{B}$ under the additional assumption that states of $\mathcal{B}$ fit into a constant number of RAM registers. 
This ensures that vertices of $\DAG \otimes \mathcal{B}$ fit into a constant number of RAM registers.
Moreover, the path decorations
$\gamma(\pi)$ that appear in the DAG $\DAG \otimes \mathcal{B}$ are the same as the path decorations in the DAG $\DAG$.
Recall that \autoref{assumption-category} is needed in order to apply \autoref{thm-enumerate-paths}.

Note further that for the proof of \autoref{mainResultForests}, the DAG $\DAG$ is the input f-SLP $\FSLP$ 
(for which \autoref{assumption-category} is justified  at the end
of \autoref{sec:preorderNumbers}),
whereas the dBUTA $\mathcal{B}$ only depends on the MSO-query $\Psi$. Since we assume the data complexity setting, the 
size of $\Psi$ and hence the number of states of $\mathcal{B}$  is considered to be a constant. 

This concludes the proof of \autoref{main2} and hence the proof of our main result, \autoref{main-preorder}.

\section{Dealing with Updates}\label{sec:updates}

We have seen that even if unranked forests are highly compressed by f-SLPs, we can still enumerate the result set of an MSO-query with output-linear delay and after linear preprocessing. A natural question is whether our approach can be extended to the dynamic setting, i.e., after updating our data, we want to enumerate the query result again but with respect to the updated data and without having to repeat the whole preprocessing from scratch. Solving this task aims at two objectives: On the one hand, we have to be able to perform the update directly on the compressed data (clearly, we do not want to decompress our data, update it and then compress it again) and, on the other hand, our updates have to maintain the data structures that are necessary for running the enumeration algorithm. 

\subsection{The enumeration data structure}

As mentioned above, in order to perform an update for an f-SLP-compressed forest, we do not only have to update the f-SLP accordingly, but we also have to update the data structures that are needed for our linear preprocessing and output-linear delay enumeration algorithm. Let us formally define these data structures. 

First, let us recall that as a component of the enumeration algorithm of \autoref{main-preorder}, we had to devise an algorithm for a certain enumeration problem on DAGs in \autoref{sec-path-enumeration}. As a preprocessing of this algorithm, we first applied several simplification steps that eventually produced a binary version of the input DAG. In the following, we use the notation $\DAG_b$ for the binary 
decorated DAG obtained from the decorated DAG $\DAG$ by applying the construction from \autoref{sec-path-enumeration} (recall that in the course of the overall algorithm from \autoref{main-preorder}, the decorated DAG $\DAG$ on which we applied this construction is the DAG $\FSLP \otimes\,\mathcal{B}$ for an f-SLP $\FSLP$ and a dBUTA $\mathcal{B}$). 

Let $\FSLP$ be an f-SLP (viewed as a $\mathcal{C}_{\text{pre}}$-decorated DAG), $\mathcal{A}$ be an nSTA and $\mathcal{B}$ be the corresponding dBUTA from \autoref{thm-MNN}.
The \emph{enumeration data structure} for $\FSLP$ and $\mathcal{B}$ consists of the following data:
\begin{itemize}
\item the sets $\Conf^x(\FSLP)$ for all $x \in \{a,u,\emptyset\}$ (as defined in \autoref{sec-bagan-dag} for $\DAG = \FSLP$),
\item a linear order on the set $\sucFSLP^a(v,q)$ for every configuration $(v,q) \in \Conf^u(\FSLP)$; see \eqref{def sucDAGa},
\item the preorder effects $f_e$ for every edge $e$ of the DAG $\FSLP$ (see \autoref{sec:preorderNumbers}),
\item the binary DAG $(\FSLP \otimes\,\mathcal{B})_b$ (see \autoref{sec-path-enumeration}) together with the vertex $\omega_r[v]$
and the $\mathcal{C}_{\text{pre}}$-morphism $\gamma_r[v]$ (it is also a preorder effect) for every vertex $v$ of $(\FSLP \otimes\,\mathcal{B})_b$.
\end{itemize}
In \autoref{sec-main-result}, we have seen how the enumeration data structure can be computed in time
$\bigO(\card{\FSLP})$ (in data complexity) and, provided that we have the enumeration data structure 
at our disposal, how it can be used in order to enumerate $\select(\mathcal{A}, \derivsub{\FSLP}{A})$ 
for any chosen forest vertex $A$ of $\FSLP$
with output-linear delay in data complexity.

\subsection{Extensions of f-SLPs}

The relabelling updates considered in the next \autoref{sec:relabellingUpdates} are achieved by manipulating the initial f-SLP by only adding new vertices to it (together with their vertex labels and outgoing edges). 

Let $\FSLP = (V, E, \lambda)$ be an f-SLP and let $\newNodeMarker{V}$ be a finite set of vertices 
with $V \cap \newNodeMarker{V} = \emptyset$.
 An f-SLP $\FSLP' = (V \cup \newNodeMarker{V}, E', \lambda')$ is called a $\newNodeMarker{V}$-\emph{extension} of $\FSLP$ if 
 $E' \cap (V \times \{\ell,r\} \times V) = E$ and $\lambda(v) = \lambda'(v)$ for all $v \in V$.
 We call $V$ the \emph{old} vertices and $\newNodeMarker{V}$ the \emph{new} vertices of $\FSLP'$. 
 If $\newNodeMarker{V}$ is not important, we speak of an extension of $\FSLP$.
 Note that the extension $\FSLP'$ still contains all old vertices with exactly the same outgoing edges and labels as in $\FSLP$.
Hence, for every $v \in V$ we have $\derivsub{\FSLP'}{v} = \derivsub{\FSLP}{v}$.
Therefore we can omit the indices $\FSLP$ and $\FSLP'$  (for a vertex from $\newNodeMarker{V}$ only the index $\FSLP'$ makes sense).

\begin{lemma} \label{lemma-1-ext}
Let $\mathcal{B}$ be a dBUTA with state set $Q$ and let
$\FSLP, \FSLP'$ be f-SLPs, where 
$\FSLP'$ is a $\newNodeMarker{V}$-extension of $\FSLP$.
From $\FSLP'$ and the enumeration data structure for $\FSLP$ and $\mathcal{B}$ one can compute in time 
$\bigO(\card{Q}^2 \cdot \card{\newNodeMarker{V}})$ (i.e., time $\bigO(\card{\newNodeMarker{V}})$  in data complexity)
 the enumeration data structure for $\FSLP'$ and $\mathcal{B}$.
\end{lemma}

\begin{proof} 
It suffices to prove the lemma for the case $\card{\newNodeMarker{V}}=1$ since a $\newNodeMarker{V}$-extension 
can be obtained by doing $\card{\newNodeMarker{V}}$ many one-node extensions. Hence, assume that 
$\newNodeMarker{V} = \{\newNodeMarker{v}\}$ and let 
$v_1$ and $v_2$ be the children of $\newNodeMarker{v}$ (thus, $v_1, v_2$ belong to $\FSLP$).
We proceed as follows:
\begin{itemize}
\item For every $q \in Q$, we determine whether $(\newNodeMarker{v},q)$ belongs to the set
$\Conf^x(\FSLP')$ for all $x \in \{a,u,\emptyset\}$. For this, we only need to access the sets
$\{ v_1, v_2 \} \times Q \cap \Conf^x(\FSLP)$, which are available in the enumeration data structure for $\FSLP$
(see also the proof of Lemma~\ref{remark-precompute2}).
\item For every new configuration $( \newNodeMarker{v},q) \in \Conf^u(\FSLP)$, we compute a linear order of the 
set $\suc^a( \newNodeMarker{v},q)$ as described in \autoref{sec bagan preproc}. 
\item We compute the leaf size and left size (see \autoref{sec:preorderNumbers}) for the new vertex $\newNodeMarker{v}$
from the corresponding values for $v_1$ and $v_2$. Then we compute the preorder effects for the two edges
from $\newNodeMarker{v}$ to $v_1$ and $v_2$, respectively.
\item For every new configuration from $(\newNodeMarker{v}, q) \in \Conf^a(\FSLP')$, we compute the children of this configuration in the DAG $\FSLP' \otimes\,\mathcal{B}$, but we have to do this in such a way that we actually compute the binary DAG $(\FSLP' \otimes\,\mathcal{B})_b$. It is important here that the number of children of a new vertex $(\newNodeMarker{v}, q)$ in $\FSLP' \otimes\,\mathcal{B}$ is bounded by $2 |Q|$. This ensures that for every new vertex $(\newNodeMarker{v}, q)$ only a constant number of new vertices and edges have to be added to $(\FSLP \otimes\,\mathcal{B})_b$. Thereby we can also compute the vertex $\omega_r[y]$ and the weight $\gamma_r[y]$
(see \autoref{subsec path enumeration}, page~\pageref{page-omega_r})
 for every new vertex $y$ of $(\FSLP \otimes\,\mathcal{B})_b$. 
\end{itemize}
Each of the above steps needs time  $\bigO(\card{Q}^2)$.
\end{proof}

\subsection{Relabelling updates for f-SLP-compressed unranked forests}\label{sec:relabellingUpdates}

We now argue that our enumeration algorithm can be easily extended with relabelling updates, i.e., updates
that change the symbol of a specified vertex of the queried forest.
Formally, we define for a forest $F = (V,E,R,\lambda)$, a vertex $v \in V$ and a symbol $a \in \Sigma$ the new forest
$\mathsf{relabel}(F,v,a) = (V,E,R,\lambda')$, where $\lambda'(v) = a$ and $\lambda'(v') = \lambda(v')$ for all
$v' \in V \setminus \{v\}$.

Let us now consider an f-SLP $\FSLP$. It should be seen as a specification of a set of forests, one for each forest vertex
$A$. Given a forest vertex
$A$, a vertex $v$ in $\derivsub{\FSLP}{A}$ (represented by its preorder number in $\derivsub{\FSLP}{A}$) and
a symbol $a \in \Sigma$, it is our goal to compute a $\newNodeMarker{V}$-extension $\FSLP'$ (for some set $\newNodeMarker{V}$ of new
vertices) that contains a new vertex $\newNodeMarker{A} \in \newNodeMarker{V}$
with $\derivsub{\FSLP'}{\newNodeMarker{A}}=\mathsf{relabel}(\derivsub{\FSLP}{A},v,a)$.
The time needed for this depends on the height $\ordersub{\FSLP}{A}$ of $A$, which is
the maximal length of an $A$-to-leaf path in the DAG $\FSLP$. We write $\order{A}$ if the f-SLP $\FSLP$ is clear from the context. For instance, we have $\order{A} = 5$ in the f-SLP from \autoref{fig-fslp}.

For a vertex $B$ of an f-SLP recall its type $\tau(B)$ (see \autoref{subsubsection:FSLPs}), its
leaf size $s(B)$ and its left size $\ell(B)$ (see \autoref{sec:preorderNumbers}).

\begin{theorem}\label{relabellingTheorem}
Assume that the following is given:
\begin{itemize}
\item an f-SLP $\FSLP = (V,E,\lambda)$ together with the values $\tau(B)$, $s(B)$ and $\ell(B)$ (the latter only 
in case $\tau(B)=1$) for all $B \in V$,
\item a forest vertex $A$ of $\FSLP$,
\item the preorder number $k$ of a vertex $v$ from $\derivsub{\FSLP}{A}$ and 
\item a symbol $a \in \Sigma$.
\end{itemize}
One can then compute in time $\bigO(\order{A})$ a $\newNodeMarker{V}$-extension $\FSLP'$ of $\FSLP$ and a vertex $\newNodeMarker{A} \in \newNodeMarker{V}$ such that
\begin{itemize}
\item $\derivsub{\FSLP'}{\newNodeMarker{A}}=\mathsf{relabel}(\derivsub{\FSLP}{A},v,a)$,
\item $\order{\newNodeMarker{B}} \leq \order{A}$ for all  $\newNodeMarker{B} \in \newNodeMarker{V}$ and 
\item $|\newNodeMarker{V}|\leq \order{A}+1$.
\end{itemize}
If in addition the enumeration data structure for $\FSLP$ and some dBUTA $\mathcal{B}$ with state set $Q$ is given then one can also compute the enumeration data structure for $\FSLP'$ and $\mathcal{B}$
in time $\bigO(\card{Q}^2 \cdot  \order{A})$.
\end{theorem}

\begin{proof}
Let us first explain how to compute the unique path $\pi$ from $A$ to a leaf of $\FSLP$ that
corresponds to the vertex $v$ (with preorder number $k$) of $\derivsub{\FSLP}{A}$.
The algorithm walks from $A$ down in $\FSLP$ and thereby stores in each step
a pair $(B,m) \in V \times \mathbb{N}$ if $\tau(B)=0$ and a triple $(B,m,p) \in V \times \mathbb{N} \times \mathbb{N}$
if $\tau(B)=1$. The 
pair $(B,m)$ means that the current goal is to compute the unique path $\xi$ from $B$ to a leaf of $\FSLP$
that corresponds to the vertex with preorder number $m$ in the forest $\derivsub{\FSLP}{B}$.
A triple $(B,m,p)$ means that the current goal is to compute the unique path $\xi$ from $B$ to a leaf of $\FSLP$
that corresponds to the vertex with preorder number $m$ in the forest context $\derivsub{\FSLP}{B}$
under the additional assumption that the unique occurrence of $\ast$ in $\derivsub{\FSLP}{B}$ is replaced by a forest of size $p$.
It is always ensured by the algorithm that $m$ belongs to the range of 
preorder numbers of vertices belonging to the forest (resp., forest context) produced from vertex $B$.

We start with the pair $(B,m) := (A, k)$. Assume now that $(B, \ell, B_1)$ and $(B, r, B_2)$ are the two outgoing edges
of vertex $B$ in $\FSLP$. First assume that $\tau(B)=0$. Hence, $\lambda(B) = \conch$ implies $\tau(B_1) = \tau(B_2) = 0$ and
$\lambda(B) = \concv$ implies $\tau(B_1) = 1$ and $\tau(B_2) = 0$.
The algorithm currently stores a pair $(B,m)$ and updates the data as follows (recall that the preorder
numbers start with $0$):

\medskip
\noindent
Case $\lambda(B) = \conch$ and $m < s(B_1)$: $(B,m)$ is replaced by $(B_1, m)$.

\medskip
\noindent
Case $\lambda(B) = \conch$ and $s(B_1) \leq m$: $(B,m)$ is replaced by $(B_2, m-s(B_1))$.

\medskip
\noindent
Case $\lambda(B) = \concv$ and ($m < \ell(B_1)$ or $\ell(B_1) + s(B_2) \leq m$): $(B,m)$ is replaced by $(B_1,m, s(B_2))$.

\medskip
\noindent
Case $\lambda(B) = \concv$ and $\ell(B_1) \leq m < \ell(B_1) + s(B_2)$: $(B,m)$ is replaced by $(B_2,m-\ell(B_1))$.

\medskip
\noindent
Now assume that $\tau(B)=1$.  Hence, if $\lambda(B) = \concv$ then $\tau(B_1) = \tau(B_2) = 1$ and if 
$\lambda(B) = \conch$ then either $\tau(B_1) = 1$ and $\tau(B_2) = 0$ or $\tau(B_1) = 0$ and $\tau(B_2) = 1$.
The algorithm currently stores a triple $(B,m,p)$ and updates the data as follows:

\medskip
\noindent
Case $\lambda(B) = \conch$, $\tau(B_1) =0$, $\tau(B_2) =1$, and $m < s(B_1)$: $(B,m,p)$ is replaced by $(B_1, m)$.

\medskip
\noindent
Case $\lambda(B) = \conch$, $\tau(B_1) =0$, $\tau(B_2) =1$, and $s(B_1) \leq m$: $(B,m,p)$ is replaced by $(B_2, m-s(B_1),p)$.

\medskip
\noindent
Case $\lambda(B) = \conch$, $\tau(B_1) =1$, $\tau(B_2) =0$, and $m < s(B_1)+p$: $(B,m,p)$ is replaced by $(B_1, m,p)$.

\medskip
\noindent
Case $\lambda(B) = \conch$, $\tau(B_1) =1$, $\tau(B_2) =0$, and $s(B_1)+p \leq m$: $(B,m,p)$ is replaced by $(B_2, m-s(B_1)-p)$.

\medskip
\noindent
Case $\lambda(B) = \concv$ and ($m < \ell(B_1)$ or $\ell(B_1) + s(B_2)+p \leq m$): $(B,m,p)$ is replaced by $(B_1, m, s(B_2)+p)$.

\medskip
\noindent
Case $\lambda(B) = \concv$ and $\ell(B_1) \leq m < \ell(B_1) + s(B_2)+p$: $(B,m,p)$ is replaced by $(B_2, m-\ell(B_1), p)$.

\medskip
\noindent
The algorithm terminates when the first component of the current pair (resp., triple) is a leaf of the DAG $\FSLP$. 
The first components of the pairs (resp., triples) that are produced by the algorithm form exactly the path that corresponds
to the vertex with the initial preorder number $k$.

Note that once the above path $\pi$ is computed, one can obtain
from $\FSLP$ the $\newNodeMarker{V}$-extension $\FSLP'$ and the vertex $\newNodeMarker{A} \in \newNodeMarker{V}$
from the lemma using
$|\pi| \leq \order{A}$ many one-node extensions by adding copies of the vertices from the path $\pi$ bottom-up (starting with the leaf where $\pi$ ends). An example for this construction can be found in \autoref{fig-fslp-update}.

The final statement of the theorem concerning the computation of the enumeration data structure follows directly from Lemma~\ref{lemma-1-ext}.
\end{proof}
\begin{figure}
\begin{center}
\scalebox{1}{
\scalebox{1.4}{\includegraphics{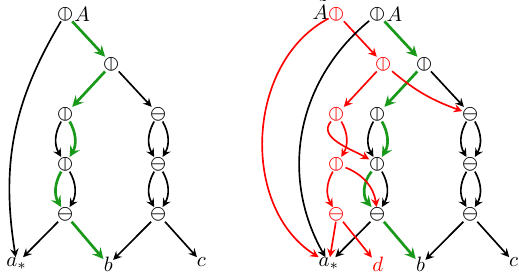}}
         }
\end{center}
\caption{An update in the f-SLP $\FSLP$ from \autoref{fig-fslp}. The vertex identified by the green path on the left is relabelled with the symbol $d$. The red vertices and edges on the right are new.}
	\label{fig-fslp-update}
\end{figure} 
Recall the notion of a rooted f-SLP $\FSLP$ from \autoref{subsubsection:FSLPs}, which defines a single forest $\deriv{\FSLP}$.
We assume in the following that for a rooted f-SLP all vertices $A$ can be reached from the root vertex $S$.
One then defines the height of $\FSLP$ as $\order{\FSLP} = \ordersub{\FSLP}{S}$. We also write $N_{\FSLP}$ for the size
of the produced forest $\deriv{\FSLP}$.
Note that $\ordersub{\FSLP}{A} \leq \order{\FSLP}$ for all vertices of $A$. 
The following balancing theorem for rooted f-SLPs from \cite{GanardiJL21} allows us to 
bound the height of a rooted f-SLP logarithmically in the size of the produced forest:

\begin{theorem}[cf.~\mbox{\cite[Corollary 3.28]{GanardiJL21}}] \label{thm-balance-f-SLP}
Given a rooted f-SLP $\FSLP$,
one can compute in time $\bigO(\card{\FSLP})$ a rooted f-SLP 
$\FSLP'$ such that $\deriv{\FSLP'} = \deriv{\FSLP}$, $|\FSLP'| = \Theta(\card{\FSLP})$ 
and $\order{\FSLP'} = \Theta(\log N_{\FSLP})$.
\end{theorem}
The following variant of \autoref{relabellingTheorem} for 
rooted f-SLPs is shown in the same way as \autoref{relabellingTheorem}. 
By \autoref{thm-balance-f-SLP}, the height bound $\Theta(\log N_{\FSLP})$ can be assumed without loss of generality.

\begin{theorem}\label{relabellingTheorem-rooted}
Assume that the following is given:
\begin{itemize}
\item a rooted f-SLP $\FSLP = (V,E,\lambda)$ with $\order{\FSLP} = \Theta(\log N_{\FSLP})$
 together with the values $\tau(B)$, $s(B)$ and $\ell(B)$ (the latter only in case $\tau(B)=1$) for all $B \in V$,
\item the preorder number $k$ of a vertex $v$ from $\deriv{\FSLP}$ and 
\item a symbol $a \in \Sigma$.
\end{itemize}
One can then compute in time $\bigO(\log N_{\FSLP})$ a rooted $\newNodeMarker{V}$-extension $\FSLP'$ of $\FSLP$ such that
\begin{itemize}
\item $\deriv{\FSLP'}=\mathsf{relabel}(\deriv{\FSLP},v,a)$,
\item $\order{\FSLP'} = \order{\FSLP} = \Theta(\log N_{\FSLP})$ and 
\item $|\newNodeMarker{V}| = \bigO(\log N_{\FSLP})$.
\end{itemize}
If in addition the enumeration data structure for $\FSLP$ and some dBUTA $\mathcal{B}$ with state set $Q$ is given then one can also compute the enumeration data structure for $\FSLP'$ and $\mathcal{B}$
in time $\bigO(\card{Q}^2 \cdot \log N_{\FSLP})$.
\end{theorem}
Since  $\order{\FSLP'} = \order{\FSLP}$ and $N_{\FSLP} = N_{\FSLP'}$ (the latter is trivial since a relabelling does not change the set of vertices
of a forest), \autoref{thm-balance-f-SLP} has to be only applied once in the beginning. Successive relabelling updates do not destroy
the balancedness.

Let us stress that the f-SLP $\FSLP'$ from \autoref{thm-balance-f-SLP} does  in general not
satisfy the balancedness property $\ordersub{\FSLP'}{A} = \Theta(\log N_A)$ (where $N_A$ is the size of the forest $\derivsub{\FSLP}{A}$)
for all forest vertices $A$ of $\FSLP'$.
This is only guaranteed for the root vertex of $\FSLP'$.
For s-SLPs this stronger balancedness property 
is in fact achievable, see \cite{Ganardi2021}. We conjecture that the construction from \cite{Ganardi2021} can be extended
 to f-SLPs. If this is true then one could replace after the appropriate balancing the height $\order{A}$ in 
\autoref{relabellingTheorem} by $\bigO(\log N_A)$.

Notice that the relabelling procedure from the proof of \autoref{relabellingTheorem}
makes the f-SLP always larger, even if the relabelled forest might be better compressible
(with respect to f-SLPs) than the original forest. It is not clear how to avoid this. In this context one might consider the following
decision problem: given a rooted f-SLP $\FSLP$, 
a vertex $v$ in $\deriv{\FSLP}$ (given by its preorder number) and a symbol
$a \in \Sigma$, is there a rooted f-SLP for $\mathsf{relabel}(\deriv{\FSLP},v,a)$ of size 
at most $|\FSLP|$? It is not clear, whether this problem can be solved in polynomial time.
The question whether a given string has an s-SLP of size at most a given number is 
already NP-complete \cite{CaselEtAl2021,CharikarLLPPSS05}.

\subsection{Beyond relabelling updates} \label{sec other updates}

While relabelling a single vertex of an unranked forest is a natural update operation, it is also quite simple. In particular, it does not change the overall structure of the forest.

More complex tree updates were considered in \cite{MMN22} in the context of query enumeration on uncompressed trees. In \cite{MMN22}, the authors consider in addition to relabeling updates also certain insertion and deletion updates, where vertices are inserted (resp., deleted) at certain specified positions in the current forest. The authors in \cite{MMN22} achieve time $\bigO(\log |F|)$ for these updates. For this, they represent the input forest $F$ by a forest algebra expression $\expr$ of height $\bigO(\log |F|)$. 

It is tempting to extend the approach from \cite{MMN22} to an f-SLP $\FSLP$, i.e., to a forest algebra expression $\expr$ that is represented by a DAG. We conjecture that this might be possible, but the technical difficulties are considerable. The main idea of~\cite{MMN22} is to maintain a certain balancing property of the forest algebra expression $\expr$ that describes the tree $T$. With this balancing property, the update time depends on the height of $\expr$, which is logarithmic in the size of $T$. The approach for keeping $\expr$ balanced is to perform certain rotations similar to the classical ones that are used for balanced search trees like AVL-trees. However, in order to do such rotations with respect to the syntax tree of an algebraic expression without changing the evaluation of the expression, one needs associativity of the respective algebra operations. For the two binary forest algebra operations, this is not the case: $F_1 \concv (F_2 \conch F_3)$ is in general not the same as $(F_1 \concv F_2) \conch F_3$. In \cite{MMN22} the authors found a quite technical workaround for this problem in the uncompressed setting. It is not obvious that the same workaround can be also used in the situation, where the forest algebra expression is represented by a DAG.

\subsection{A lower bound}\label{sec:updateLowerBound}

In this section we prove a lower bound on the size increase of relabelling updates. Notice that relabelling updates increase the size of the f-SLP by an additive term that is bounded by the height of the f-SLP. We show that this additive term can be only improved by a multiplicative $\log\log$-factor. 

We prove our lower bound for strings and s-SLPs. Since every string can be seen as a particular forest (see \autoref{forests})
and an s-SLP can be seen as an f-SLP (see \autoref{subsubsection:FSLPs}), our lower bound also holds for forests and f-SLPs.

For a string $S$ over a finite alphabet $\Sigma$, a position $1 \le i \le |S|$ and $x \in \Sigma$, we denote by $\mathsf{relabel}(S,i,x)$ the string obtained by relabelling the $i^{\text{th}}$ symbol of $S$ into $x$. For a string $S$ we write $g(S)$ for the size of a smallest s-SLP for $S$.

\begin{theorem} \label{thm-update-size}
There is a family of strings $(S_k)_{k \geq 1}$ over the alphabet $\{a,b\}$ of strictly increasing length 
and for every $k \geq 1$ there is an $i_k \le |S_k|$ such that
\[ g(\mathsf{relabel}(S_k,i_k,b)) - g(S_k) = \Omega\bigg(\frac{\log |S_k|}{\log \log |S_k|}\bigg).\]
\end{theorem}

\begin{proof}
We start with a string $w_k \in \{0,1\}^*$ of length $k^2$ that is algorithmically incompressible in the sense
of Kolmogorov complexity.
 It is well-known that such a word exists; see, e.g.,~\cite{LiV08}.\footnote{For our purpose it is not necessary
 to define the concept of Kolmogorov complexity. That a word
  $x \in \{0,1\}^*$ is algorithmically incompressible means the following: for every partial computable function
 $f : \{0,1\}^* \to \{0,1\}^*$ there exists a constant $c_f$ (that only depends on $f$) such that every $y \in f^{-1}(x)$ satisfies
 $|y| \geq |x| - c_f$. Intuitively speaking, $y$ is a description of $x$ and $f$ is a decoding function that produces from a 
 description $y$ the word $x = f(x)$.} 
 Let us write $w_k = w_{k,1} w_{k,2} \cdots w_{k,k}$ with $|w_{k,i}| = k$ for all $1 \leq i \leq k$. Let $\mu_k : \{0,1\}^k \to [2^k]$ be the function such that every string $s \in \{0,1\}^k$ is the $\mu_k(s)^{\text{th}}$ word in the lexicographic enumeration of all strings from $\{0,1\}^k$. Let $m_{k,i} = \mu_k(w_{k,i}) \leq 2^k$. Consider now the following two strings:
\begin{align*}
&v_k =& &a^{2^k}& &a& &a^{m_{k,1}}& &a& &a^{m_{k,2}}& &a& &\cdots& &a^{m_{k,k-1}}& &a& &a^{m_{k,k}}\,,& \\
&v'_k =& &a^{2^k}& &b& &a^{m_{k,1}}& &b& &a^{m_{k,2}}& &b& &\cdots& &a^{m_{k,k-1}}& &b& &a^{m_{k,k}}\,.&
\end{align*}
Note that $2^k \leq |v_k| = |v'_k| \leq k + (k+1)2^k$. Since $v_k$ is a unary string, we have $g(v_k) = \Theta(\log |v_k|) = \Theta(k)$
(this is a folklore fact; see, e.g.,~\cite[Lemma~2]{CharikarLLPPSS05} for a more general statement).

Let us estimate $n_k := g(v'_k)$. Since $v'_k$ has an s-SLP of size $n_k$, one can encode 
$v'_k$ by a bit string of length $\bigO(n_k \cdot \log n_k)$; see, e.g.,~\cite{TabeiTS13}. Since $v'_k$ encodes the 
algorithmically incompressible word $w_k$
we must have $k^2 = |w_k| = \bigO(n_k \cdot \log n_k)$, i.e., $n_k \cdot \log n_k = \Omega(k^2)$.
In addition, we can easily construct an s-SLP for $v'_k$ of size $\bigO(k^2)$ ($a^{2^k}$ and all $a^{m_{k,i}} $ have s-SLPs of size $\bigO(k)$).
Hence, $n_k = \bigO(k^2)$. We thus obtain 
\[ g(v'_k) = n_k = \Omega\bigg(\frac{k^2}{\log n_k}\bigg) = \Omega\bigg(\frac{k^2}{\log k}\bigg) .\]
This means that the size difference $g(v'_k) - g(v_k)$ is $\Omega(\frac{k^2}{\log k} - k)$.

Notice that $v'_k$ is obtained from $v_k$ by $k$ relabelling operations that change occurrences of $a$ into occurrences of $b$. We can conclude that at least
one of those $k$ relabellings must increase the minimal s-SLP size by at least 
\begin{equation}\label{eq-lower bound}
\Omega\bigg( \frac{\frac{k^2}{\log k} - k}{k} \bigg) = \Omega\bigg( \frac{k}{\log k} \bigg)\,. 
\end{equation}
We define $S_k$ as the word right before this relabelling. It satisfies $2^k \leq |S_k| \leq k + (k+1)2^k$
so that $k$ can be replaced by $\log |S_k|$ in \eqref{eq-lower bound}. This proves the theorem.
\end{proof}
Note that by \autoref{thm-balance-f-SLP} for every s-SLP $\SLP$ producing a string $S$ one can 
reduce the height of $\SLP$ to $\Theta(\log\card{S})$; thereby the size of the s-SLP increases only by a
constant factor. Of course, this statement also applies to a smallest s-SLP for $S$. As a consequence, one
obtains from \autoref{thm-update-size} the lower bound $\Omega(\order{\SLP}/\log\order{\SLP})$ for the size increase when applying a relabelling update to an s-SLP $\SLP$. Clearly, this size increase also gives a lower bound for the running time of a relabelling update.

\section{More Background on Straight-Line Programs}\label{sec:SLPsBackground}

Since this is the central concept of our work, let us provide more background on straight-line programs. A rather important motivation for our work is that compressing a given string or forest by an SLP is a problem that can be solved rather well in a practical context (i.e., we can compute SLPs with excellent compression ratios in linear time). Due to the relevance of this aspect, we shall discuss it in more detail in \autoref{sec:practicalSLPAlgos}.

String SLPs (s-SLPs)  date back several decades; see, e.g.,~\cite{Nev96,StoSzy82}. Nowadays, they are very popular and play a prominent role in the context of string algorithms and other areas of theoretical computer science. They are mathematically easy to handle and therefore very appealing for theoretical considerations. Independent of their applications in data compression, they have been used in many different contexts as a natural tool for representing (and reasoning about) hierarchical structure in sequential data; see, e.g.,~\cite{KieYan2000,Loh2014,Loh12survey,NevWit97a,Nev96,StoSzy82}.
Good sources for further details on s-SLPs are the survey~\cite{Loh12survey}, the PhD-thesis~\cite{Cording2015PhD} and the comprehensive introductions of the papers~\cite{AbboudEtAl2017,CaselEtAl2021}.

String SLPs are also of high practical relevance, mainly because many practically applied dictionary-based compression schemes (e.g., run-length encoding, and --~most notably~-- the various Lempel-Ziv variants LZ77, LZ78, LZW, etc. which are relevant for practical tools like the built-in Unix utility \textsf{compress} or data formats like GIF, PNG, PDF and some ZIP archive file formats) can be converted efficiently into s-SLPs of similar size, i.e., with size blow-ups by only moderate constants or log-factors (see~\cite{AbboudEtAl2017,Cording2015PhD,GotoEtAl2011,Loh12survey,Ryt03}). Hence, algorithms for SLP-compressed strings carry over to these practical formats.

A possible drawback of s-SLPs is that computing a minimal size s-SLP for a given string is intractable (even for fixed alphabets)~\cite{CaselEtAl2021,CharikarLLPPSS05}. However, this has never been an issue for the application of s-SLPs, since many 
heuristical SLP-compressors achieve very good compression rates for practical inputs. This aspect will be discussed in more detail and tailored to our results in \autoref{sec:practicalSLPAlgos}.

While in the early days of computer science fast compression and decompression was an important factor, it is nowadays common to also rate compression schemes according to how suitable they are for solving problems directly on the compressed data without prior decompression, a paradigm that is known as algorithmics on compressed strings. In this regard, s-SLPs have very good properties: Many basic problems on strings like comparison, pattern matching, membership in a regular language, retrieving subwords, etc.~can all be efficiently solved 
directly on s-SLPs~\cite{Loh12survey}.

String SLPs are usually defined in terms of context-free grammars, i.e., an s-SLP for a word $w$ is a context-free grammar in Chomsky normal form for the language $\{w\}$. It is straightforward to see that our definition from \autoref{sec:sSLPs} is equivalent: 
We can interpret every vertex $A$ of the DAG that represents the s-SLP as a non-terminal symbol with a context-free rule $A \to BC$, if it is an inner vertex with left edge $(A,\ell,B)$ and right edge $(A,r,C))$, or with a context-free rule $A \to a$ if it is an $a$-labelled leaf. 
This context-free grammar based definition has the advantage that it can be easily extended to trees, by simply using a context-free grammar formalism for trees, which leads to tree SLPs~\cite{Lohrey15,LohreyMR18,LohreyEtAl2012}. Context-free tree grammars have rules of the form
$A \to T$, where $A$ labels a vertex $v$ with children $u_1, u_2, \ldots, u_k$ and $T$ is a tree that, among others, has distinguished leaves $x_1, x_2, \ldots, x_k$. The idea is that $v$ is replaced by $T$ in such a way that the subtree rooted in $u_i$ is plugged in at the position of the leaf $x_i$ of $T$. However, for such a formalism $A$-labelled vertices must have a fixed rank $k$. Thus, such tree SLPs can only compress ranked trees. 
This is a disadvantage, since in the context of database theory, we are rather interested in unranked trees and forests as data model. A typical example of such data are XML tree structures or decision trees. Therefore, we use forest SLPs, which were introduced in~\cite{GasconLMRS20} in a more grammar-like way that is nevertheless equivalent to our approach in \autoref{sec:FSLP}. 

These f-SLPs have many desirable properties, which make them a suitable compression scheme for our setting. Most importantly, they can compress vertex-labelled unranked forests, which cover a rather large class of tree structures (e.g., XML tree structures, decision trees, tree decompositions); in particular, they are not limited to ranked trees, which would be too restrictive for typical applications in data management. At the same time, f-SLPs share most of the desirable properties of s-SLPs, e.g., they are mathematically easy to handle and can achieve exponential compression rates. Moreover, f-SLPs are robust in the sense that they also cover other popular tree compression schemes like top dags \cite{BilleGLW15,DudekG18,Hubschle-Schneider15} and tree straight-line programs \cite{GanardiHJLN17,LohreyMM13}.

It is also possible to compute small f-SLPs for given forests in acceptable running times, mainly because compression techniques for s-SLPs can be adapted to the case of trees and forests.\footnote{Observe that the general intractability of computing a \emph{smallest} f-SLP obviously carries over from the string case. However, just like in the string case, this is not an obstacle for practically relevant approximations and heuristics.} Since this aspect is very important for our results, we will discuss it in more detail in \autoref{sec:practicalSLPAlgos}. 

It is also known that for every forest with $n$ vertices and $k$ different vertex labels,
one can construct in linear time an f-SLP of size $\bigO(n \log k/ \log n)$ (so $\bigO(n/ \log n)$ for a fixed $k$)~\cite{GanardiHJLN17}. Finally, a recent balancing result for s-SLPs~\cite{GanardiJL21} also holds for f-SLPs, which we applied in the context of updates  (see \autoref{thm-balance-f-SLP}).

\subsection{Practical algorithms for SLP-compression of strings and forests}\label{sec:practicalSLPAlgos}

Our whole work hinges on the assumption that we get our input data in SLP-compressed form, and that these SLPs are substantially smaller than the uncompressed data. Let us discuss now in a bit more detail why this assumption is justified for both strings and forests.

The problem of computing a smallest s-SLP for a given string $S$ cannot be solved in polynomial time unless P=NP
\cite{CaselEtAl2021,CharikarLLPPSS05}. An algorithm for computing a smallest s-SLP that runs in time $\bigO(3^{\card{S}})$ is presented in~\cite{CaselEtAl2021}. However, there is a large number of algorithms that compute small s-SLPs in linear time or low-degree polynomial time. For example, there exist several algorithms that compute for a given string $S$ of length $n$ in time $\bigO(n)$ an s-SLP
of size $\bigO(g \cdot \log n)$, where $g$ is the size of a smallest s-SLP for $S$ \cite{CharikarLLPPSS05,Jez2015,Ryt03}.
$\bigO(\log n)$ is currently the best known approximation ratio of polynomial time grammar-based compressors.
Upper and lower bounds for the approximation ratios of several practical grammar-based compressors are studied in
 \cite{BannaiEtAl2021,CharikarLLPPSS05}.  Some of these compressors (in particular RePair \cite{LarssonMoffat2000})  are known to perform very well in practical scenarios. Finally, it is also known that strings represented by many practical compression schemes (e.g., run length encoding, LZ77 and LZ78 encoding) can be transformed into s-SLPs with moderate size increase \cite{GotoEtAl2011}. In summary, the problem of compressing a string by an s-SLP is theoretically well-understood and a rich toolbox of practical methods exists. 

While the state-of-the-art for f-SLPs is not as developed as for s-SLPs, we can observe that compression by f-SLP can also be handled by existing algorithmic techniques. Most importantly, the above mentioned grammar-based string compressor RePair
can be adapted so that it computes an f-SLP for a given unranked forest. More precisely, the so-called TreeRePair algorithm~\cite{LohreyMM13} computes a tree SLP for a \emph{ranked} tree, but it can also be used on the first-child-next-sibling encoding of an unranked forest $F$. The resulting tree SLP for the first-child-next-sibling encoding of $F$ can then be transformed in linear time into an equivalent f-SLP for $F$ (see~\cite{GasconLMRS20}). This approach shows excellent compression ratios in practice, which is also demonstrated by an experimental study: For a corpus of typical XML documents, the number of edges of the original tree is reduced to approximately 3\% using TreeRePair on the first-child next-sibling encoding of the XML tree~\cite{LohreyMM13}. Other available grammar-based tree compressors are BPLEX~\cite{BusattoLM08} and CluX~\cite{BottcherHK10}. 

For maintaining relabelling updates, it is advantageous if our f-SLPs are balanced (i.e., the height is logarithmic in the size of the decompressed forest), since then the update procedure has a running time that is bounded logarithmically in the size of the data. Due to \autoref{thm-balance-f-SLP} from \cite[Corollary 3.28]{GanardiJL21}, this can be achieved with a linear time preprocessing. Moreover, relabelling updates do not change the size of the decompressed forest or the height of the f-SLP, which implies that the f-SLP stays balanced after a relabelling update.

\section{Conclusions}\label{sec:conclusions}

We remark that the special case of \autoref{thm-enumerate-paths} where the category $\mathcal{C}$ is a groupoid (a category where
all morphisms are invertible)
can also be proven by using a known technique for the real-time traversal of SLP-compressed strings (see \cite{GasieniecKPS05,LohreyMR18}). However, our category from \autoref{sec:preorderNumbers} is not a groupoid.
Another disadvantage is that the real-time traversal of SLP-compressed strings needs a tree data structure for so-called next link queries. While such data structures can be constructed in linear time, this is not straightforward and would significantly complicate an implementation of our algorithm. In general, we believe that our approach is simple to implement, which makes an experimental analysis in the vein of~\cite{LohreyMM13} possible.

An important open question is whether our enumeration algorithm for f-SLP-compressed unranked forests can also be extended by insertion and deletion updates in logarithmic time instead of only relabellings. We conjecture that this is indeed possible, but respective constructions will be technically rather involved.

SLPs have also been formulated for graphs (see~\cite{ManethP18}). 
In a recent paper \cite{LMS25} we showed that the result sets of queries formulated in 
first-order logic can be enumerated with linear preprocessing and constant delay on SLP-compressed graphs
of bounded degree (with a technical restriction on the graph SLPs).
This extends previous work for uncompressed graphs of bounded degree \cite{DurandGrandjean2007,KazSeg11}.

For automaton-based queries on strings it has been recently shown that 
enumeration algorithms can directly deal with \emph{nondeterministic} automata and therefore avoid exponential preprocessing in combined complexity, see~\cite{AmarilliEtAl2021}. Moreover, 
enumeration algorithms for the weighted case, where the results are to be enumerated sorted by their weight, have been developed for strings, see~\cite{BourhisEtAl2021,GawrychowskiEtAl2024}. Both these aspects are also worth investigating in the context of MSO-enumeration over SLP-compressed forests.

\printbibliography

\newpage

\appendix

\section*{Appendix: Proof of \autoref{thm-MNN}}\label{sec:proofOfthm-MNN}

\begin{proof}
The following construction is from \cite{MMN22} (see the definition before \cite[Lemma~4.3]{MMN22}).
We fix the nSTA $\mathcal{A} = (Q, \Sigma, \delta, \iota, q_0, q_f)$. 
 Recall that  $\Sigma_0 = \{ a,a_\ast : a \in \Sigma\}$ and $\Sigma_2 = \{\conch, \concv \}$. 
 Let $\mathcal{P} = 2^{Q^2} \cup 2^{Q^4} \cup \{\mathsf{failure}\}$
 be the set of states of the dBUTA  $\mathcal{B}$.
 The transition mappings $\delta_0 : \Sigma_0 \to \mathcal{P}$  and $\delta_2 : \mathcal{P} \times \mathcal{P} \times \Sigma_2 \to \mathcal{P}$ of $\mathcal{B}$ are defined as follows, where $P_1, P_2 \subseteq Q^2$ and $Q_1, Q_2 \subseteq Q^4$:
  \begin{align*}
\delta_0(a) = \, & \{ \langle p_1, p_2 \rangle \in Q^2 : \exists q \in \iota(a) : (p_1, q, p_2) \in \delta \} \\
\delta_0(a_\ast) = & \{ \langle p_1, p_2, p_3 ,p_4 \rangle \in Q^4 : p_3 \in \iota(a), (p_1, p_4, p_2) \in  \delta  \} \\
\delta_2( P_1, P_2, \conch)  = \,& \{ \langle p_1, p_3\rangle  : \exists p_2 \in Q : \langle p_1, p_2\rangle \in P_1, \langle p_2, p_3\rangle \in P_2 \} \\
\delta_2( P_1, Q_2, \conch) = \, & \{ \langle p_1, p_3, q_1, q_2 \rangle  : \exists p_2 \in Q : \langle p_1, p_2\rangle \in P_1, \; \langle p_2, p_3, q_1, q_2\rangle \in Q_2 \} \\
\delta_2( Q_1, P_2, \conch)  = \, & \{ \langle p_1, p_3, q_1, q_2 \rangle  : \exists p_2 \in Q : \langle p_1, p_2, q_1, q_2 \rangle \in Q_1, \;  \langle p_2, p_3\rangle \in P_2 \} \\
\delta_2( Q_1, P_2, \concv)  = \, & \{ \langle p_1, p_2 \rangle  : \exists q_1, q_2 \in Q : \langle p_1, p_2, q_1, q_2 \rangle \in Q_1, \; \langle q_1, q_2 \rangle \in P_2 \} \\
\delta_2( Q_1, Q_2, \concv)  = \, & \{ \langle p_1, p_2, p_5, p_6 \rangle  : \exists p_3, p_4 \in Q : \langle p_1, p_2, p_3, p_4 \rangle \in Q_1, \; \langle p_3, p_4, p_5, p_6 \rangle \in Q_2 \}
\end{align*}
In all cases, where $\delta_0$ and $\delta_2$ are not defined by the above rules, we return {\sf failure}.
One can show that for every $\expr \in \EXP(\Sigma)$ with $F = \valX{\expr} \in \F(\Sigma)$
and all states $p,q \in Q$ we have: there is a $(p,q)$-run of $\mathcal{A}$ on $F$ if and only if $(p,q) \in \mathcal{B}(\expr)$.
This is the content of \cite[Lemma 4.5]{MMN22} (for forests). Hence, we can take $(q_0, q_f)$ as the unique final state of $\mathcal{B}$.
\end{proof}

\end{document}